\documentclass[12pt,a4paper]{article}

\usepackage{amsfonts, mathrsfs, amsmath, amssymb, amsthm, color, 
enumitem, float, graphicx, lipsum, longtable, lscape, mathrsfs, mathtools, multirow, nccmath, pdflscape, rotating, tocloft}
\usepackage{algorithm}
\usepackage[top=1.4in, bottom=1.4in, left=1.35in, right=1.35in]{geometry}
\usepackage{natbib}
\usepackage{url}
\usepackage[pdftex,colorlinks=true,hypertexnames=false]{hyperref}
\usepackage[doublespacing]{setspace}
\usepackage[table]{xcolor}
\usepackage{tikz}
\usetikzlibrary{arrows,decorations.pathmorphing,backgrounds,positioning,fit,matrix}
\usepackage{subcaption}
\usepackage{palatino,eulervm}
\usepackage[T1]{fontenc}
\usepackage{bm}
\usepackage{booktabs}
\usepackage{caption}
\usepackage{algorithm, algorithmic}%
\usepackage[normalem]{ulem}
\usepackage{cancel, soul}

\usepackage{xr-hyper}
\definecolor{darkblue}{rgb}{0,0,.6}
\hypersetup{citecolor=darkblue,linkcolor=darkblue,urlcolor=darkblue}
\def\be{\begin{equation}}
\def\ee{\end{equation}}
\def\bea{\begin{eqnarray}}
\def\eea{\end{eqnarray}}
\def\mE{{\sf E}}

\providecommand{\U}[1]{\protect\rule{.1in}{.1in}}

\DeclareMathOperator*{\argmin}{\arg\!\min}
\DeclareMathOperator*{\argmax}{\arg\!\max}

\DeclareMathOperator*{\diag}{\normalfont\textrm{diag}}

\newcommand{\pr}[1]{\mathbb{P} \left( #1 \right)}

\newcommand{\var}[1]{\mathrm{Var} \left( #1 \right)}
\newcommand{\cov}[2]{\mathrm{Cov} \left( #1, \, #2 \right)}

\usepackage{caption}
\usepackage{float}
\usepackage{algorithm}
\usepackage{algorithmic}


\begin{document}

\newtheorem{corollary}{Corollary}
\newtheorem{definition}{Definition}
\newtheorem{lemma}{Lemma}
\newtheorem{proposition}{Proposition}
\newtheorem{remark}{Remark}
\newtheorem{theorem}{Theorem}
\newtheorem{example}{Example}
\newtheorem{assumption}{Assumption}
\newtheorem{prop}{Proposition}

\numberwithin{corollary}{section}
\numberwithin{definition}{section}
\numberwithin{equation}{section}
\numberwithin{lemma}{section}
\numberwithin{proposition}{section}
\numberwithin{remark}{section}
\numberwithin{theorem}{section}

\onehalfspacing

\newtheorem{Proof}{Proof}
\newtheorem{Lemma}{Lemma}
\newtheorem{Mth}{Main Theorem}
\newtheorem{Res}{\underline{\bf Result}}
\newtheorem{Def}{Definition}
\newtheorem{Rem}{\underline{\bf Remark}}
\newtheorem{Qes}{Question}
\newtheorem{Aim}{Aim}
\newtheorem{Pro}{Proposition}
\newtheorem{Lem}{\underline{\bf Lemma}}
\newtheorem{Cor}{\underline{\bf Corollary}}
\newtheorem{Ex}{Example}
\newtheorem{Eq}{Equation}
\newtheorem{condition}{Assumption}
\renewcommand{\thecondition}{\arabic{condition}}
\newtheorem{conditionnew}{Condition}
\newcommand{\calW}{{\mathscr W}}
\newcommand{\calD}{{\mathscr D}}
\def\lm{\lambda_{\tiny\max}}
\def\lmin{\lambda_{\tiny\min}}
\def\bse{\begin{eqnarray*}}
	\def\ese{\end{eqnarray*}}
\def\be{\begin{eqnarray}}
	\def\ee{\end{eqnarray}}
\def\bsq{\begin{equation*}}
	\def\esq{\end{equation*}}
\def\bq{\begin{equation}}
	\def\eq{\end{equation}}

\def\blue{\color{blue}}
\def\top{{^\intercal}}
\def\Y{{\bf Y}}
\def\X{{\bf X}}
\def\Z{{\bf Z}}
\def\w{{\bf w}}
\def\bb{\boldsymbol{\beta}}
\def\ba{\boldsymbol {\alpha}}
\def\bg{{\boldsymbol\gamma}}
\def\btheta{{\boldsymbol\theta}}
\def\beps{{\boldsymbol\epsilon}}
\def\bmu{{\boldsymbol\mu}}
\def\bPi{{\boldsymbol \Pi}}
\def\m{{{(m)}}}
\def\sQ{{\sf Q}}
\def\sN{{\sf N}}
\def\sT{{\sf T}}
\def\sK{{\sf K}}
\def\sP{{\sf P}}
\def\sM{{\sf M}}
\def\supw{\sup_{\w\in{\mathscr W}}}
\def\argmin{\mbox{argmin}}
\def\argmax{\mbox{argmax}}
\def\wh{\widehat}
\def\wt{\widetilde}
\def\n{\nonumber}
\def\calS{\mbox{$\mathcal{S}$}}

\def\bveps{{\boldsymbol\varepsilon}}
\def\hb{\widehat{\bb}}
\def\he{\widehat{\varepsilon}}
\def\cs{\calS_{Y\mid\x}}
\def\cms{\calS_{E(Y\mid\x)}}
\def\defby{\stackrel{\mbox{\textrm{\tiny def}}}{=}}
\def\0{{\bf 0}}
\def\A{{\bf A}}
\def\cA{\mathcal{A}}
\def\a{{\bf a}}
\def\B{{\bf B}}
\def\c{{\bf c}}
\def\D{{\bf D}}
\def\FV{\text{FV}}
\def\e{{\bf e}}
\def\H{{\bf H}}
\def\V{{\bf V}}
\def\g{{\bf g}}
\def\r{{\bf r}}
\def\f{{\bf f}}
\def\l{{\bf l}}
\def\L{{\bf L}}
\def\h{{\bf h}}
\def\b{{\bf b}}
\def\bM{{\bf M}}
\def\M{\mbox{ $\mathcal{M}$}}
\def\BB{\mbox{ $\mathcal{B}$}}
\def\N{\mbox{ $\mathcal{N}$}}
\def\K{{\bf K}}
\def\t{{\bf t}}
\def\T{{\bf T}}
\def\G{{\bf G}}
\def\bP{{\bf P}}
\def\bQ{{\bf Q}}
\def\bV{{\bf V}}
\def\hQ{{\widehat \bQ}}
\def\U{{\bf U}}
\def\S{{\bf S}}
\def\u{{\bf u}}
\def\v{{\bf v}}
\def\W{{\bf W}}
\def\bO{{\bf O}}
\def\w{{\bf w}}
\def\X{{\bf X}}
\def\x{{\bf x}}
\def\eps{\epsilon}
\def\I{{\bf I}}
\def\J{{\bf J}}
\def\tx{{\widetilde \x}}

\def\Ybar{{\overline{Y}}}
\def\xbar{{\overline{\x}}}
\def\Wbar{{\overline{\W}}}
\def\wbar{{\overline{\w}}}
\def\bphi{{\boldsymbol \phi}}
\def\bPhi{{\bf \Phi}}
\def\bUps{{\bf \Upsilon}}
\def\bSig{{\bf \Sigma}}
\def\bOme{{\bf \Omega}}
\def\bDel{{\bf \Delta}}
\def\bdel{{\boldsymbol \delta}}
\def\bbeta{{\boldsymbol \beta}}
\def\bphi{{\boldsymbol \phi}}
\def\bPhi{{\bf \Phi}}
\def\bpsi{{\boldsymbol \psi}}
\def\bPsi{{\bf \Psi}}
\def\tDel{{\wt{{\bf \Delta}}}}
\def\tdel{{\wt{{\boldsymbol \delta}}}}
\def\tphi{{\wt{{\boldsymbol \phi}}}}
\def\tPhi{{\wt{{\bf \Phi}}}}

\def\bLam{{\bf \Lambda}}
\def\diag{\hbox{diag}}
\def\bq{\begin{equation}}
	\def\eq{\end{equation}}
\def\pr{\hbox{pr}}
\def\wt{\widetilde}
\def\diag{\hbox{diag}}
\def\overp{\stackrel{p}\longrightarrow}
\def\log{\hbox{log}}
\def\bias{\hbox{bias}}
\def\Siuu{\boldSigma_{i,uu}}
\def\squarebox#1{\hbox to #1{\hfill\vbox to #1{\vfill}}}

\def\balpha{{\boldsymbol \alpha}}
\def\boeta{{\boldsymbol \eta}}
\def\bpi{{\boldsymbol \pi}}

\def\bx{{\bf x}}
\def\F{{\bf F}}
\def\vec{\mathrm{vec}}
\def\mA{\mathcal{A}}
\def\mB{\mathcal{B}}
\def\mC{\mathcal{C}}
\def\mH{\mathcal{H}}
\def\my{\mathcal Y}
\def\cov{\hbox{cov}}
\def\corr{\hbox{corr}}
\def\trace{\hbox{trace}}
\def\bse{\begin{eqnarray*}}
	\def\ese{\end{eqnarray*}}
\def\be{\begin{eqnarray}}
	\def\ee{\end{eqnarray}}
\def\bsq{\begin{equation*}}
	\def\esq{\end{equation*}}
\def\bq{\begin{equation}}
	\def\eq{\end{equation}}
\def\pr{\hbox{pr}}
\def\wt{\widetilde}
\def\diag{\hbox{diag}}
\def\log{\hbox{log}}
\def\bias{\hbox{bias}}
\def\Siuu{\boldSigma_{i,uu}}
\def\whT{\widehat{\Theta}}
\def\diag{\hbox{diag}}
\def\s{{(s)}}
\def\m{{(m)}}
\def\sums{\sum\nolimits_{s=1}^{{S}}}
\newcommand{\calR}{{\cal R}}
\def\Normal{\hbox{Normal}}
\def\blue{\color{blue}}
\def\red{\color{red}}
\def\th{{\mbox{\tiny\rm th}}}
\def\MMA{\text{MMA}}
\def\CV{\text{CV}}
\def\FVL{\text{FVL}}
\def\FV{\text{FV}}
\def\CM{\text{CM}}
\def\bXi{{{{\boldsymbol \Xi}}}}

\def\maxi{\mathop{\max}\limits_{i}}
\def\mini{\mathop{\min}\limits_{i}}
\def\maxj{\mathop{\max}\limits_{j}}

\def\maxs{\mathop{\max}\limits_{s}}

\def\sumi{\mathop{\sum}\nolimits_{i=1}^n}
\def\sumj{\mathop{\sum}\nolimits_{j=1}^n}
\def\full{{\mbox{\tiny\rm full}}}

\def\var{\hbox{var}}


\def\cp{\stackrel{P}\rightarrow}
\def\bth{\boldsymbol{\theta}}

\allowdisplaybreaks[4]

\title{Time-Varying Model Averaging of Multi-layer Network Vector Autoregressions}
\author{{\normalsize Degui Li}\thanks{Faculty of Business Administration, Asia-Pacific Academy of Economics and Management and Department of Economics, University of Macau. Li's research is supported by the National Natural Science Foundation of China (72525015) and the CPG and SRG grants funded by the University of Macau.},\ \  {\normalsize Yuying Sun}\thanks{Academy of Mathematics and Systems Science, Chinese Academy of Sciences. Sun's research is supported by the National Natural Science Foundation of China (72322016, 72073126, 72091212) and Young Elite Scientists Sponsorship Program by CAST (2020QNRC001). },\ \ {\normalsize Boyao Wu}\thanks{China School of Banking and Finance, University of International Business and Economics. Wu's research is supported by the National Natural Science Foundation of China (72403048 and 72271055)}\\
{\small\em University of Macau,\ \ Chinese Academy of Sciences,\ \ University of International Business and Economics}}
\date{{\small Version: \today}}
\maketitle
\thispagestyle{empty}



\begin{abstract}

In this paper, we introduce a flexible time-varying multi-layer network vector autoregression (VAR) model framework for large-scale time series, allowing agents in dynamic systems to interact through multiple channels and incorporating multiple adjacency matrices to capture network spillover effects. We propose a penalized model averaging method to determine a time-varying optimal combination of multi-layer network VAR candidate models whose number may be divergent. Under some regularity conditions, the asymptotic properties such as asymptotic optimality and convergence rates of the proposed time-varying weight estimation are derived in the contexts of both the in-sample fitting and out-of-sample prediction. In addition, we extend the conformal prediction method to construct prediction bands for locally stationary time series. Monte-Carlo simulation studies and an empirical application to forecast CPI inflation by combining multiple network information are given to illustrate reliable finite-sample estimation and predictive performance of the developed methodology.

\bigskip


\end{abstract}

\setcounter{page}{1}
\newpage


\section{Introduction}\label{sec1}
\renewcommand{\theequation}{1.\arabic{equation}} \setcounter{equation}{0}

The vector autoregression (VAR) has become a standard econometric tool to model multivariate time series since the seminal work by \cite{Si80}. In the classic setting with a fixed number of time series variables, \cite{Lu06} and \cite{KL17} provide a comprehensive review of various VAR-based model estimation, inference and forecasting techniques. However, the conventional OLS or MLE estimation and forecasting methods often perform poorly when the number of variables diverges with the time series length. To address this issue in the high-dimensional time series setup, a commonly-used approach is to impose some structural restriction on transition matrices such as the sparsity or reduced-rank assumption and subsequently adopt regularized estimation methods \citep[e.g.,][]{NW11, KC15, DZZ16, BLM19, MPS23}. However, these models and methods often neglect potential inter-connection or network structure among a large number of variables, a feature which commonly exists for large-scale dynamic macroeconomic or financial systems. The network VAR introduced by \cite{ZPLLW17} whose model formulation is virtually similar to the so-called global VAR \citep{PSW04}, provides a flexible autoregressive framework for large-scale time series to jointly incorporate the network spillover (cross-lag) effect, momentum (own-lag) effect and nodal effect.

\smallskip

Network models are an effective tool to tackle interconnected systems with applications to various fields such as economics, finance and social networks. There have been extensive studies on estimation and inference of both the static network \citep[e.g.,][]{DY14, S17, DDLY18, BB19} and time-varying dynamic network \citep[e.g.,][]{KSAX10, ZLW10, CLLL25}. The existing literature on network VAR models typically assumes that the network structure is available and incorporates an observed adjacency matrix in VAR to capture the network spillover effect \citep[e.g.,][]{ZPLLW17, CFZ23, YSM23, YSM26, LPTW26}. They often limit attention to the assumption of a single adjacency matrix to measure linkages between nodes. This assumption may be restrictive for modern social or economic networks as agents in dynamic systems often interact through multiple channels, resulting in the so-called multi-layer network \citep{KABGMP14}. \cite{BCM25} introduce a factor network VAR model and propose a tensor decomposition of multiple adjacency matrices which may vary over time to achieve dimension reduction in the subsequent estimation procedure. In this paper, we propose an alternative approach for multi-layer network VAR, adopting the model averaging method to construct an optimal combination of candidate time-varying multi-layer network VAR models with different linkage structures (to capture network spillover effects), varying lagged orders of dependent variables and node-specific predictors.

\smallskip

Model averaging tackles model uncertainty by averaging over a number of candidate predictive models and serves as an alternative to model selection which aims to select one model according to a proper criterion. Model averaging methods can be broadly classified into two categories: Bayesian model averaging with prior information \citep{HMRV99}, and frequentist model averaging, which focuses on determining optimal weights for candidate models. Commonly-used frequentist model averaging approaches include the Mallows model averaging \citep{H07,LT20} and jackknife model averaging \citep{HR12, LZZZ19, YZL25}. These classic methods aim to select time-invariant combination weights. Their numerical performance may be unstable for dynamic environments with structural changes. \cite{SHLWZ21} is among the first to propose time-varying jackknife model averaging to accommodate smooth structural changes, with subsequent extensions to high-dimensional penalized time-varying model averaging \citep{SHWZ23}, factor-augmented regression \citep{CHL24} and quantile regression \citep{TW25}, and time-varying VAR \citep{SCG25}. The aforementioned works, however, do not account for interconnections between time series variables and potential multi-layer network structures. Our paper aims to fill in this gap by studying time-varying model averaging of multi-layer network VAR for large-scale time series with smooth structural changes.

\smallskip

Meanwhile, there has been increasing attention in recent years on time-varying VAR models, accommodating smooth structural changes over a long time span. \cite{GPY24} study estimation and inference of time-varying VAR models including the impulse response analysis for finite-dimensional time series. The time-varying VAR is naturally connected to the locally stationary time series model framework, which has been systematically studied in the literature since \cite{D97}. In the high-dimensional time-varying VAR model setup, \cite{DQC17} combine the $\ell_1$-regularised method with kernel smoothing to estimate transition matrices; \cite{XCW20} study the multiple break structure and estimate smooth time-varying covariance and precision matrices between the break points; and \cite{CLLL25} estimate dual network structures via directed Granger causality and undirected partial correlation linkages. A recent paper by \cite{LPTW26} further introduces a general time-varying network VAR model framework with a latent group structure, allowing both structural changes and structural breaks on model coefficients. However, they limit attention to a single-layer network structure to capture the network spillover effect.

\smallskip

In this paper, we propose a penalized time-varying model averaging (PTVMA) method to estimate the optimal weights for candidate multi-layer network VAR models whose time-varying coefficients are estimated via the classic local linear smoothing. The candidate models are constructed via distinct adjacency matrices in network spillover effects, varying lags in momentum effects and different subsets of node-specific exogenous predictors. In particular, we allow all the candidate models to be misspecified, enhancing the applicability of the proposed model and methodology. Under some regularity conditions, we establish the asymptotic optimality and convergence rates of the proposed time-varying weight estimates in the contexts of both the in-sample fitting and out-of-sample prediction. In particular, we show that the PTVMA criterion achieves the so-called ``selection consistency" when there exist correctly-specified candidate models, i.e., the sum of the estimated time-varying weights for correctly-specified candidate models approaches one asymptotically.

\smallskip

In addition to point forecasts, we further construct the prediction interval, accounting for forecast uncertainty. Specifically, we propose a PTVMA-based conformal prediction method to construct interval forecasts without requiring the number of candidate models or predictors to be fixed. The conformal prediction method is first introduced by \cite{VGS05} as a sequential approach to form prediction intervals and has been extended by \cite{LW14} and \cite{LGRTW18} to accommodate the nonparametric regression framework. In this paper, we make a further extension by combining the conformal prediction with the developed PTVMA method to construct prediction bands in the context of large-scale and locally stationary network time series. The proposed PTVMA-based conformal prediction algorithm is easy to implement and can guarantee the coverage probability without requiring the asymptotic distribution theory of the model averaging estimation.

\smallskip

Monte Carlo simulations demonstrate that the proposed PTVMA method delivers nearly identical in-sample fit and superior out-of-sample prediction relative to the benchmark, even when all candidate models are misspecified. By assigning time-varying weights to candidate models to accommodate parameter instability, the PTVMA criterion achieves selection consistency when correctly specified candidates exist, and concentrates weight on the best-approximating models when all candidates are misspecified. The proposed conformal prediction intervals exhibit robust finite-sample performance, with empirical coverage probabilities close to the nominal level, particularly as the sample size increases. The developed model and methodology are further applied to analyze the monthly CPI inflation in 36 economies from January 2007 to December 2024. The empirical results reveal multiple channels of CPI inflation transmission across economies, with the production, equity, trade, and policy networks alternately dominating global inflation transmission. Our model and method consistently outperform a broad set of competing models, including the static and time-varying network VAR models and state-of-the-art machine learning methods, highlighting the advantage of the proposed time-varying multi-layer network VAR and PTVMA.

\smallskip

The rest of the paper is organized as follows. Section \ref{sec2} introduces the model framework and the PTVMA methodology. Section \ref{sec3} gives the asymptotic properties of the developed weight estimation together with technical assumptions. Section \ref{sec4} proposes the PTVMA-based conformal prediction interval construction. Sections \ref{sec5} and \ref{sec6} report the simulation and empirical studies, respectively. Section \ref{sec7} concludes the paper. All the theoretical proofs are contained in the appendix. 


\section{Model and methodology}\label{sec2}
\renewcommand{\theequation}{2.\arabic{equation}} \setcounter{equation}{0}

In this section, we introduce the time-varying multi-layer network VAR model framework and propose the PTVMA method to determine a time-varying optimal combination of candidate models.

\subsection{Time-varying multi-layer network VAR}\label{sec2.1}

Consider $\sN$-dimensional vectors of time series observations $\Y_t=(y_{1,t},y_{2,t},\ldots,y_{\sN, t})^{\top}$, $1\leq t\leq {\sf T}$, with $\sN$ being the number of nodes in the large-scale network, and $\sQ$-dimensional node-specific vectors $\X_i=(x_{i,1},\ldots,x_{i,\sQ})^{\top}$, $1\leq i\leq \sN$. Throughout the paper, $\sN$ is allowed to diverge to infinity while $\sQ$ is a fixed positive integer.  Let ${\mathbf W}^{(k)}=(\omega_{ij}^{(k)})_{\sN\times \sN}$ denote the $k$-th adjacency matrix, which is observed and non-stochastic, $\omega_{ii}^{(k)}=0$ and $\omega_{ij}^{(k)}=1$ if $i$ follows $j$, $k=1,\ldots,\sK$. These adjacency matrices capture economic linkages, geographic distances, or other forms of connectedness. All the adjacency matrices in this paper can be either directed or undirected. We consider the following data generating process:
\begin{eqnarray}
y_{i,t}&=&\mu_{i,t}+\epsilon_{i,t}\notag\\
&=&\sum_{k=1}^{\sK}\beta_{tk}\sum_{j\neq i} \widetilde \omega_{ij}^{(k)} y_{j,t-1}+\sum_{p=1}^{\sP}\alpha_{tp}y_{i,t-p}+\sum_{q=1}^{\sQ}\gamma_{tq}x_{i,q}+\epsilon_{i,t}\notag\\
&=&\widetilde\Y_{i,t-1}^\top\bb_t+\overline\Y^{\top}_{i,t-1}\ba_t+\X_i^{\top}\bg_t+\epsilon_{i,t}\notag\\
&=:&\Z_{i,t-1}^\top\btheta_t+\epsilon_{i,t},\quad 1\leq i\leq \sN,\label{eq2.1}
\end{eqnarray}
where $\widetilde \omega_{ij}^{(k)}=\omega_{ij}^{(k)}/n_i^{(k)}$ with $n_i^{(k)}=\sum_{j\neq i} \omega_{ij}^{(k)}$ denoting the total number of nodes that $i$ follows according the $k$-th linkage, $\bb_t=(\beta_{t1},\ldots,\beta_{t\sK})^{\top}$, $\ba_t= (\alpha_{t1}, \ldots,\alpha_{t\sP} )^{\top}$ and $\bg_t=(\gamma_{t1},\ldots,\gamma_{t\sQ})^{\top}$ are unknown time-varying parameter vectors to be estimated, $\Z_{i,t}$ denotes the predictor vector with dimension $\sK+\sP+\sQ$, i.e.,
\begin{eqnarray}
&&\Z_{i,t}=\left(\widetilde\Y_{i,t}^\top,\ \overline\Y^{\top}_{i,t},\ \X_i^{\top}\right)^\top,\notag\\
&&\widetilde\Y_{i,t}=\left(\sum_{j\neq i} \widetilde \omega_{ij}^{(1)} y_{j,t},\ldots, \sum_{j\neq i} \widetilde \omega_{ij}^{(\sK)} y_{j,t}\right)^\top,\notag\\
&&\overline\Y_{i,t}=(y_{i,t},\ldots,y_{i,t-\sP+1})^\top,\quad \X_i=(x_{i,1},\ldots,x_{i,\sQ})^\top,\notag
\end{eqnarray}
$\btheta_t=(\bb_t^\top,\ba_t^\top, \bg_t^\top)^\top$, and $\epsilon_{i,t}=\sigma_i(t/T)u_{i,t}$ with $\sigma_i^2(\cdot)$ being a time-varying variance function and $u_{i,t}$ being independent and identically distributed (i.i.d.) with zero mean and unit variance.

\smallskip

The time-varying VAR in (\ref{eq2.1}) is split into multiple network spillover effects (or cross-lag effects), the momentum effects with order $\sP$ (or own-lag effects) and the nodal effects, all of which are allowed to evolve smoothly over time. Here we only include the first lag of network spillover effects for simplicity and the extension to higher-order cross-lag effects is straightforward. 


\subsection{Penalized time-varying model averaging }\label{sec2.2}

Unlike the model selection which aims to select one time-varying network VAR model, to reduce model estimation and prediction uncertainty caused by selection of adjacency matrices, order of lagged dependent variables and subset of node-specific exogenous predictors, we next propose the PTVMA procedure, seeking to construct a dynamic optimal combination of candidate time-varying network VAR models. Consider a sequence of candidate network VAR models indexed by $m=1,\ldots, \sM$, which may be misspecified for the underlying data generating process. Specifically, we formulate the $m$-th candidate model as
\begin{eqnarray}
y_{i,t}&=&\mu_{i,t}^{(m)}+\epsilon_{i,t}^{(m)}\notag\\
&=&\sum_{k\in {\cal I}_{\sK}^{(m)}}\beta_{tk}^{(m)}\sum_{j\neq i} \widetilde \omega_{ij}^{(k)} y_{j,t-1}+\sum_{p\in {\cal I}_{\sP}^{(m)}}\alpha_{tp}^{(m)}y_{i,t-p}+ \sum_{q\in  {\cal I}_{\sQ}^{(m)}}\gamma_{tq}^{(m)}x_{i,q} +\epsilon_{i,t}^{(m)},\nonumber\\
&=:& \Z_{i,t-1}^{(m)^\top}\btheta_t^{(m)}+\epsilon_{i,t}^{(m)}, \quad 1\leq i\leq \sN, \notag
\end{eqnarray}
where ${\cal I}_{\sK}^{(m)}=\{i_1,\ldots,i_{k_m}\}$ is a subset of $\{1,2,\ldots,\sK\}$, ${\cal I}_{\sP}^\m=\{i_1,\ldots,i_{p_m}\}$ is a subset of $\{1,2,\ldots,\sP\}$, ${\cal I}_{\sQ}^\m=\{i_1,\ldots,i_{q_m}\}$ is a subset of $\{1,2,\ldots,\sQ\}$,
$$\btheta_t^{(m)}=\left(\bb_t^{(m)\top},\ba_t^{(m)\top},\bg_t^{(m)\top}\right)^{\top}$$
with
\begin{eqnarray}
\bb_t^{(m)}=(\beta_{ti_1},\ldots,\beta_{ti_{k_m}})^{\top}=\left[\beta_{ti_1}(t/\sT),\ldots,\beta_{ti_{k_m}}(t/\sT)\right]^{\top},\notag\\
\ba_t^{(m)}=(\alpha_{ti_1},\ldots,\alpha_{ti_{p_m}})^{\top}=\left[\alpha_{ti_1}(t/\sT),\ldots,\alpha_{ti_{p_m}}(t/\sT)\right]^{\top},\notag\\
\bg_t^{(m)}=(\gamma_{ti_1},\ldots,\alpha_{ti_{q_m}})^{\top}=\left[\gamma_{ti_1}(t/\sT),\ldots,\gamma_{ti_{q_m}}(t/\sT)\right]^{\top},\notag
\end{eqnarray}
and $\Z_{i,t}^{(m)}$ denotes the predictor vector with dimension $k_m+p_m+q_m$. All the time-varying parameters are defined as smooth functions of scaled times. Let
$$\bar{\rho}= (\sK+\sP+\sQ)\vee \max \{ \rho_1, \rho_2, \cdots, \rho_\sM\} 
\ \ \ \text{with}\ \  \rho_m=k_m+p_m+q_m,$$
$\bmu_t^{(m)}=(\mu_{1,t}^\m,\ldots,\mu_{\sN, t}^\m)^\top$ and $\beps_t^\m=(\epsilon_{1,t}^\m,\ldots,\epsilon_{\sN, t}^\m)^\top$. Throughout the paper, we allow candidate time-varying network VAR models to be non-nested and misspecified, and both $\bar{\rho}$ and $\sM$ may diverge to infinity.

\smallskip

For any given $t_0$, we define the local linear estimate of ${\btheta}_{t_0}^{(m)}$ as
\begin{equation}\label{eq2.2}
\widehat{\btheta}_{t_0}^{(m)}=\left({\mathbf I}_{\rho_m},\ {\mathbf O}_{\rho_m}\right)\left(\sum_{t=1}^{\sT}\widetilde{\mathbf Z}_{t-1}^{(m)\top}\widetilde{\mathbf Z}_{t-1}^{(m)}K_{t,t_0}\right)^{-1}\left(\sum_{t=1}^{\sT}\widetilde{\mathbf Z}_{t-1}^{(m)\top}\Y_tK_{t,t_0}\right),
\end{equation}
where ${\mathbf I}_{\rho_m}$ and ${\mathbf O}_{\rho_m}$ are the identity and zero matrices, respectively, with size $\rho_m
\times \rho_m$,
\[
\widetilde{\mathbf{Z}}_{t-1}^{(m)}=\left[{\mathbf Z}_{t-1}^{(m)\top},\ \left(\frac{t-t_0}{Th}\right){\mathbf Z}_{t-1}^{(m)\top}\right],\quad \Z_{t}^{(m)}=(\Z_{1,t}^{(m)},\ldots,\Z_{\sN,t}^{(m)}),\quad K_{t,t_0}=K\left(\frac{t-t_0}{Th}\right)
\]
with $K(\cdot)$ being a kernel function and $h$ being a bandwidth.

\smallskip

Let
\[
{\mathscr W}=\left\{\w=(w^1,w^2,\ldots,w^\sM)^{\top}\in[0,1]^\sM:\ \sum_{m=1}^\sM w^m=1\right\}.
\]
With a given vector $\w=(w^1,w^2,\ldots,w^\sM)^{\top}\in{\mathscr W}$, we construct the model averaging estimators of time-varying parameters and  conditional mean as follows
\begin{eqnarray}
\widehat{\btheta}_t(\w)&=&\sum_{m=1}^\sM w^m\bPi^{\m\top}\widehat{\btheta}_t^\m,\label{eq2.3}\\
\widehat{\bmu}_{t}(\w)&=&\left[\widehat{\mu}_{1,t}(\w),\ldots,\widehat{\mu}_{\sN, t}(\w)\right]^{\top}=\sum_{m=1}^\sM w^m{\mathbf Z}_{t-1}^{(m){\top}}\widehat\btheta_t^{(m)},\label{eq2.4}
\end{eqnarray}
where $\bPi^\m$ is a matrix of size $\rho_m\times(\sK+\sP+\sQ)$ mapping $\Z_{i,t}$ to $\Z_{i,t}^\m$ and
\[
\widehat{\mu}_{i,t}(\w)=\sum_{m=1}^\sM w^m\Z_{i,t-1}^{(m)^{\top}}\widehat{\btheta}_t^{(m)}.
\]
We propose the PTVMA criterion:
\be\label{eq2.5}
{\sf PTVMA}_{t}(\w)=\sum_{s=1}^{\sT}\left[\Y_{s}-\widehat\bmu_s(\w)\right]^{\top}\left[\Y_s-\widehat\bmu_s(\w)\right]K_{s,t}+\lambda\sum_{m=1}^\sM w^m\rho_m,
\ee
where $\lambda$ is a user-specified tuning parameter. The optimal weights are determined by
\begin{equation}\label{eq2.6}
\widehat \w_t=\left( \widehat w_t^1,\widehat w_t^2,\ldots, \widehat w_t^\sM\right)^{\top}=\argmin_{\w\in {\mathscr W}}{\sf PTVMA}_{t}(\w).
\end{equation}
Finally the one-step ahead forecast combination is constructed by
\be\label{eq2.7}
\widehat{\Y}_{\sT+1}(\widehat{\w}_\sT)=\left[\widehat y_{1,\sT+1}(\widehat{\w}_\sT),\ldots,\widehat y_{\sN,\sT+1}(\widehat{\w}_\sT)\right]^{\top}\quad \text{with}\quad \widehat{y}_{i,\sT+1}(\widehat{\w}_\sT)=\sum_{m=1}^\sM \widehat w_\sT^m\Z_{i,\sT}^{(m)^{\top}}\widehat{\btheta}_\sT^{(m)}.
\ee


\section{Asymptotic theory}\label{sec3}
\renewcommand{\theequation}{3.\arabic{equation}} \setcounter{equation}{0}

In this section, we give some regularity conditions and present the asymptotic properties of the PTVMA estimation developed in Section \ref{sec2.2} including the asymptotic optimality and convergence rates. Section \ref{sec3.1} considers the in-sample fitting, i.e., $\widehat\w_t$ for the interior point $t$ in the sense that $0<t/\sT<1$, whereas Section \ref{sec3.2} tackles the out-of-sample prediction, i.e., $\widehat\w_\sT$ for any forecast time point $\sT+1$.

\subsection{In-sample asymptotic theory}\label{sec3.1}

To establish the asymptotic optimality of the proposed PTVMA method for in-sample fitting, we define the following infeasible locally weighted (in-sample) quadratic error loss:
\begin{equation}\label{eq3.1}
L_{t}(\w)=\sum_{s=1}^\sT[\bmu_s-\widehat{\bmu}_s(\w)]^\top[\bmu_s-\widehat{\bmu}_s(\w)]K_{s,t},
\end{equation}
where $\bmu_s=(\mu_{1,s},\ldots,\mu_{\sN, s})^\top$, $ \mu_{i,s} = \Z_{i,s-1}^{\top} \btheta_{s} $, $t$ is a fixed interior point such that $0<t/\sT<1$, in which case the local kernel smoothing utilizes the time series sample information on both sides of $t$. Denote
$$\bmu_t^*(\w)=\sum_{m=1}^\sM w^m\bmu_{t,\ast}^{(m)}=\left[\mu_{1,t}^{*}(\w),\ldots,\mu_{\sN,t}^{*}(\w)\right]^\top,$$
where $\bmu_{t,\ast}^{(m)}=(\mu_{1,t}^{(m)*},\ldots,\mu_{\sN,t}^{(m)*})^\top$ and $\mu_{i,t}^{(m)*}=\Z_{i,t-1}^{(m)\top}\btheta_{t,\ast}^{(m)}$ with $\btheta_{t,\ast}^{(m)}$ being the pseudo time-varying coefficient vector for the $m$-th candidate model to be defined in Assumption \ref{ass:1} below. Write
\begin{equation}\label{eq3.2}
L_t^*(\w)=\sum_{s=1}^\sT[\bmu_s-{\bmu}^*_s(\w)]^\top[\bmu_s-{\bmu}^*_s(\w)]K_{s,t}\quad \text{and}\quad \xi_{t}=\inf_{\w\in\calW}L_t^*(\w).
\end{equation}
The following assumptions are required to derive the in-sample asymptotic optimality.

\smallskip

\begin{condition}\label{ass:1}

(i) For the $m$-th candidate model, there exists a limit time-varying vector $\btheta_{s,\ast}^{(m)}$ such that
$$\left\|\widehat{\btheta}_s^{(m)}-\btheta_{s,\ast}^{(m)}\right\|=O_P\left(\bar{\rho}^{1/2}\zeta\right)$$
uniformly over $m$ and $s$, where $\zeta$ is a typical nonparametric uniform convergence rate. In addition, elements of $\btheta_{t,\ast}^{\m}$ are bounded over $t$ and $m$.

(ii) The second moments for each element of $\Z_{i,t}$ exist and are uniformly bounded.

(iii) Both true and pseudo time-varying coefficients $ \btheta_s $ and $ \btheta_{s,\ast}^{\m} $ are Lipschitz continuous over $ [0, 1] $ and their difference $ \| \btheta_s - \bPi^{\m\top} \btheta_{s,\ast}^{\m} \| = O(\bar{\rho}^{1/2}) $ uniformly over $ m $ and $ s $.

\end{condition}

\smallskip

\begin{condition}\label{ass:2}

(i)\ The kernel $K(\cdot)$ is positive, symmetric and continuous with a bounded support $[-1,1]$.

(ii)\ $\bar{\rho}\zeta(\sM\sN\sT h)^{1/2}$ is larger than a positive number, and
$$\bar{\rho}\sM\zeta \rightarrow0,\quad \left(\lambda+ \bar{\rho}\zeta\sM\sN\sT h\right)\bar{\rho}\xi_t^{-1}=o_P(1).$$

\end{condition}

\smallskip

\noindent{\bf Remark 3.1}.\ \ Assumption \ref{ass:1}(i) imposes a high-level condition on uniform convergence of local linear smoothing for misspecified candidate models, which is comparable to those adopted by some recent works on time-varying model averaging \citep[e.g.,][]{SHWZ23, SCG25}. With standard techniques for locally stationary time series and nonparametric kernel-weighted smoothing \citep[e.g.,][]{LPTW26}, we may derive an explicit rate: $\zeta=\sqrt{\log \sT/(\sN\sT h)}+h^2$ under some sufficient conditions which require higher moment restriction than Assumption \ref{ass:1}(ii). With this explicit rate for $\zeta$, assuming that both $\bar\rho$ and $\sM$ are bounded, $\sN\sT h^5=O(1)$ and $\lambda=O((\sN \sT h\cdot\log \sT)^{1/2})$, we may verify Assumption \ref{ass:2}(ii) if
\[
(\sN\sT h\cdot \log \sT)^{1/2}/\xi_t=o_P(1).
\]
The latter indicating that {\em all the candidate models must be misspecified}, otherwise $\xi_t\equiv0$.

\smallskip

Theorem \ref{thm:3.1} below establishes the point-wise asymptotic optimality of the PTVMA estimator, similar to Theorem 1 in \cite{SHWZ23} and Theorem 3.3 in \cite{CHL24}.

\smallskip

\begin{theorem}\label{thm:3.1}

Suppose that Assumptions \ref{ass:1} and \ref{ass:2} are satisfied. For any given time $t$ such that $0<t/\sT<1$, the proposed PTVMA estimator satisfies the asymptotic optimality, i.e.,
$$\frac{L_{t}(\widehat{\w}_t)}{\inf_{\w\in \calW}L_{t}(\w)}\stackrel{P}\longrightarrow 1.$$

\end{theorem}

\smallskip

The following extra condition is required to derive an explicit in-sample convergence rate of the PTVMA weight estimation.

\smallskip

\begin{condition}\label{ass:3}

(i) There exists a positive constant $\underline{\kappa}$ such that
$$
\lambda_{\min}\left(\frac{1}{\sN}{\sf E}\left[\widetilde{\bmu}^*_{s}\widetilde{\bmu}_{s}^{*\top}\right]\right)> \underline{\kappa}\quad \text{with}\quad \widetilde\bmu_{s}^* = \left(\bmu_{s,\ast}^{(1)},\ldots,\bmu_{s,\ast}^{(\sM)}\right)^\top
$$
for any $s$ such that $|s-t|\leq \sT h$.

(ii) For any given time point $t$,
\[
\left\Vert\sum_{s=1}^\sT\sum_{i=1}^\sN K_{s,t}\left(\Z_{i,s}\Z_{i,s}^\top-{\sf E}\left[\Z_{i,s}\Z_{i,s}^\top\right]\right)\right\Vert_{\sf F}=O_P\left(\bar\rho(\sN \sT h)^{1/2}\right)
\]
and
\[
\sum_{s=1}^\sT\sum_{i=1}^\sN \sigma_{i}^{2}(s/\sT)(u_{is}^2-1)K_{s,t}=O_P\left((\sN\sT h)^{1/2}\right).
\]

(iii) Let the following rate condition hold:
$$\left(\lambda+\zeta\sM\sN\sT h\right)\bar{\rho}(\sN\sT h)^{-1}+\bar\rho^2(\sN\sT h)^{-1/2}\rightarrow0$$

\end{condition}

\smallskip

\noindent{\bf Remark 3.2.}\ \ Assumption \ref{ass:3}(i) is comparable to some similar restriction on predictor vectors or the so-called ``hat matrix". Assumption \ref{ass:3}(ii) implies weak temporal dependence as well as cross-sectional correlation for predictors as in \cite{LPTW26}. It is easy to verify Assumption \ref{ass:3}(iii) when $\zeta=\sqrt{\log \sT/(\sN\sT h)}+h^2$, $\lambda=O((\sN\sT h \log\sT)^{1/2})$ and both $\bar\rho$ and $\sM$ are bounded.

\smallskip

\begin{theorem}\label{thm:3.2}

Suppose Assumptions \ref{ass:1}--\ref{ass:3} are satisfied. For any given time point $t$ such that $0<t/\sT<1$, $\widehat{\w}_t$ is a consistent in-sample estimate of $\w_t^*$, satisfying
\begin{equation}\label{eq3.3}
\left\|\widehat{\w}_t-\w_t^*\right\|=O_P\left((\lambda+\zeta\sM\sN\sT h)^{1/2}\bar{\rho}^{1/2}(\sN\sT h)^{-1/2}+\bar\rho(\sN\sT h)^{-1/4}\right),
\end{equation}
where $\w_t^*=\argmin_{\w\in\calW} {\sf E}[L_{t}(\w)]$.

\end{theorem}

\smallskip

\noindent{\bf Remark 3.3.}\ \ Theorem \ref{thm:3.2} above reveals that the point-wise convergence rate of the in-sample time-varying weight estimation $\widehat{\w}_t$ relies on the number of candidate models $\sM$, the maximum dimension of potential predictors $\bar\rho$, the tuning parameter in penalization $\lambda$, the bandwidth $h$ as well as $(\sN, \sT)$. With the sufficient conditions discussed in Remark 3.2 to verify Assumption \ref{ass:3}(iii), we may simplify the rate in (\ref{eq3.3}) to
\[
\left\|\widehat{\mathbf w}_t-{\mathbf w}_t^*\right\|=O_P\left(h+(\log \sT)^{1/4}(\sN\sT h)^{-1/4}\right).
\]

\smallskip

We next consider the scenario when there exist correctly-specified candidate time-varying multi-layer network VAR models. Let $\calD$ denote the index set of correctly-specified models, i.e., $m\in\calD$ if the $m$-th candidate model is correctly specified. Similar to $\xi_t$ in (\ref{eq3.2}), we define
$$\wt\xi_{t}=\inf_{\w\in\calW(\calD)} L_{t}^*(\w),\quad \calW(\calD)=\{\w=(w^1,w^2,\ldots,w^\sM)^{\top}\in\calW:\ w^m=0, \ m\in\calD\}.$$
Theorem \ref{thm:3.3} below shows that the proposed PTVMA criterion achieves the in-sample selection consistency (of correctly-specified candidate models) as the conventional model selection.

\smallskip

\begin{theorem}\label{thm:3.3}

Suppose Assumptions \ref{ass:1}, \ref{ass:2} and \ref{ass:3}(ii) are satisfied but with the condition $\left(\lambda+\zeta\sM\sN Th\right)\bar{\rho}\xi_t^{-1}=o_P(1)$ in Assumption \ref{ass:2}(ii) replaced by
\begin{equation}\label{eq3.4}
\left[\bar\rho(\lambda+ \zeta\sM\sN Th)+\bar\rho^2(\sN Th)^{1/2}\right]\wt\xi_t^{-1}=o_P(1),
\end{equation}
and $\calD$ is not empty. Then, $\sum_{m\in \mathscr{D}} \widehat{w}_t^m\stackrel{P}\longrightarrow 1$ for any interior time point $t$.

\end{theorem}


\subsection{Out-of-sample asymptotic theory}\label{sec3.2}

We next consider the asymptotic properties of $\widehat\w_\sT$ which has been adopted for constructing out-of-sample prediction combination. In this case, the local linear smoothing and ${\sf PTVMA}_{T}(\w)$ essentially use one-sided local sample information, which would subsequently slow down the weight estimation convergence. Define the out-of-sample prediction risk:
\[
R_{\sT+1}(\w)={\sf E}\left\{\left[\Y_{\sT+1}-\widehat{\Y}_{\sT+1}(\w)\right]^\top\left[\Y_{\sT+1}-\widehat{\Y}_{\sT+1}(\w)\right]\right\}-\frac{1}{\sT h}\sum_{i=1}^\sN\sum_{s=1}^\sT\sigma_{i}^{2}(s/\sT)K_{s,\sT},
\]
where $\widehat{\Y}_{\sT+1}(\w)$ is defined similarly to $\widehat{\Y}_{\sT+1}(\widehat\w_\sT)$ defined in (\ref{eq2.7}) but with $\widehat\w_\sT$ replaced by $\w$. Note that the second term of $R_{\sT+1}(\w)$ is independent of $\w$. In the out-of-sample prediction, we expect the PTVMA estimate $\widehat\w_\sT$ to asymptotically minimize $R_{\sT+1}(\w)$ over $\w\in{\calW}$. Furthermore, we let
\[
R_{\sT+1}^*(\w)={\sf E}\left\{\left[\Y_{\sT+1}-\Y_{\sT+1}^*(\w)\right]^\top\left[\Y_{\sT+1}-\Y_{\sT+1}^*(\w)\right]\right\}-\frac{1}{\sT h}\sum_{i=1}^\sN\sum_{s=1}^\sT\sigma_{i}^{2}(s/\sT)K_{s,\sT},
\]
and
\[
\xi_{\sT+1}^*=\inf_{\w\in \calW}R_{\sT+1}^*(\w),
\]
where $\Y_{\sT+1}^\ast(\w)=\sum_{m=1}^\sM w^m\Y_{\sT+1}^{(m)\ast}$ and $\Y_{\sT+1}^{(m)\ast}=(y_{1,\sT+1}^{(m)*},\ldots,y_{\sN,\sT+1}^{(m)*})^\top$ with $y_{i,\sT+1}^{(m)*}=\Z_{i,\sT}^{(m)\top}\btheta_{\sT,\ast}^{(m)}$. Let $\beps_{t}^{(m)*}=(\epsilon_{1,t}^{(m)*},\ldots,\epsilon_{\sN,t}^{(m)*})^\top$ with $\epsilon_{i,t}^{(m)*}=y_{i,t}-\mu_{i,t}^{(m)*}$, $1\leq t\leq \sT$, and $\beps_{\sT+1}^{(m)*}=(\epsilon_{1,\sT+1}^{(m)*},\ldots,\epsilon_{\sN, \sT+1}^{(m)*})^\top$ with $\epsilon_{i, \sT+1}^{(m)*}=y_{i, \sT+1}-y_{i, \sT+1}^{(m)*}$ for the out-of-sample period.

\smallskip

\begin{condition}\label{ass:4}

(i) $\sigma_i(\cdot)$, $1\leq i\leq \sN$, are uniformly bounded and Lipschitz continuous.

(ii) There exists $\delta_1>0$ such that
\begin{equation}\label{eq3.5}
\sup_\sT{\sf E}\left[\xi_{\sT+1}^{*-1/2}\left(\max_{1\leq m\leq \sM}\left\|\widehat{\Y}_{\sT+1}^\m-\Y_{\sT+1}^{(m)*}\right\|+\max_{1\leq m,m'\leq \sM}\left\|\widehat{\Y}_{\sT+1}^{(m)}-\Y_{\sT+1}^{(m)*}\right\|^{1/2} \left\|\beps_{\sT+1}^{(m')*}\right\|^{1/2}\right)\right]^{2+\delta_1}<\infty,
\end{equation}
and in addition, $\sup_\sT\xi_{\sT+1}^*/\sN<\infty$.

(iii) Let the following conditions hold:
\[
\xi_{\sT+1}^{\ast-1}\bar\rho \left[\sM\sN \zeta+\lambda(\sT h)^{-1}\right]\rightarrow0,\quad
\xi_{\sT+1}^{\ast-1}\left[\bar{\rho} \sN  h+\sM^2\sN^{1/2}(\sT h)^{-1/2}\right]\rightarrow0.
\]

(iv) For any given time point $t$ such that $\sT- \sT h \leq t\leq \sT$,
$$
\left\| \frac{1}{\sT h} \sum_{t=1}^{\sT} K_{t,\sT} \mE \left[\Z_{i,t-1}\Z_{i,t-1}^\top\right] - \mE \left[\Z_{i,\sT}\Z_{i,\sT}^\top\right] \right\|_{\sf op} %
= O_P \left(h\right),
$$
where $\|\cdot\|_{\sf op}$ denotes the operator norm of a square matrix.


\end{condition}

\smallskip

\noindent{\bf Remark 3.4.}\ \ Assumption \ref{ass:4}(i) imposes a mild smoothness restriction on heterogeneous variance functions $\sigma_i^2(\cdot)$ and is automatically satisfied for the case of homoskedasticity. The condition (\ref{eq3.5}) requires the potential best prediction (over all the candidate models) has finite risk in large-scale networks. A similar condition can be found in Assumptions 3(iii) and 5(iii) of \cite{ZZ23}. In fact, it provides a sufficient condition to derive the uniform integrability (over $\w$) of $\xi_{\sT+1}^{*-1}\left|\|\Y_{\sT+1}-\widehat{\Y}_{\sT+1}(\w)\|^2-\|\Y_{\sT+1}-\Y^*_{\sT+1}(\w)\|^2\right|$, which is required in the proof of Theorem \ref{thm:3.4} in the appendix. If, in addition, $\sup_\sT\xi_{\sT+1}^*/\sN$ is strictly larger than a positive constant, we may replace $\xi_{\sT+1}^\ast$ by $\sN$ in (\ref{eq3.5}) and further show that all the candidate models are misspecified. In fact, if the $m_0$-th candidate model is correctly specified, $\bPi^{(m_0)\top}\widehat\btheta_\sT^{(m_0)}$ converges to the true coefficient vector $\btheta_\sT$. Consequently, $\Y_{\sT+1}^{(m_0)*}=\bmu_{\sT+1}(1+o_P(h))$ and
\[
 \xi_{\sT+1}^{*}=\inf_{\w\in\calW}R_{\sT+1}^*(\w)\leq {\sf E}\left[\left(\Y_{\sT+1}-\Y_{\sT+1}^{(m_0)*}\right)^\top\left(\Y_{\sT+1}-\Y_{\sT+1}^{(m_0)*}\right)\right]-\frac{1}{\sT h}\sum_{i=1}^\sN\sum_{s=1}^\sT\sigma_{i}^{2}(s/\sT)K_{s,\sT}=o(\sN),
\]
which leads to contradiction.  Assumption \ref{ass:4}(iii) is comparable to Assumption \ref{ass:2}(ii) and can be easily justified by using the sufficient conditions discussed in Remarks 3.1 and 3.2. Assumption \ref{ass:4}(iv) contains a high-level condition for locally stationary $\Z_{i,t}$.

\smallskip

\begin{theorem}\label{thm:3.4}

Suppose that Assumptions \ref{ass:1}, \ref{ass:2}(i) and \ref{ass:4} are satisfied. For any forecast time point $\sT+1$,
$$\frac{R_{\sT+1}(\widehat{\w}_\sT)}{\inf_{\w\in \calW}R_{\sT+1}(\w)}\stackrel{P}\rightarrow 1.$$

\end{theorem}

\smallskip

\noindent{\bf Remark 3.5.}\ \ Theorem \ref{thm:3.4} above shows that the PTVMA estimator is out-of-sample asymptotically optimal in the sense that its out-of-sample prediction risk is asymptotically equivalent to that of the infeasible optimal weight estimator $\w_\sT^\ast:=\inf_{\w\in \calW}R_{\sT+1}(\w)$. A similar property is also derived by \cite{ZL23} and \cite{ZZ23}. A combination of Theorems \ref{thm:3.1} and \ref{thm:3.4} demonstrates that our proposed PTVMA estimator achieves not only the minimum of in-sample estimation error at the interior time point but also the minimum of out-of-sample prediction risk.

\smallskip

To derive an explicit out-of-sample convergence rate of $\widehat{\w}_\sT$, we require the following assumption which is analogous to Assumption \ref{ass:3}(i)(iii) in Section \ref{sec3.1}.

\smallskip

\begin{condition}\label{ass:5}

(i) There exists a positive constant $\underline{\kappa}_\ast$ such that
\[
\lambda_{\min}\left(\frac{1}{\sN}{\sf E}\left[\wt{\Y}^*_{\sT+1}\wt{\Y}_{\sT+1}^{*\top}\right]\right)> \underline{\kappa}_\ast,\quad \wt\Y_{\sT+1}^* = \left(\Y_{\sT+1,\ast}^{(1)},\ldots,\Y_{\sT+1,\ast}^{(\sM)}\right)^\top.
\]

(ii) Let the following condition hold:
\[
\left(\lambda+\zeta\sM\sN\sT h\right)\bar{\rho}(\sN\sT h)^{-1}+\sM^2(\sN\sT h)^{-1/2}\rightarrow0.
\]

\end{condition}

\smallskip

\smallskip

\begin{theorem}\label{thm:3.5}

Suppose that Assumptions \ref{ass:1}, \ref{ass:2}(i), \ref{ass:3}(ii)(iii), \ref{ass:4}(i) and \ref{ass:5} are satisfied, and $\w_\sT^*$ uniquely minimizes $R_{\sT+1}(\w)$ in $\calW$. Then, for any forecast time point $\sT+1$, $\widehat{\w}_\sT$ is a consistent estimate of $\w_\sT^\ast$, satisfying
\begin{equation}\label{eq3.6}
\left\Vert\widehat{\w}_\sT-\w_\sT^*\right\Vert=O_P\left( (\lambda+\zeta\sM\sN\sT h)^{1/2}\bar{\rho}^{1/2}(\sN\sT h)^{-1/2}+\sM(\sN\sT h)^{-1/4}+ (\bar{\rho} h)^{1/2 } \right).
\end{equation}

\end{theorem}

\smallskip

\noindent{\bf Remark 3.6.}\ \ As in Theorem \ref{thm:3.2}, the convergence rate of $\widehat{\w}_\sT$ relies on $\sM$, $\bar\rho$, $\lambda$, $h$ and $(\sN, \sT)$. In fact, the rate in (\ref{eq3.6}) is slightly slower than that in (\ref{eq3.3}) due to the extra term $h^{1/2}$, which is caused by the well-known boundary effect of kernel smoothing. Furthermore, as discussed in Remark 3.3, under some sufficient conditions, we may simplify the rate in (\ref{eq3.6}) to
\[
\left\|\widehat{\mathbf w}_\sT-{\mathbf w}_\sT^*\right\|=O_P\left(h^{1/2}+(\log \sT)^{1/4}(\sN \sT h)^{-1/4}\right).
\]

\smallskip

We finally consider the case when there exist correctly-specified candidate models. Let $\calD$ and $\calW(\calD)$ be defined as in Section \ref{sec3.1} and
$$\wt\xi_{\sT+1}^{*}=\inf_{\w\in\calW(\calD)}R_{\sT+1}^*(\w),$$
denoting the minimum out-of-sample prediction risk over misspecified models. The following theorem is an extension of Theorem \ref{thm:3.3} from in-sample fitting to out-of-sample prediction.

\smallskip

\begin{theorem}\label{thm:3.6}

Suppose that Assumptions \ref{ass:1}, \ref{ass:2}(i), \ref{ass:3}(ii)(iii) and \ref{ass:4}(i) are satisfied. In addition,
\begin{equation}\label{eq3.7}
\widetilde\xi_{\sT+1}^{\ast-1}\bar\rho \left[\sM\sN \zeta+\lambda(\sT h)^{-1}\right]+
\wt\xi_{\sT+1}^{\ast-1}\left[\bar{\rho} \sN h+\sM^2\sN^{1/2}(\sT h)^{-1/2}\right]=o_P(1),
\end{equation}
and $\calD$ is not empty. Then, for any forecast time point $\sT+1$, we have $\sum_{m\in {\calD}} \widehat{w}_\sT^m\stackrel{P}\rightarrow 1$.

\end{theorem}


\section{PTVMA-based conformal prediction}\label{sec4}
\renewcommand{\theequation}{4.\arabic{equation}} \setcounter{equation}{0}

Despite the growing literature on frequentist model averaging, construction of prediction intervals using model averaging estimators remain under-explored. A commonly-used approach to constructing interval forecasts often relies on the asymptotic distribution of the model averaging estimator. The relevant literature can be divided into two strands. The first examines the limiting distributions of least squares averaging estimators within a local asymptotic framework, where time-invariant regression coefficients lie in a local $\sT^{-1/2}$ neighborhood of zero \citep{HC03, H14,Liu15}. However, the realism of this local misspecification assumption has been questioned by \cite{RZ03}, who argue that it lacks face validity in certain empirical settings. The second strand relies on the asymptotic distribution of model averaging estimators under proper model assumptions such as the (non-)linear regression with fixed parameters \citep{ZL18,YLSZH24} and time-varying coefficient regressions \citep{SHWZ23,TW25}. These results are restricted to the setting where both the number of candidate models and the number of predictors remain fixed. We next propose the PTVMA-based conformal prediction approach to construct interval forecasts without requiring the number of candidate models or predictors to be fixed; see Algorithm \ref{alg}. Without loss of generality, we assume that $\sT h$ takes a positive integer value.

\smallskip


\begin{algorithm}\captionsetup{labelfont={sc,bf}} 
  \caption{The PTVMA-based conformal prediction}
  \begin{algorithmic}\label{alg}

\STATE \textbf{Step 1}: Generate residuals by implementing PTVMA to the training sample.

\smallskip

\textbf{1.1}.\ \ \ Calculate $\widehat{\btheta}_t^{(m)}$ by Eq (\ref{eq2.2}) and obtain $\widehat{\mu}_{i,t}^{(m)}=\Z_{i,t-1}^{(m)\top}\widehat\btheta_{t}^{(m)}$, $t=\sT-\sT h,\ldots,\sT$, $m=1,\ldots,\sM$;

\textbf{1.2}.\ \ \ Estimate $\widehat{\w}_\sT$ by Eq (\ref{eq2.6}) with $t=\sT$, and construct
\[
\widehat{\epsilon}_{i,t}:=\widehat{\epsilon}_{i,t}(\widehat{\w}_\sT)=y_{i,t}-\widehat{\mu}_{i,t}(\widehat{\w}_\sT),\quad \widehat{\mu}_{i,t}(\widehat{\w}_\sT)=\sum_{m=1}^{\sM}\widehat{w}_{\sT}^m\widehat{\mu}_{i,t}^{(m)},\quad t=\sT-\sT h,\ldots,\sT;
\]

\textbf{1.3}.\ \ \ Construct the set of residuals:
$$\widehat{\boldsymbol\epsilon}_i=\{|\widehat{\epsilon}_{i,t}|,\quad t=\sT-\sT h,\ldots,\sT\}.$$

\medskip

\STATE\textbf{Step 2}: Determine the quantile of absolute residuals.

\smallskip

Let $\widehat{\epsilon}_i^{\alpha}$ be defined as the $\lceil(1-\alpha)(n_i+1)\rceil$-th smallest value of the elements in $\widehat{\boldsymbol\epsilon}_i$, where $n_i$ is the cardinality of $\widehat{\boldsymbol\epsilon}_i$.

\medskip

\textbf{Step 3}: Generate the prediction interval.

\smallskip

Integrate the quantile $\widehat{\epsilon}_i^{\alpha}$ with $\widehat{y}_{i,\sT+1}(\widehat{\w}_\sT)$ defined in Eq (\ref{eq2.7}) to generate the conformal prediction interval with confidence level $1-\alpha$:
$$C_{i}^{\alpha}(\Z_{\sT})=[\widehat{y}_{i,\sT+1}(\widehat{\w}_\sT)-\widehat{\epsilon}_i^{\alpha},\quad \widehat{y}_{i,\sT+1}(\widehat{\w}_\sT)+\widehat{\epsilon}_i^{\alpha}].$$

\medskip

\end{algorithmic}
\end{algorithm}

\smallskip

To establish a valid coverage property of the proposed conformal prediction procedure, we require the following technical assumptions.

\begin{condition}\label{ass:6}

(i) $\sigma_i(\cdot)$ is second-order continuously differentiable, and $\underline{\sigma}\leq\sigma_i(\cdot)\leq\overline\sigma$, where $\underline{\sigma}$ and $\overline\sigma$ are finite positive constants.

(ii) $\{u_{i,t}\}$ are i.i.d. with mean zero, unit variance and $\mE [u_{i,t}^4]<\infty$. Furthermore, their common cumulative distribution function ${\sf F}$ satisfies a Lipschitz condition, i.e.,
\[\sup_{x_1,x_2}|{\sf F}(x_1)-{\sf F}(x_2)|\leq M_{\sf F}|x_1-x_2|,\quad 0<M_{\sf F}<\infty.
\]

(iii) Let $|\btheta_{\sT+1} - \btheta_\sT | \leq M_\theta/\sT$ with $M_\theta$ being a positive constant.

\end{condition}

Letting $\wh{u}_{i,t}=\sigma_i^{-1}(\tau_t)\widehat{\epsilon}_{i,t}$ with $\tau_t=t/\sT$, we define the following in-sample and out-of-sample performance measures:
$$\delta_{i,\sT}^2=\frac{1}{\sT h+1}\sum_{t=\sT-\sT h}^{\sT}(|\wh{u}_{i,t}|-|u_{i,t}|)^2\quad \text{and}\quad \psi_{i,\sT+1}=\big||\wh{u}_{i,\sT+1}|-|u_{i,
\sT+1}|\big|.$$
The following proposition gives the finite-sample upper bound on the gap between the actual coverage probability and the nominal level $1-\alpha$ and further establishes the asymptotic validity of the developed PTVMA-based conformal prediction method.

\smallskip

\begin{proposition}\label{prop:4.1}

Suppose that Assumption \ref{ass:6} is satisfied. There exists a positive constant $c_\dag$ such that
\begin{equation}\label{eq4.1}
\left|{\sf P}\left(y_{i,\sT+1}\in C_{i}^{\alpha}(\Z_{\sT})\right)-(1-\alpha)\right|\leq c_\dag \left(\left[{\sf E}\left(\delta_{i,\sT}^{2}\right)\right]^{1/4}+\left[{\sf E}\left(\psi_{i,\sT+1}\right)\right]^{1/2}+h^{4/5}+(\sT h)^{-1/3}\right)
\end{equation}
for any fixed $i$, if $\left[{\sf E}\left(\psi_{i,\sT+1}\right)\right]<1$. Furthermore, if 
$$ \bar{\rho} \widetilde\xi_{\sT+1}^{\ast-1} \left[ \bar\rho \sM\sN \zeta + \bar\rho\lambda(\sT h)^{-1} + \bar{\rho} \sN  h + \sM^2\sN^{1/2}(\sT h)^{-1/2} \right]= o_P(1) $$ 
and $\calD$ is not empty, we have
\begin{equation}\label{eq4.2}
{\sf P}\left(y_{i,\sT+1}\in C_{i}^{\alpha}(\Z_{\sT})\right)=1-\alpha+o(1).
\end{equation}

\end{proposition}

\smallskip

\noindent{\bf Remark 4.1.}\ \ We allow all the candidate models to be misspecified to obtain the finite-sample upper bound for the deviation of the actual coverage probability from the nominal level in (\ref{eq4.1}). To further derive the asymptotic validity of the constructed prediction interval in (\ref{eq4.2}), we need to prove that ${\sf E}(\delta_{i,\sT}^{2})\rightarrow0$ and ${\sf E}(\psi_{i,\sT+1})\rightarrow0$, which requires $\calD$ to be not empty, indicating that there exists at least one correctly-specified candidate model so that Theorem \ref{thm:3.6} is applicable.


\section{Monte-Carlo simulation}\label{sec5}
\renewcommand{\theequation}{5.\arabic{equation}} \setcounter{equation}{0}

To demonstrate advantages of the proposed PTVMA procedure, this section creates a sharp challenge for classic approaches in which any single-network specification may be misspecified. In the two simulation settings, model selection over a set of single-network specifications is global misspecified because no candidate model coincides with the true DGP throughout the entire period. By comparing the finite-sample performance of PTVMA with several competing models, our simulation experiment examines whether the proposed PTVMA can adaptively learn relevant network layers and variables over time and deliver accurate predictions despite persistent specification uncertainty.

\subsection{Evaluation measures on finite-sample performance}

To evaluate the in-sample performance, we calculate the in-sample ratio of the root mean squared prediction errors (RMSPE) to compare the proposed model with benchmark models:
\begin{align*}
{\sf RMSPE}_\text{in} = \sqrt{ \frac{ \sum_{r=1}^{R} \sum_{t=1}^{\sT} \sum_{i=1}^{\sN} \big( y_{it:r} - \widehat{y}_{it:r}^\text{PTVMA} \big)^2 }{\sum_{r=1}^{R} \sum_{t=1}^{\sT} \sum_{i=1}^{\sN} \big( y_{it:r} - \widehat{y}_{it:r}^\text{Benchmark} \big)^2 } },
\end{align*}
where $R=1000$ is the replication number in simulations, $y_{it:r} $ is the true value in the $r$-th replication, $ \widehat{y}_{it:r}^\text{PTVMA}$ is the fitted value from the PTVMA, and $ \widehat{y}_{it:r}^\text{Benchmark} $ is from a benchmark model. We consider two benchmarks. The network vector autoregressive (NAR) model \citep{ZPLLW17} includes the same response and predictors as the DGP but restricts the parameters to be time-invariant, estimating them via the OLS. The time-varying NAR (tv-NAR) model \citep{LPTW26} adopts the same DGP specification, utilizing a local linear approach to estimate time-varying coefficients. The proposed PTVMA method outperforms the benchmark model when the ${\sf RMSPE}_\text{in}$ ratio is smaller than one.

\smallskip

The out-of-sample predictive performance is assessed by employing the rolling forecasting procedure. Specifically, we compute the out-of-sample ratio of the RMSPE for the $l$-step-ahead out-of-sample rolling prediction of $ \widehat{\bm{Y}}_{\sT+l} $:
\begin{align*}
{\sf RMSPE}_\text{out} (l) = \sqrt{ \frac{ \sum_{r=1}^{R} \sum_{t=\sT_0+1}^{\sT-l} \sum_{i=1}^{\sN} \big( y_{i,t+l:r} - \widehat{y}_{i,t+l:r}^\text{PTVMA} \big)^2 }{ \sum_{r=1}^{R} \sum_{t=\sT_0+1}^{\sT-l} \sum_{i=1}^{\sN} \big( y_{i,t+l:r} - \widehat{y}_{i,t+l:r}^\text{Benchmark} \big)^2 } },
\end{align*}
where $ \widehat{y}_{i,t+l:r}^\text{PTVMA}$ and $\widehat{y}_{i,t+l:r}^\text{Benchmark}$ denotes the $l$-step-ahead rolling prediction in the $r$-th replication using the proposed PTVMA and benchmark methods, respectively, $\sT_0 = \lfloor 0.4\sT \rfloor $ is the rolling window length. Under the iterated scheme, multi-step-ahead forecasts are obtained by estimating a one-step-ahead model and recursively replacing future lagged values with their forecasts.

\smallskip

The performance of conformal prediction is evaluated by the empirical coverage probability. Specifically, we use the test set to calculate empirical coverage probabilities:
\begin{align*}
{\sf CP} = \frac{1}{R\sN(\sT-\sT_0)} \sum_{r=1}^{R} \sum_{t=\sT_0+1}^{\sT} \sum_{i=1}^{\sN} I \left\{ y_{i,t:r} \in C_{i:r}^{\alpha}(\mathbf{Z}_{t-1}) \right\},
\end{align*}
where $C_{i:r}^{\alpha}(\mathbf{Z}_{t-1})$ is the conformal prediction interval constructed for individual $i$ at time $t$ in the  $r$-th replication.

\subsection{One candidate model is correctly specified over half-period.}

The DGP 1 simulates a setting in which the response variable $ y_{it} $ is driven by two distinct networks, $\mathbf{W}^{(1)}$ and $\mathbf{W}^{(2)}$, with  time-varying impacts $\beta_{t1}$ and $\beta_{t2}$. Consider
\begin{align}
y_{i,t} = \beta_{t1} \sum_{j \neq i} \widetilde{\omega}_{ij}^{(1)} y_{j,t-1} +  \beta_{t2} \sum_{j \neq i} \widetilde{\omega}_{ij}^{(2)} y_{j,t-1} + \alpha_{t1} y_{i,t-1} + \alpha_{t2} y_{i,t-2}+ \sum_{q=1}^{3} \gamma_{tq} x_{i,q} + \epsilon_{i,t},
\label{Eq Simulation DGP 1}
\end{align}
where $ \widetilde{\omega}_{ij}^{(k)} = \omega_{ij}^{(k)} / \sum_{j \neq i} \omega_{ij}^{(k)} $ is the $ (i,j) $-th element of the matrix obtained by row-normalizing the adjacency matrix $ \mathbf{W}^{(k)} $, $ \mathbf{X}_i = \left( x_{i,1}, x_{i,2}, x_{i,3} \right)^\top $ is independently drawn from a three-dimensional normal distribution across $ i $ with mean $ \bm{0} $ and covariance $ \bm{\Sigma}_X = (\sigma_{x,q_1q_2})_{3 \times 3} $ with $ \sigma_{x,q_1q_2} = 0.5^{|q_1-q_2|} $ for $ 1 \leq q_1, q_2 \leq 3 $, the error term $ \epsilon_{i,t} = \sigma_i(t/\sT) u_{i,t} $, $ u_{i,t} $ is independently simulated by the standard normal distribution, $ \sigma_i(t/\sT) = (0.25 \sin (t/\sT)+0.3)^{0.5} $ for odd $ i $, and $ \sigma_i(t/\sT) = (0.25 \cos (t/\sT)+0.3)^{0.5} $ for even $ i $. True time-varying coefficients are set as
\be
&& \beta_{t1} = 0.4 / (1 + \exp \{ 50(t/\sT-0.5) \} ),\quad \beta_{t2} = 0.4 / (1 + \exp \{ -50(t/\sT-0.5) \} ),\notag\\
&& \alpha_{t1} = -0.3 \sin (\pi t/\sT) + 0.25,\quad \alpha_{t2} = -0.3 \cos (\pi t/\sT) + 0.25,\notag\\
&&\gamma_{t1} = 2 \cos (\pi t/\sT),\quad  \gamma_{t2} = 2 (t/\sT - 0.5)^2,\quad \gamma_{t3} = 1 - 2t/\sT.\notag
\ee

\smallskip

The first adjacency matrix $ \mathbf{W}^{(1)} $ is generated by the stochastic block model with two blocks. Each node is randomly assigned for one block label with equal probability and probabilities within and between blocks are generated by: $ {\sf P} (\omega_{ij}^{(1)} = 1) = 0.5 $ if $ i $ and $ j $ are in the same block; $ {\sf P} (\omega_{ij}^{(1)} = 1) = 0.1 $ otherwise. The second adjacency matrix $ \mathbf{W}^{(2)} $ is generated by using the power-law distribution model. The power-law distribution of degrees widely exists in the real world, where most agents have low degrees while a few hubs have very high degrees. We follow \citet{ZPLLW17} and define $ \mathbf{W}^{(2)} $ as follows. First, generate for each agent's $ degree_i $ via the discrete power-law distribution $ {\sf P} ({\sf degree}_i = k) \propto k^{-c} $, where $c $ is the exponent parameter and a smaller value implies a heavier distribution tail.\footnote{Here, we set $ c = 2.5 $ and the created networks share a similar density size at around $ 20\% $.} Next, we randomly select ${\sf degree}_i $ agents to be agent $ i $'s followers.

\smallskip

We consider $ \sM = 8 $ candidate models in the proposed PTVMA for DGP 1 and list their respective predictors in Table \ref{Tab Candidate models for DGP 1 and 2}. Note that $ \beta_{t1} $ is significant in the first half of the simulation period but not in the second, while $ \beta_{t2} $ exhibits the reverse pattern that is insignificant in the first half and significant in the second. Consequently, candidate models 2 and 4 are correctly specified in the first and second half of the sample period, respectively, but misspecified in the other period. We also involve irrelevant predictors $ x_{i4} $, $ x_{i5} $, $ \sum_{j \neq i} \widetilde{\omega}_{ij}^{(3)} y_{j,t-1} $ and $ \sum_{j \neq i} \widetilde{\omega}_{ij}^{(4)} y_{j,t-1} $ into some candidate models to demonstrate that the PTVMA method can dynamically select the correctly specified model if it is in the candidate models. Irrelevant predictors $ x_{i4} $ and $ x_{i5} $ are independently drawn from a normal distribution with mean $ 1 $ and variance $ 1 $. We refer to \cite{ZY18} and adopt the left-right matrix as one irrelevant adjacency matrix $ \mathbf{W}^{(3)} $, where each agent interacts only with its left and right neighbors. The other irrelevant adjacency matrix $ \mathbf{W}^{(4)} $ is constructed by a distance-based spatial weight matrix. We first independently generate agent $ i $'s location $ (x_i, y_i) $ from $ \mathrm{U} [0, 1] $, and then compute the pairwise Euclidean distance $ d_{ij} $ between agent $ i $ and $ j $. The spatial weight matrix is constructed using an exponential distance-decay function, with off-diagonal elements $ \omega_{ij}^{(4)} = \mathrm{I} \left( \exp \{ -10 d_{ij} \} < d_0 \right) $ and zero diagonal, where $ d_0 $ is the 25\% quantile of $ \{ d_{ij} \} $.

\smallskip

As reported in Table \ref{Tab Simulations DGP 1_In sample performance}, the PTVMA delivers uniformly better in-sample fit than the NAR benchmark. The RMSPE ratio (PTVMA/NAR) ranges from $ 0.746 $ to $ 0.825 $ across all $ (\sN, \sT) $ designs, indicating significant reductions of in-sample prediction error. Furthermore, PTVMA delivers nearly the same in-sample performance as the correct specification (tv-NAR model), with RMSPE ratios tightly around $ 1.002 $ and decreasing as either $ \sN $ or $ \sT $ grows. Given that the tv-NAR model is correctly specified here, this comparison suggests that the PTVMA approximates the true model specification without priori information.

\smallskip

Table \ref{Tab Simulations DGP 1_In sample performance} also reports the average weights of eight candidate models over the $ 1000 $ replications. The estimated weights are highly concentrated on candidate model 2 and 4, each receiving about $ 0.48 $ and jointly accounting for more than $ 95\% $ of the total weight. In contrast, the (averaged) time-varying weight estimates for the other six candidate models are negligible across all $ \sN $ and $ \sT $. This pattern is consistent with the design of DGP 1: the true network is generated from a combination of $ \mathbf{W}^{(1)} $ and $ \mathbf{W}^{(2)} $ while $ \mathbf{W}^{(3)} $ and $ \mathbf{W}^{(4)} $ are irrelevant alternatives. Overall, the evidence indicates that PTVMA places most of the weight on DGP-relevant network components while shrinking irrelevant candidates toward zero, confirming Theorem \ref{thm:3.3} in finite samples.

\begin{table}[H]
\centering
\begin{spacing}{1.3}

\caption{\textbf{Candidate models for DGPs 1 and 2 in the simulation study.}}
\label{Tab Candidate models for DGP 1 and 2}

\begin{tabular}{cll}
	\toprule
	Index & \multicolumn{1}{c}{DGP 1 predictors}                                                                                                   & \multicolumn{1}{c}{DGP 2 predictors}                                                                                          \\ \hline
	cm 1  & $ \sum_{j \neq i} \widetilde{\omega}_{ij}^{(1)} y_{j,t-1} $, $ y_{i,t-1} $, $ x_{i1} $, $ x_{i2} $, $ x_{i3} $, $ x_{i4} $, $ x_{i5} $ & $ \sum_{j \neq i} \widetilde{\omega}_{ij}^{(1)} y_{j,t-1} $, $ y_{i,t-1} $, $ x_{i1} $, $ x_{i2} $, $ x_{i3} $, $ x_{i4} $    \\
	cm 2  & $ \sum_{j \neq i} \widetilde{\omega}_{ij}^{(1)} y_{j,t-1} $, $ y_{i,t-1} $, $ y_{i,t-2} $, $ x_{i1} $, $ x_{i2} $, $ x_{i3} $          & $ \sum_{j \neq i} \widetilde{\omega}_{ij}^{(1)} y_{j,t-1} $, $ y_{i,t-2} $, $ x_{i1} $, $ x_{i2} $, $ x_{i3} $, $ x_{i5} $    \\
	cm 3  & $ \sum_{j \neq i} \widetilde{\omega}_{ij}^{(2)} y_{j,t-1} $, $ y_{i,t-1} $, $ x_{i1} $, $ x_{i2} $, $ x_{i3} $, $ x_{i4} $, $ x_{i5} $ & $ \sum_{j \neq i} \widetilde{\omega}_{ij}^{(1)} y_{j,t-1} $, $ y_{i,t-1} $, $ y_{i,t-2} $, $ x_{i1} $, $ x_{i2} $, $ x_{i3} $ \\
	cm 4  & $ \sum_{j \neq i} \widetilde{\omega}_{ij}^{(2)} y_{j,t-1} $, $ y_{i,t-1} $, $ y_{i,t-2} $, $ x_{i1} $, $ x_{i2} $, $ x_{i3} $          & $ \sum_{j \neq i} \widetilde{\omega}_{ij}^{(2)} y_{j,t-1} $, $ y_{i,t-1} $, $ x_{i1} $, $ x_{i2} $, $ x_{i3} $, $ x_{i4} $    \\
	cm 5  & $ \sum_{j \neq i} \widetilde{\omega}_{ij}^{(3)} y_{j,t-1} $, $ y_{i,t-1} $, $ x_{i1} $, $ x_{i2} $, $ x_{i3} $, $ x_{i4} $, $ x_{i5} $ & $ \sum_{j \neq i} \widetilde{\omega}_{ij}^{(2)} y_{j,t-1} $, $ y_{i,t-2} $, $ x_{i1} $, $ x_{i2} $, $ x_{i3} $, $ x_{i5} $    \\
	cm 6  & $ \sum_{j \neq i} \widetilde{\omega}_{ij}^{(3)} y_{j,t-1} $, $ y_{i,t-1} $, $ y_{i,t-2} $, $ x_{i1} $, $ x_{i2} $, $ x_{i3} $          & $ \sum_{j \neq i} \widetilde{\omega}_{ij}^{(2)} y_{j,t-1} $, $ y_{i,t-1} $, $ y_{i,t-2} $, $ x_{i1} $, $ x_{i2} $, $ x_{i3} $ \\
	cm 7  & $ \sum_{j \neq i} \widetilde{\omega}_{ij}^{(4)} y_{j,t-1} $, $ y_{i,t-1} $, $ x_{i1} $, $ x_{i2} $, $ x_{i3} $, $ x_{i4} $, $ x_{i5} $ & \multicolumn{1}{c}{/}                                                                                                         \\
	cm 8  & $ \sum_{j \neq i} \widetilde{\omega}_{ij}^{(4)} y_{j,t-1} $, $ y_{i,t-1} $, $ y_{i,t-2} $, $ x_{i1} $, $ x_{i2} $, $ x_{i3} $          & \multicolumn{1}{c}{/}                                                                                                         \\
	\bottomrule
\end{tabular}

\vspace*{5mm}
\caption*{\footnotesize Note: The response variable is $ y_{it} $. In DGP 1, $ \mathbf{W}^{(3)} $ and $ \mathbf{W}^{(4)} $ are irrelevant adjacency matrixes, and $ x_{i4} $ and $ x_{i5} $ are irrelevant variables. Candidate models 2 and 4 are correctly specified in the first and second half of the sample period, respectively, but misspecified in the other period. The other six candidate models are misspecified. In DGP 2, all the candidate models are misspecified over the sample period.}

\end{spacing}
\end{table}

\begin{table}[H]
\centering
\begin{spacing}{1.2}

\caption{\textbf{In-sample performance of simulations in DGP 1.}}
\label{Tab Simulations DGP 1_In sample performance}

\begin{tabular}{clcccccccccccc}
	\toprule
	&           &  & \multicolumn{2}{c}{RMSPE} &  & \multicolumn{8}{c}{Average weights of candidate models} \\ \hline
	&           &  & NAR         & tv-NAR      &  & cm 1  & cm 2  & cm 3 & cm 4 & cm 5 & cm 6 & cm 7 & cm 8 \\ \hline
	& $ \sT=60 $  &  & 0.745       & 1.005       &  & 0.00  & 0.46  & 0.00 & 0.49 & 0.00 & 0.03 & 0.00 & 0.02 \\
	$ \sN=10 $ & $ \sT=80 $  &  & 0.760       & 1.004       &  & 0.00  & 0.47  & 0.00 & 0.50 & 0.00 & 0.02 & 0.00 & 0.02 \\
	& $ \sT=100 $ &  & 0.760       & 1.003       &  & 0.00  & 0.48  & 0.00 & 0.49 & 0.00 & 0.02 & 0.00 & 0.01 \\ \hline
	& $ \sT=60 $  &  & 0.795       & 1.003       &  & 0.00  & 0.47  & 0.00 & 0.49 & 0.00 & 0.02 & 0.00 & 0.02 \\
	$ \sN=20 $ & $ \sT=80 $  &  & 0.804       & 1.002       &  & 0.00  & 0.48  & 0.00 & 0.49 & 0.00 & 0.01 & 0.00 & 0.02 \\
	& $ \sT=100 $ &  & 0.811       & 1.002       &  & 0.00  & 0.48  & 0.00 & 0.49 & 0.00 & 0.01 & 0.00 & 0.01 \\ \hline
	& $ \sT=60 $  &  & 0.813       & 1.002       &  & 0.00  & 0.48  & 0.00 & 0.49 & 0.00 & 0.02 & 0.00 & 0.02 \\
	$ \sN=30 $ & $ \sT=80 $  &  & 0.821       & 1.001       &  & 0.00  & 0.48  & 0.00 & 0.49 & 0.00 & 0.02 & 0.00 & 0.02 \\
	& $ \sT=100 $ &  & 0.825       & 1.001       &  & 0.00  & 0.48  & 0.00 & 0.49 & 0.00 & 0.01 & 0.00 & 0.02 \\
	\bottomrule
\end{tabular}

\vspace*{5mm}
\caption*{\footnotesize Note: This table reports the in-sample RMSPE ratio of the PTVMA relative to NAR and tv-NAR model and the average weight estimates of eight candidate models. The NAR and tv-NAR model involve $ \mathbf{W}^{(1)} $ and $ \mathbf{W}^{(2)} $ as network adjacency matrixes and the tv-NAR model is the correct specification.}

\end{spacing}
\end{table}

Figure \ref{Fig Simulations DGP 1_Time-varying average weights of cm1 and cm2} illustrates a central implication of the proposed PTVMA: when the strength of network dependence is time-varying, the optimal weights are expected to be time-varying as well. The solid lines report the averaged time-varying weights on candidate model 2 and 4, and the dashed lines plot the true coefficient functions: $ \beta_{t1} $ and $ \beta_{t2} $. The plots show that the assigned weights closely track the DGP 1: the PTVMA places most weight on candidate model 2 in the first half of the sample period and rapidly shifts weight toward candidate model 4 in the second half, consistent with the time-varying pattern of true coefficient functions. In addition, the weight reallocation is concentrated around the mid-sample transition where $ \beta_{t1} $ and $ \beta_{t2} $ switch most sharply, indicating PTVMA's strong adaptivity to smooth structural changes.

\smallskip

\begin{figure}[H]
\centering
\subfloat[$ \sN = 10, \sT = 60 $.]{\includegraphics[width=0.25\columnwidth]{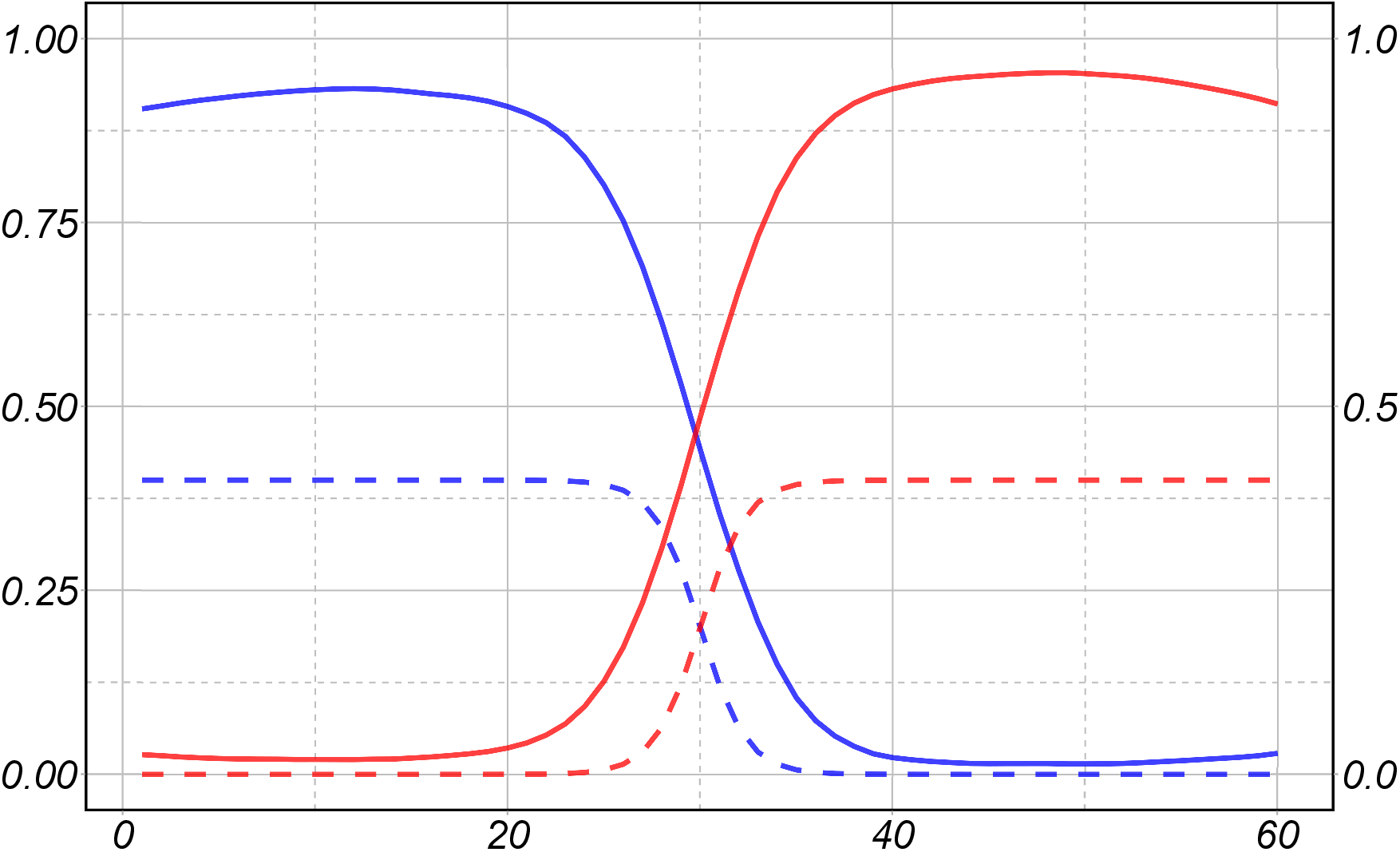}}
\hspace{1mm}
\subfloat[$ \sN = 10, \sT = 80 $.]{\includegraphics[width=0.25\columnwidth]{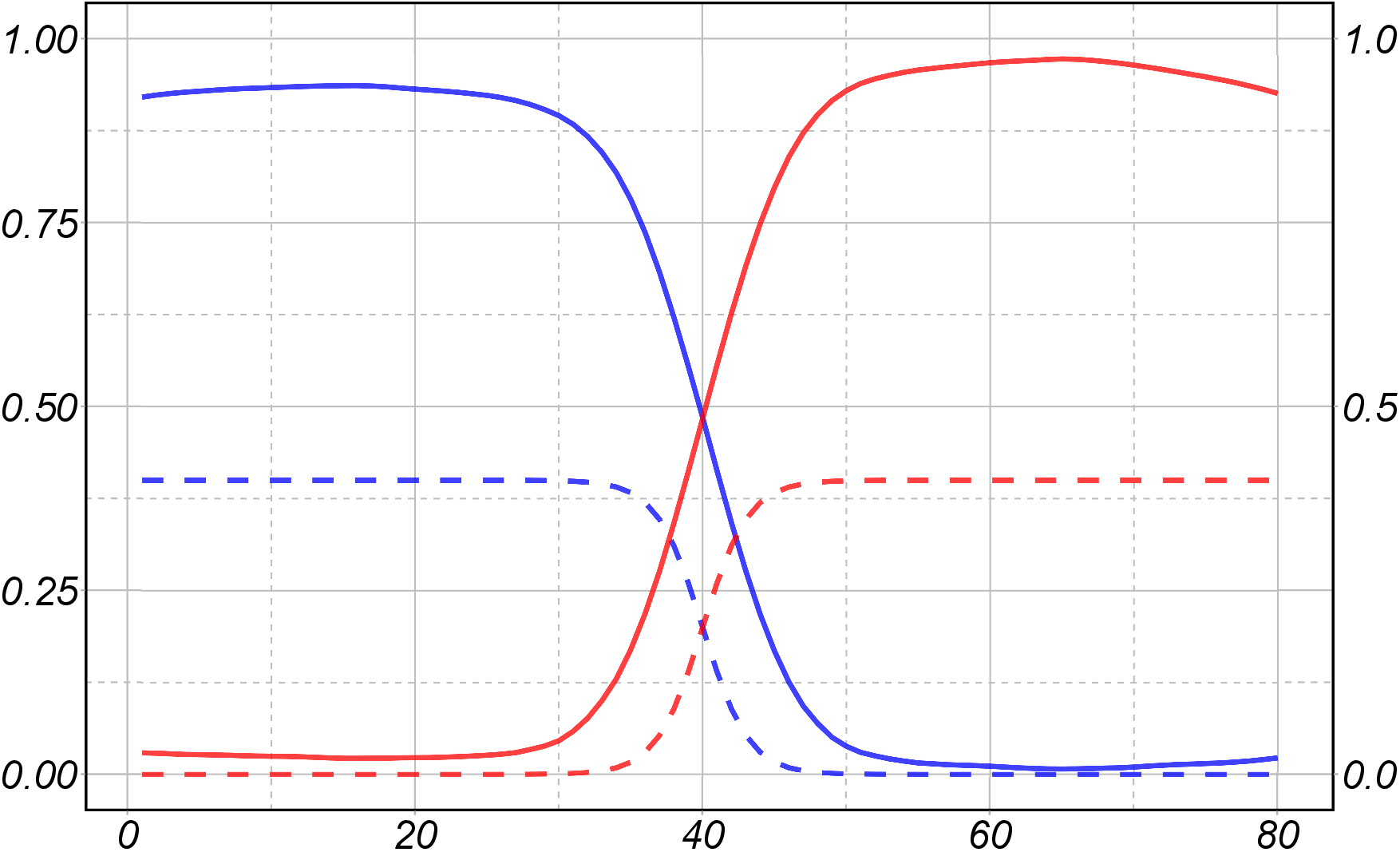}}
\hspace{1mm}
\subfloat[$ \sN = 10, \sT = 100 $.]{\includegraphics[width=0.25\columnwidth]{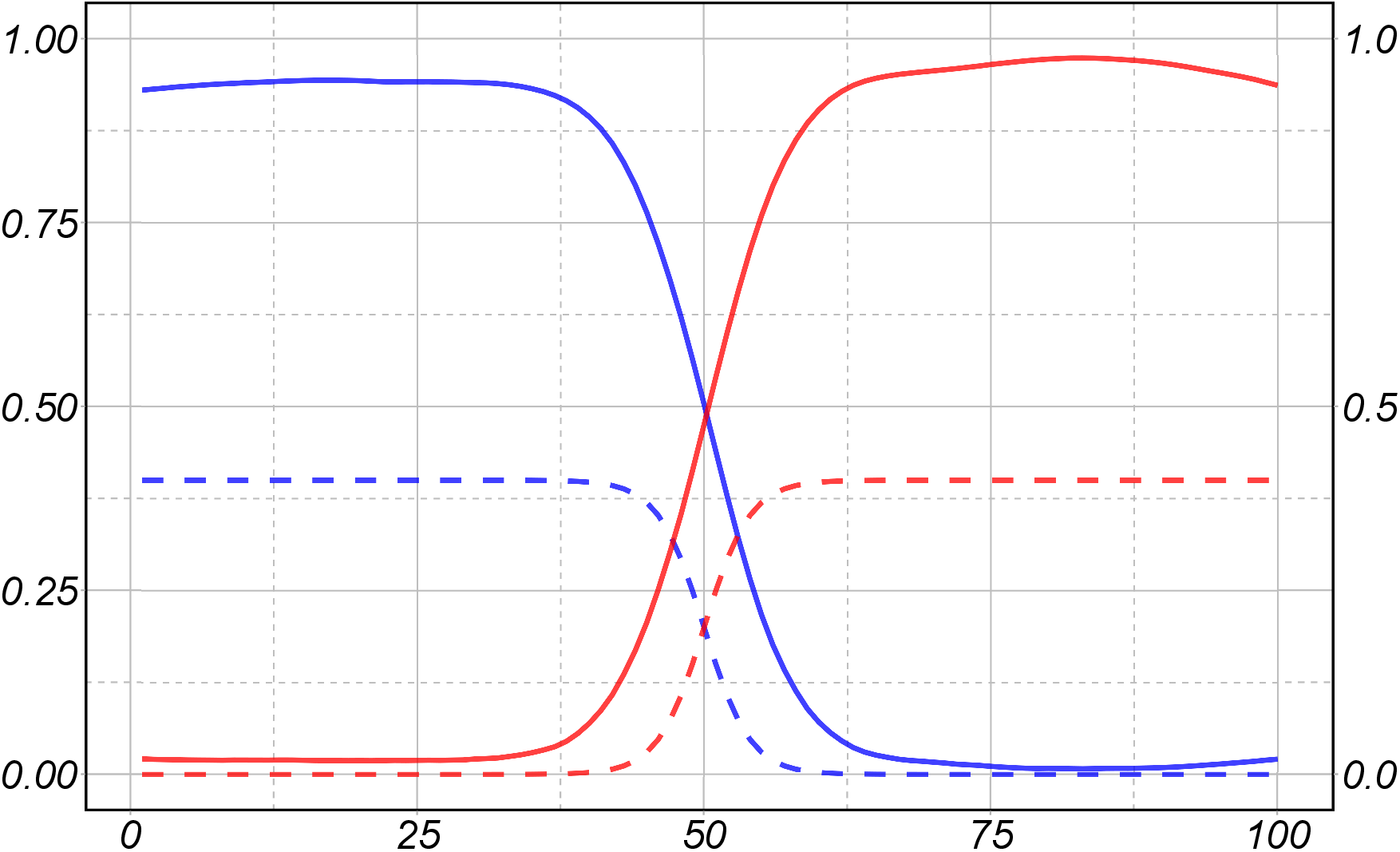}}   \\
\vspace{2mm}
\subfloat[$ \sN = 20, \sT = 60 $.]{\includegraphics[width=0.25\columnwidth]{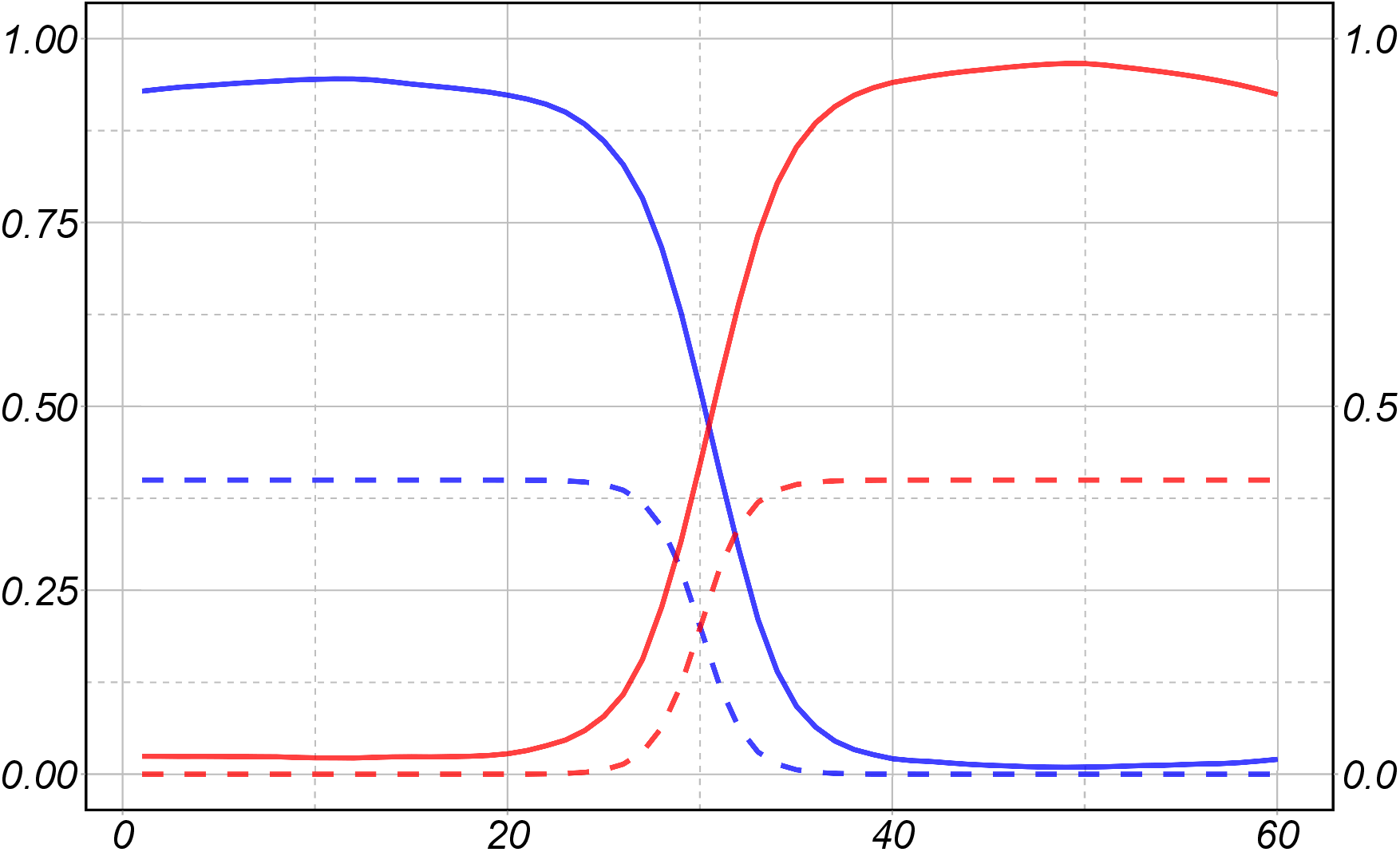}}
\hspace{1mm}
\subfloat[$ \sN = 20, \sT = 80 $.]{\includegraphics[width=0.25\columnwidth]{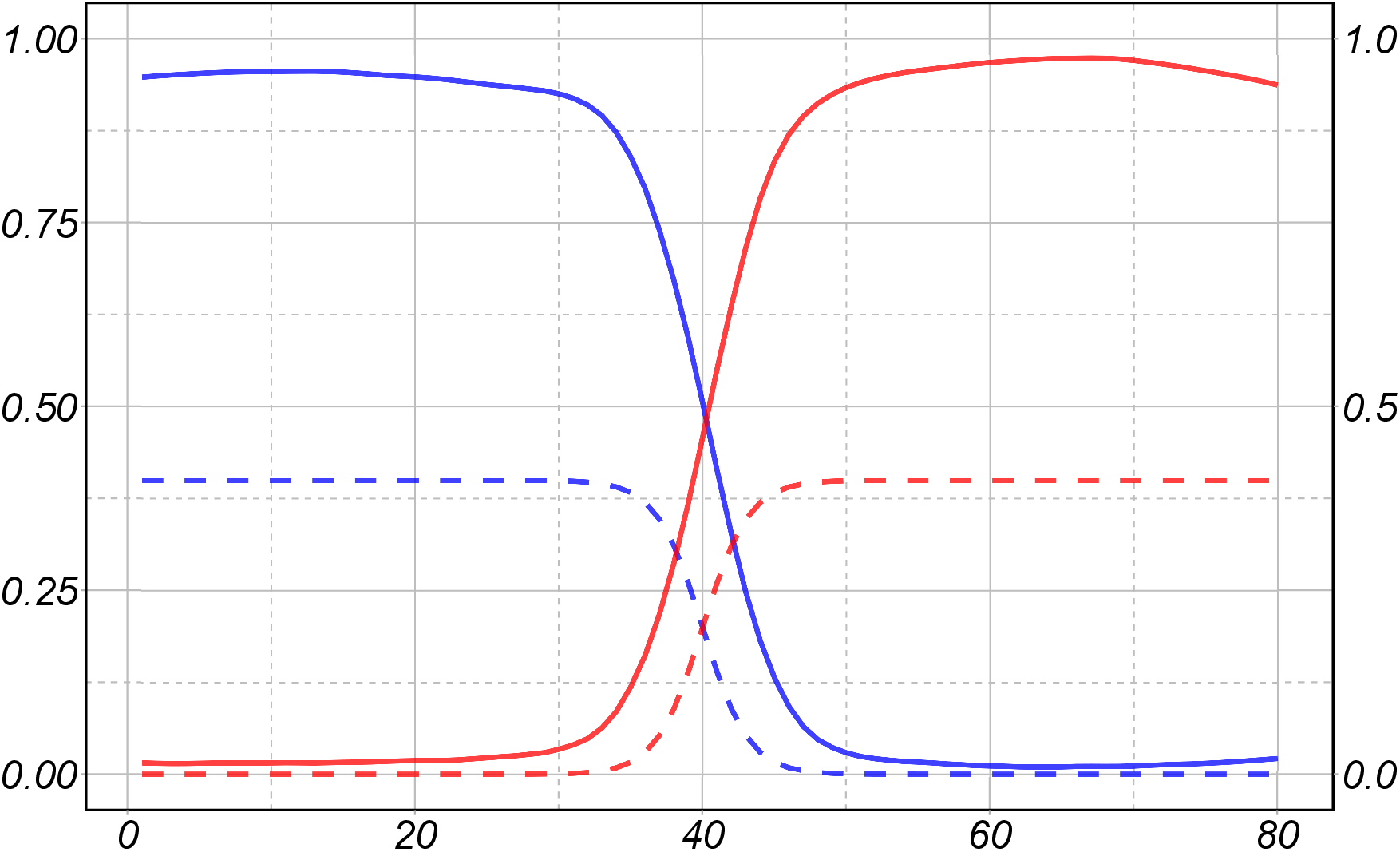}}
\hspace{1mm}
\subfloat[$ \sN = 20, \sT = 100 $.]{\includegraphics[width=0.25\columnwidth]{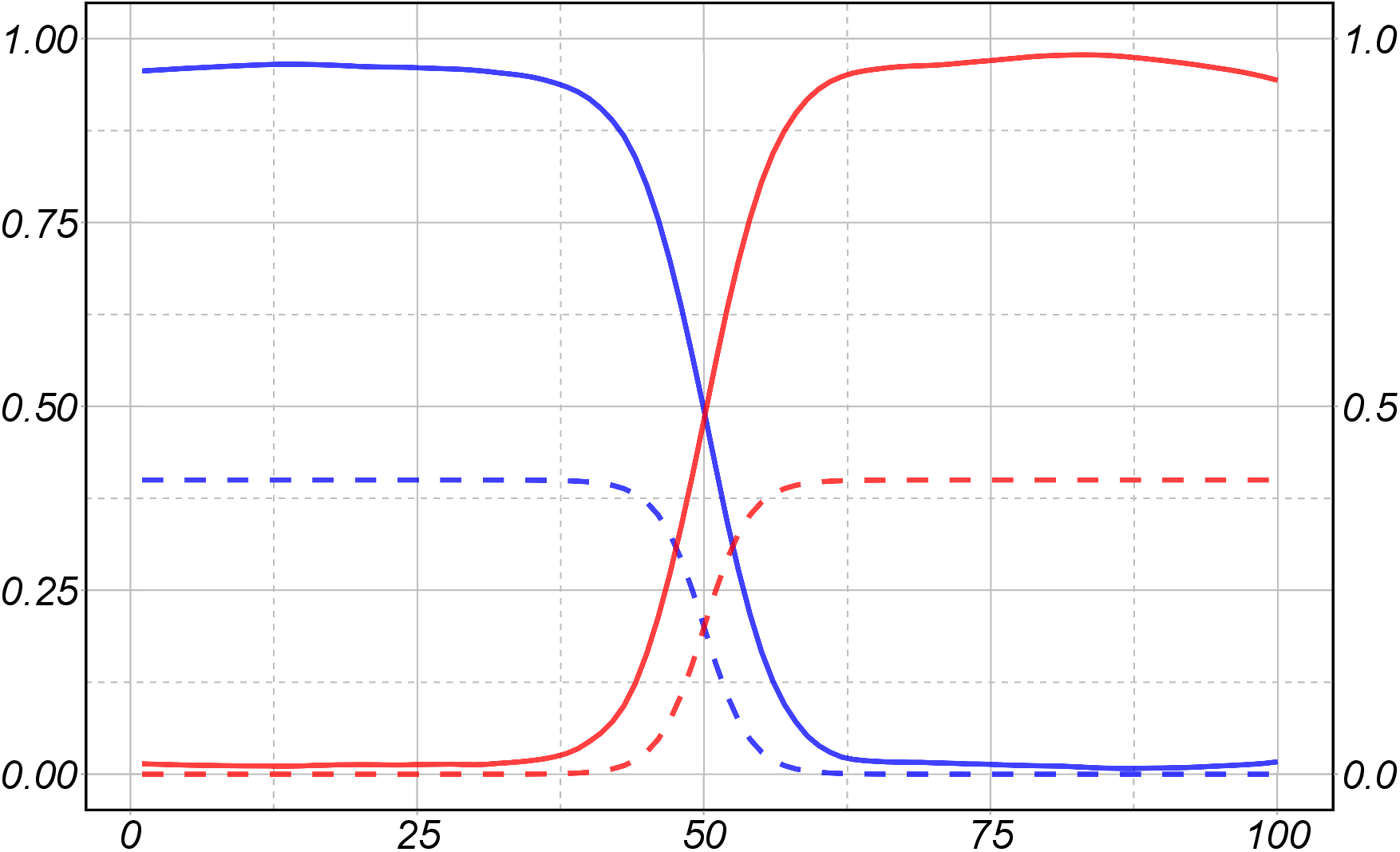}}   \\
\vspace{2mm}
\subfloat[$ \sN = 30, \sT = 60 $.]{\includegraphics[width=0.25\columnwidth]{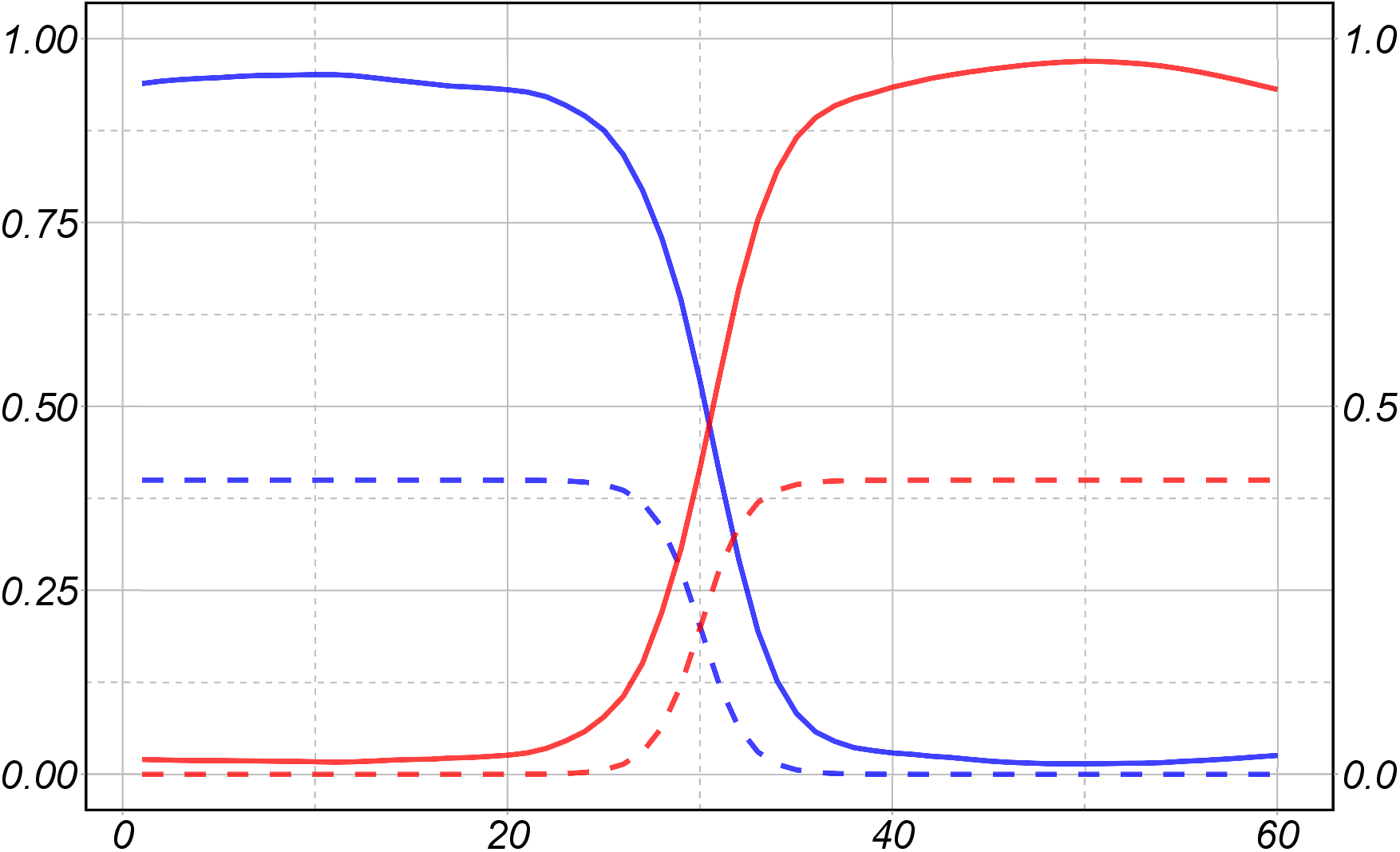}}
\hspace{1mm}
\subfloat[$ \sN = 30, \sT = 80 $.]{\includegraphics[width=0.25\columnwidth]{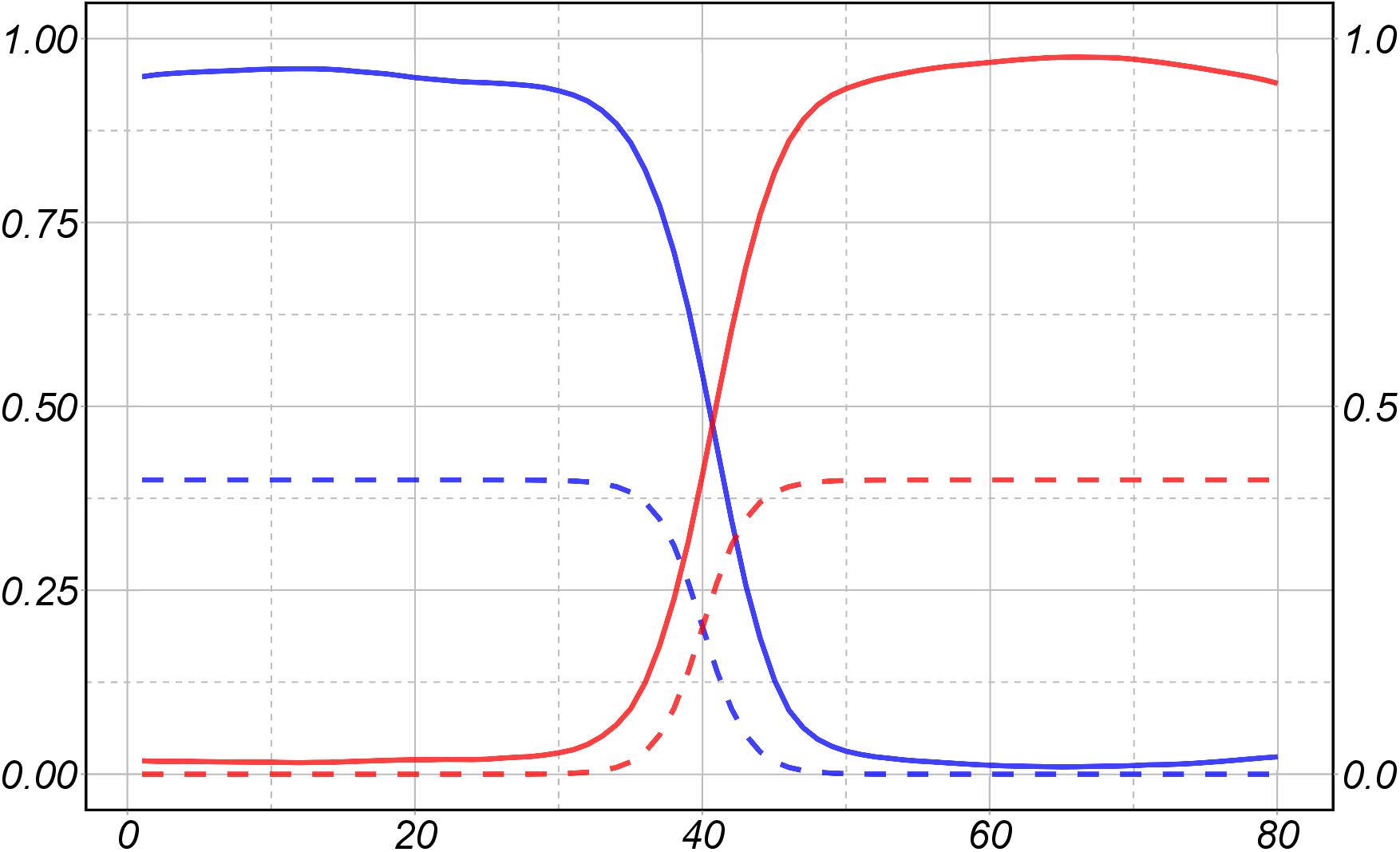}}
\hspace{1mm}
\subfloat[$ \sN = 30, \sT = 100 $.]{\includegraphics[width=0.25\columnwidth]{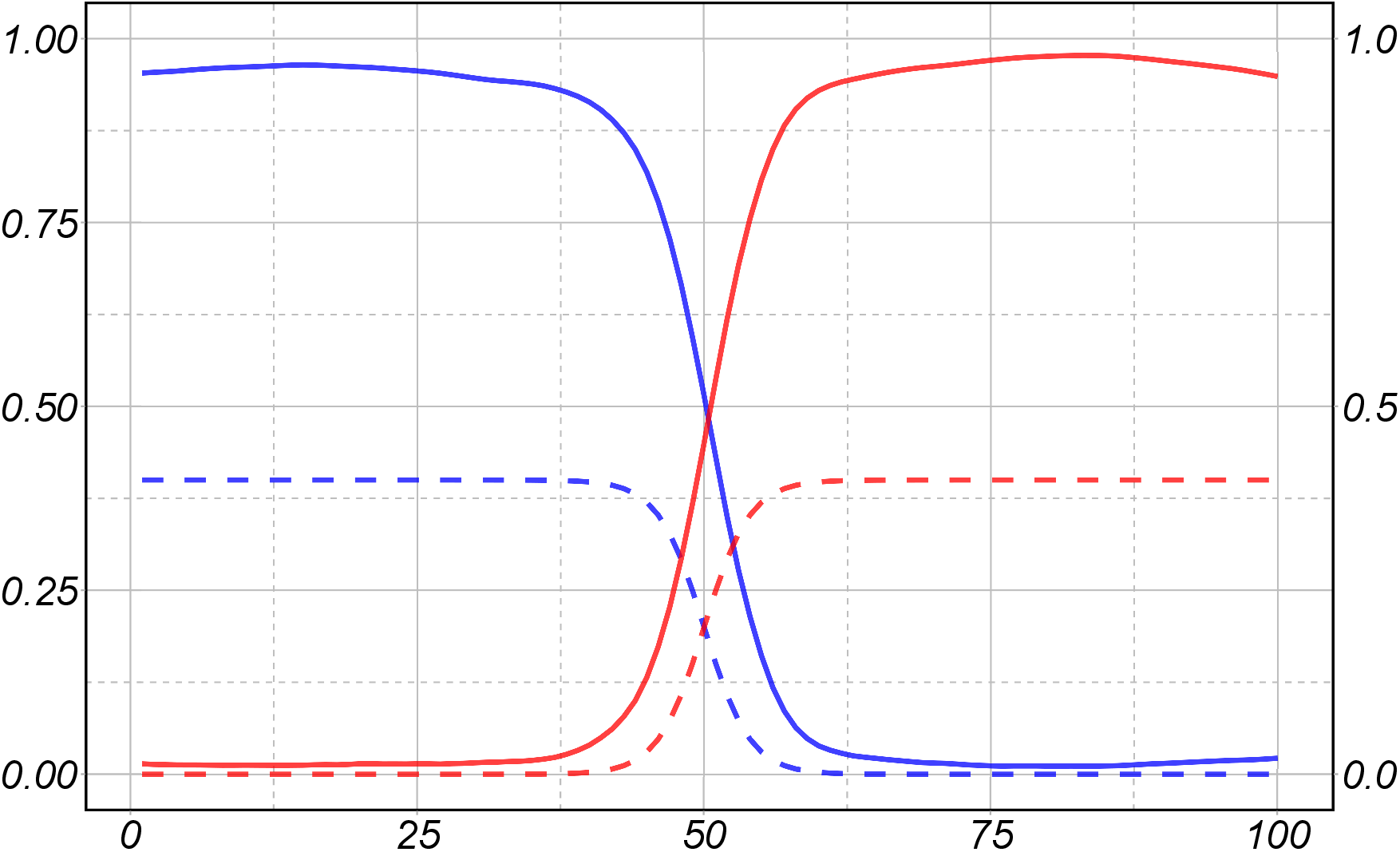}}

\caption{\textbf{Time-varying weight estimates of candidate model 2 and 4 in DGP 1.}}
\caption*{\footnotesize Note: The figures plot the average time-varying weights of candidate model 2 (solid blue line) and 4 (solid red line) across 1000 replications for different cases. The true functional patterns of $ \beta_{t1} $ and $ \beta_{t2} $ are displayed in dashed blue and red lines, respectively. The left axis corresponds to time-varying weight estimates, while the right axis corresponds to true time-varying coefficients.}
\label{Fig Simulations DGP 1_Time-varying average weights of cm1 and cm2}
\end{figure}

\smallskip

Table \ref{Tab Simulations DGP 1_Out sample performance} reports out-of-sample RMSPEs of the PTVMA relative to two benchmarks under DGP 1. Nearly all reported ratios fall below unity, confirming that the PTVMA uniformly dominates both the NAR and the tv-NAR models across all forecast horizons and sample sizes. In particular, the PTVMA outperforms the tv-NAR model even though the latter is correctly specified under DGP 1. The proposed PTVMA achieves the forecast gain by adaptively weighting candidate models as network dependence varies over time, and thus delivers higher forecast accuracy than a correctly specified benchmark model. The last column of Table \ref{Tab Simulations DGP 1_Out sample performance} reports the empirical coverage probabilities of conformal prediction intervals at 90\% nominal level. The coverage probabilities are close to 90\% in general, suggesting that the conformal intervals are well-calibrated in finite samples.

\begin{table}[H]
\centering
\begin{spacing}{1.2}

\caption{\textbf{Out-of-sample performance of simulations in DGP 1.}}
\label{Tab Simulations DGP 1_Out sample performance}

\begin{tabular}{clccccccccccc}
	\toprule
	\multicolumn{2}{c}{Prediction type} &  & \multicolumn{2}{c}{1-step-ahead} &  & \multicolumn{2}{c}{2-step-ahead} &  & \multicolumn{2}{c}{3-step-ahead} &  & \multicolumn{1}{l}{Coverage}       \\ \cline{1-11}
	\multicolumn{2}{c}{RMSPE}           &  & NAR            & tv-NAR          &  & NAR            & tv-NAR          &  & NAR            & tv-NAR          &  & \multicolumn{1}{l}{probability} \\ \hline
	& $ \sT=60 $         &  & 0.996          & 0.970           &  & 1.017          & 0.937           &  & 1.059          & 0.891           &  & 0.873                                    \\
	$ \sN=10 $         & $ \sT=80 $         &  & 0.961          & 0.973           &  & 0.960          & 0.948           &  & 0.973          & 0.910           &  & 0.894                                    \\
	& $ \sT=100 $        &  & 0.926          & 0.976           &  & 0.906          & 0.956           &  & 0.891          & 0.923           &  & 0.853                                    \\ \hline
	& $ \sT=60 $         &  & 0.937          & 0.984           &  & 0.927          & 0.970           &  & 0.911          & 0.946           &  & 0.900                                    \\
	$ \sN=20 $         & $ \sT=80 $         &  & 0.922          & 0.987           &  & 0.901          & 0.976           &  & 0.871          & 0.956           &  & 0.915                                    \\
	& $ \sT=100 $        &  & 0.905          & 0.988           &  & 0.874          & 0.979           &  & 0.826          & 0.961           &  & 0.863                                    \\ \hline
	& $ \sT=60 $         &  & 0.915          & 0.989           &  & 0.895          & 0.980           &  & 0.862          & 0.963           &  & 0.900                                    \\
	$ \sN=30 $         & $ \sT=80 $         &  & 0.903          & 0.991           &  & 0.876          & 0.983           &  & 0.831          & 0.968           &  & 0.918                                    \\
	& $ \sT=100 $        &  & 0.893          & 0.992           &  & 0.857          & 0.985           &  & 0.800          & 0.972           &  & 0.931                                    \\
	\bottomrule
\end{tabular}

\vspace*{5mm}
\caption*{\footnotesize Note: This table reports the out-of-sample RMSPEs of the PTVMA relative to two NAR and tv-NAR model for one-, two-, and three-step-ahead predictions. The last column reports the empirical coverage probabilities of conformal prediction intervals at 90\% nominal level. Under the iterated scheme, multi-step forecasts are obtained by estimating a one-step-ahead model and recursively replacing future lagged values with their forecasts.}

\end{spacing}
\end{table}

\subsection{All the candidate models are misspecified.}

In the DGP 2 below, we consider a more challenging setting with all the single-network specifications misspecified over any time period. Define
\begin{align}
y_{i,t} = \beta_{t1} \sum_{j \neq i} \widetilde{\omega}_{ij}^{(1)} y_{j,t-1} + \beta_{t2} \sum_{j \neq i} \widetilde{\omega}_{ij}^{(2)} y_{j,t-1} + \alpha_{t1} y_{i,t-1} + \alpha_{t2} y_{i,t-2} + \sum_{q=1}^{3} \gamma_{tq} x_{i,q} + \epsilon_{i,t},
\label{Eq Simulation DGP 2}
\end{align}
where the true time-varying coefficients are set as
\be
&& \alpha_{t1} = 0.5 \sin(\pi t/\sT) - 0.45,\ \ \alpha_{t2} = 0.5 \cos(\pi t/\sT) - 0.45,\notag\\
&&\beta_{t1} = -0.25 (t/\sT)^2 + 0.3 ,\ \ \beta_{t2} = 0.25 (t/\sT)^2 + 0.05,\notag
\ee
indicating that the impact of $ \mathbf{W}^{(1)} $ diminishes and that of $ \mathbf{W}^{(2)} $ grows over the time. The other model settings are the same as those in DGP 1. The predictors of candidate models for DGP 2 are listed in Table \ref{Tab Candidate models for DGP 1 and 2}.

\smallskip

The in-sample performance of the PTVMA under DGP 2 exhibits a similar pattern to that in DGP 1, demonstrating the PTVMA's advantages in the misspecification scenario. As reported in Table \ref{Tab Simulations DGP 2_In sample performance}, the PTVMA substantially outperforms the NAR model (RMSPE ranging from $ 0.550 $ to $ 0.570 $) while delivers a near-identical fit to the correctly specified tv-NAR model (RMSPE tightly around $ 1.001 $). This result confirms that even when all the candidate models are misspecified, the proposed PTVMA can adaptively combine them to approximate the response as effectively as a correctly-specified model.

\smallskip

\begin{table}[H]
\centering
\begin{spacing}{1.2}

\caption{\textbf{In-sample performance of simulations in DGP 2.}}
\label{Tab Simulations DGP 2_In sample performance}

\begin{tabular}{cccccccccccc}
	\toprule
	&          &  & \multicolumn{2}{c}{RMSPE} &  & \multicolumn{6}{c}{Average weights of candidate models} \\ \hline
	&          &  & NAR         & tv-NAR      &  & cm 1    & cm 2    & cm 3    & cm 4    & cm 5   & cm 6   \\ \hline
	& $ \sT=30 $ &  & 0.565       & 1.001       &  & 0.08    & 0.10    & 0.39    & 0.03    & 0.08   & 0.32   \\
	$ \sN=30 $ & $ \sT=40 $ &  & 0.556       & 1.001       &  & 0.08    & 0.11    & 0.41    & 0.03    & 0.07   & 0.30   \\
	& $ \sT=50 $ &  & 0.550       & 1.001       &  & 0.09    & 0.10    & 0.41    & 0.02    & 0.06   & 0.32   \\ \hline
	& $ \sT=30 $ &  & 0.567       & 1.000       &  & 0.08    & 0.11    & 0.40    & 0.03    & 0.08   & 0.30   \\
	$ \sN=40 $ & $ \sT=40 $ &  & 0.557       & 1.001       &  & 0.08    & 0.10    & 0.40    & 0.03    & 0.07   & 0.32   \\
	& $ \sT=50 $ &  & 0.550       & 1.001       &  & 0.08    & 0.10    & 0.42    & 0.02    & 0.06   & 0.31   \\ \hline
	& $ \sT=30 $ &  & 0.570       & 1.000       &  & 0.07    & 0.10    & 0.40    & 0.03    & 0.08   & 0.32   \\
	$ \sN=50 $ & $ \sT=40 $ &  & 0.559       & 1.001       &  & 0.08    & 0.10    & 0.41    & 0.02    & 0.06   & 0.32   \\
	& $ \sT=50 $ &  & 0.551       & 1.001       &  & 0.09    & 0.09    & 0.43    & 0.02    & 0.06   & 0.32   \\
	\bottomrule
\end{tabular}


\vspace{5mm}
\caption*{\footnotesize Note: This table reports the in-sample RMSPEs of the PTVMA relative to NAR and tv-NAR model and averaged weight estimates of six candidate models.}

\end{spacing}
\end{table}

\smallskip

Table \ref{Tab Simulations DGP 2_In sample performance} also reports the averaged weight estimates of six candidate models over $1000$ replications to further reveal PTVMA's weight-assignment capacity. Due to incorporating both autoregressive lags, candidate models 3 and 6 receive the majority of weight with approximately $ 0.41 $ and $ 0.31 $, respectively. Note that the other candidate models receiving negligible weights either omit lags or contain irrelevant predictors. Consequently, these specifications are less informative than candidate models 3 and 6 under the DGP 2. The results confirm that the proposed PTVMA successfully allocates significant weights to the candidate models with most relevant predictors whereas effectively shrinks the other irrelevant specifications.

\smallskip

Table \ref{Tab Simulations DGP 2_Out sample performance} presents out-of-sample RMSPE for both the one- and multiple-step-ahead forecasts. Across all horizons and sample sizes, the PTVMA consistently outperforms two benchmarks and achieves lower forecast errors than the correctly specified model (tv-NAR), again demonstrating that the proposed model averaging procedure delivers forecast gains by adaptively combining information from candidate models even when all of them are misspecified.

\smallskip

\begin{table}[H]
\centering
\begin{spacing}{1.2}
\caption{\textbf{Out-of-sample performance of simulations in DGP 2.}}
\label{Tab Simulations DGP 2_Out sample performance}
\begin{tabular}{ccccccccccc}
	\toprule
	\multicolumn{2}{c}{Prediction type} &  & \multicolumn{2}{c}{1-step-ahead} &  & \multicolumn{2}{c}{2-step-ahead} &  & \multicolumn{2}{c}{3-step-ahead} \\ \hline
	\multicolumn{2}{c}{RMSPE}           &  & NAR            & tv-NAR          &  & NAR            & tv-NAR          &  & NAR            & tv-NAR          \\ \hline
	& $ \sT=30 $         &  & 0.709          & 0.985           &  & 0.689          & 0.984           &  & 0.770          & 0.969           \\
	$ \sN=30 $         & $ \sT=40 $         &  & 0.673          & 0.987           &  & 0.645          & 0.985           &  & 0.727          & 0.975           \\
	& $ \sT=50 $         &  & 0.652          & 0.988           &  & 0.621          & 0.986           &  & 0.707          & 0.976           \\ \hline
	& $ \sT=30 $         &  & 0.689          & 0.989           &  & 0.670          & 0.990           &  & 0.755          & 0.983           \\
	$ \sN=40 $         & $ \sT=40 $         &  & 0.658          & 0.990           &  & 0.630          & 0.989           &  & 0.716          & 0.983           \\
	& $ \sT=50 $         &  & 0.641          & 0.992           &  & 0.608          & 0.990           &  & 0.696          & 0.983           \\ \hline
	& $ \sT=30 $         &  & 0.678          & 0.992           &  & 0.658          & 0.992           &  & 0.746          & 0.986           \\
	$ \sN=50 $         & $ \sT=40 $         &  & 0.651          & 0.994           &  & 0.624          & 0.993           &  & 0.711          & 0.990           \\
	& $ \sT=50 $         &  & 0.634          & 0.993           &  & 0.603          & 0.991           &  & 0.690          & 0.985           \\
	\bottomrule
\end{tabular}
\vspace{5mm}
\caption*{\footnotesize Note: This table reports the out-of-sample RMSPEs of the proposed model relative to NAR and tv-NAR model for one-, two-, and three-step-ahead predictions. Under the iterated scheme, multi-step forecasts are obtained by estimating a one-step-ahead model and recursively replacing future lagged values with their forecasts.}

\end{spacing}
\end{table}


\section{An empirical application}\label{sec6}
\renewcommand{\theequation}{6.\arabic{equation}} \setcounter{equation}{0}

Accurate forecasting of Consumer Price Index (CPI) inflation plays a critical role in the conduct of monetary policy and preservation of financial stability. Inflation expectations serve as a key anchor for policy credibility and exert a substantial influence on asset pricing. Nevertheless, forecasting inflation remains a persistent challenge, even when employing comprehensive datasets and state-of-the-art econometric methods \citep{SW07}. These difficulties are compounded in open economies, where domestic inflation is increasingly shaped by global shocks and the cross-country transmission of price pressures. A growing body of literature demonstrates that inflation dynamics are increasingly dominated by global determinants, even as domestic conditions remain relevant \citep{CM10}. Key international transmission channels include global production networks, international financial markets, trade and policy shocks \citep{CNST21, MR20, ARW19}. These channels subsequently form multi-layer networks that evolve over time, creating substantial model uncertainty and structural instability. Traditional univariate or panel models often fail to capture complex cross-country spillovers and time-varying nature of inflation transmission, leading to potentially misspecified models and degraded forecast performance. This section aims to address this issue by combining information from multiple international networks and applying the proposed PTVMA method to forecast CPI inflation,

\subsection{Data}

The CPI inflation data is obtained from the International Monetary Fund (IMF)\footnote{The data is downloaded from the IMF Data, available at: \url{https://data.imf.org/en}.}. We consider a dataset on monthly CPI inflation in 36 economies ($ \sN = 36 $) from Januaray 2007 to December 2024 with 216 observations ($ \sT = 216 $). This economy set covers the largest and most systemically relevant economies with established equity markets, ensuring representation of advanced economies, major emerging markets, and key commodity producers that jointly account for the overwhelming majority of global output.\footnote{Selected economies include Belgium, Brazil, Canada, Switzerland, Chile, China, Colombia, the Czech Republic, Germany, Denmark, Spain, Finland, France, the United Kingdom, Greece, Hong Kong, Hungary, Indonesia, India, Ireland, Israel, Italy, Japan, South Korea, Mexico, Malaysia, the Netherlands, Norway, Poland, Portugal, Singapore, Sweden, Turkey, the United States, Vietnam, and South Africa. We exclude nine candidate economies due to data limitations during the sample period: Russia and Saudi Arabia are omitted due to missing daily MSCI indices, while the United Arab Emirates, Argentina, Australia, New Zealand, Peru, the Philippines, and Thailand are excluded due to the unavailability of monthly inflation data.} We include two variables as controls in the model: a time-varying trend term, capturing exogenous shocks such as oil price fluctuations; and the monthly period-average nominal exchange rate for economy $ i $ (domestic currency per U.S. dollar), sourced from the IMF Data.

\smallskip

We consider four types of networks to characterize the international transmission channels of CPI inflation, each reflecting a distinct economic mechanism. The economic rationale for each network and the corresponding data sources are discussed below, with detailed network construction procedures presented in Appendix \ref{app:B}.

\smallskip

\begin{itemize}

\item \textbf{Global production networks.} Global production networks transmit cost shocks across borders through input-output linkages, generating price co-movement that extends beyond direct bilateral trade \citep{ALS19, CNST21}. Since the CPI measures prices of final goods and services consumed by households, we construct the production network layer by using the origin of value added embodied in final demand.\footnote{The data is downloaded from the Trade in Value-Added (TiVA) database, available at: \url{https://www.oecd.org/en/topics/sub-issues/trade-in-value-added.html}.}
Within this framework, an inflation increase in economy $ i $ raises the cost of its value added, which subsequently transmits through global value chains and becomes embodied in final goods consumed in economy $ j $, thereby exerting upward pressure on economy $ j $'s CPI.

\item \textbf{Global equity networks.} The integration of national equity markets affects cross-country inflation dynamics through information diffusion and the formation of inflation expectations \citep{MR20, CM10, CCL09}. As institutional ownership and algorithmic trading become more prevalent globally, cross-market equity linkages generate forward-looking signals about macroeconomic fundamentals. Equity-market connectedness thus serves as a leading transmission channel through which global financial conditions shape inflation dynamics across countries, often ahead of real-sector and policy channels. We construct the global equity network using daily MSCI country-level equity indices (standard, price-return, U.S. dollar–denominated) and the VAR decomposition framework of \cite{DY14}.\footnote{Data are obtained from the MSCI database. \url{https://www.msci.com/indexes\#featured-indexes}.} These MSCI indices provide comparable, internationally tradable measures of equity price movements and capture high-frequency financial channels relevant for inflation transmission.

\item \textbf{Trade and policy networks.} Trade policy acts as a distinct channel of international inflation transmission by affecting import prices and the cost of intermediate inputs. Evidence from the 2018–2019 U.S.–China trade conflicts indicates that tariff costs pass through to consumer prices and propagate along input–output linkages, generating inflationary effects that extend beyond directly targeted goods \citep{FGKK20, ARW19}. To disentangle market-based trade integration from policy-driven exposure, we construct two types of network layers capturing trade linkages and political connectedness, respectively. The bilateral trade linkages in the trade network layer are measured by monthly FOB (free on board) export values (in U.S. dollars) sourced from the IMF Direction of Trade Statistics (DOTS).\footnote{Data are obtained from the IMF DOTS database (\url{https://data.imf.org/en/datasets/IMF.STA:IMTS}).} We also quantify bilateral political connectedness using ideal-point estimates derived from United Nations General Assembly (UNGA) roll-call votes \citep{BSV17}, following standard practice in the political economy literature.

\end{itemize}

\subsection{Empirical analysis}

Consider the $m$-th candidate model specification:
\begin{align}
y_{i,t}  = \beta_{t1}^{(m)} \sum_{j \neq i} \widetilde{\omega}_{ij}^{(m)} y_{j,t-1} + \alpha_{t1}^{(m)} y_{i,t-1} + \gamma_{t1}^{(m)} + \gamma_{t2}^{(m)} {\sf ER}_{i,t} + \epsilon_{i,t}^{(m)},
\label{Eq Empirical candidate model specification}
\end{align}
where $y_{i,t} $ denotes the CPI inflation of economy $ i $ in month $ t $, $ {\sf ER}_{i,t} $ is the period-average nominal exchange rate of economy $ i $ in month $ t $, $ \gamma_{t1}^{(m)} $ is the time-varying trend term, and $ \widetilde{\omega}_{ij}^{(m)} $ is the $ (i,j) $-th element of the matrix constructed from the network layer $ \mathbf{W}^{(m)} $ by setting its diagonal elements to zero and applying row-normalization. We consider $ 4 $ and $ 10 $ network layers for in- and out-of-sample analysis, respectively. Details on network constructions are presented in Appendix \ref{app:B}.

\smallskip

Figure \ref{Fig Empirical CPI inflation_Time-varying candidate weights} plots the estimated time-varying weights of four candidate models over 2007--2024. The weights of candidate models exhibit pronounced time variation. The production network layer is dominant in the early part of the time period, especially around 2009--2012, and then rapidly loses most weight in the later period. The policy network layer becomes prominent in 2013--2019, while the global equity network layer receives substantial weight in the late period (after 2019). In contrast, the trade network layer is assigned only intermittent weight and never attains persistent dominance. Figure \ref{Fig Empirical CPI inflation_Time-varying candidate weights} suggests that the transmission mechanism of CPI inflation shifts across different types of network dependence and varies along the time. This finding is consistent with multiple transmission channels for CPI inflation, with no single network maintaining uniformly dominant weights. A single-layer specification would thus be too restrictive when another layer better matches the network dependence structure. It also shows that the candidate weights concentrate on one or two layers at a time and shift as the underlying structure evolves, indicating that the PTVMA adaptively captures locally optimal specifications which may vary over the temporal dimension.

\smallskip

\begin{figure}[H]
	\centering
	\subfloat[Production network.]{\includegraphics[width=0.48\columnwidth]{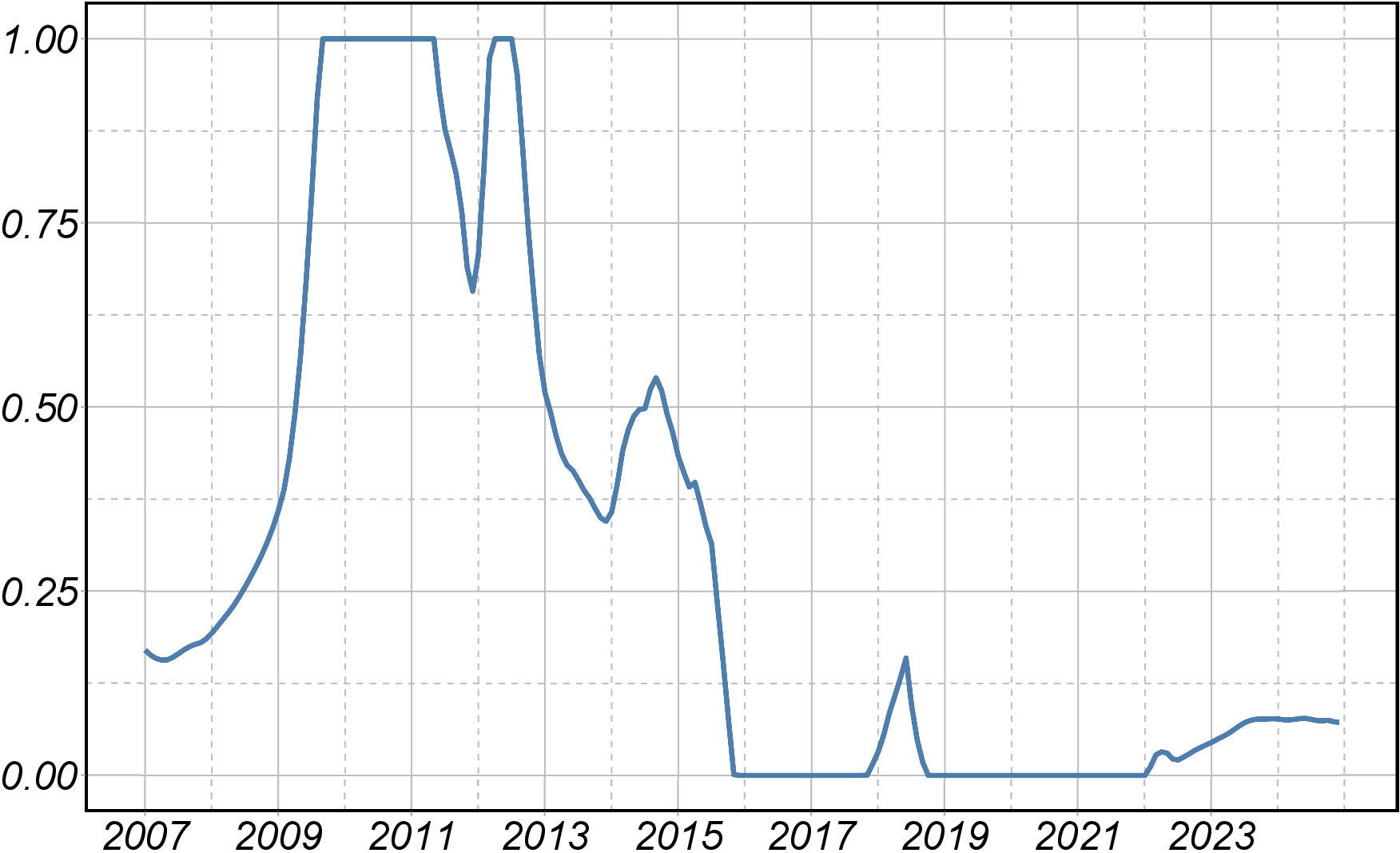}}
	\hspace{1mm}
	\subfloat[Global equity network.]{\includegraphics[width=0.48\columnwidth]{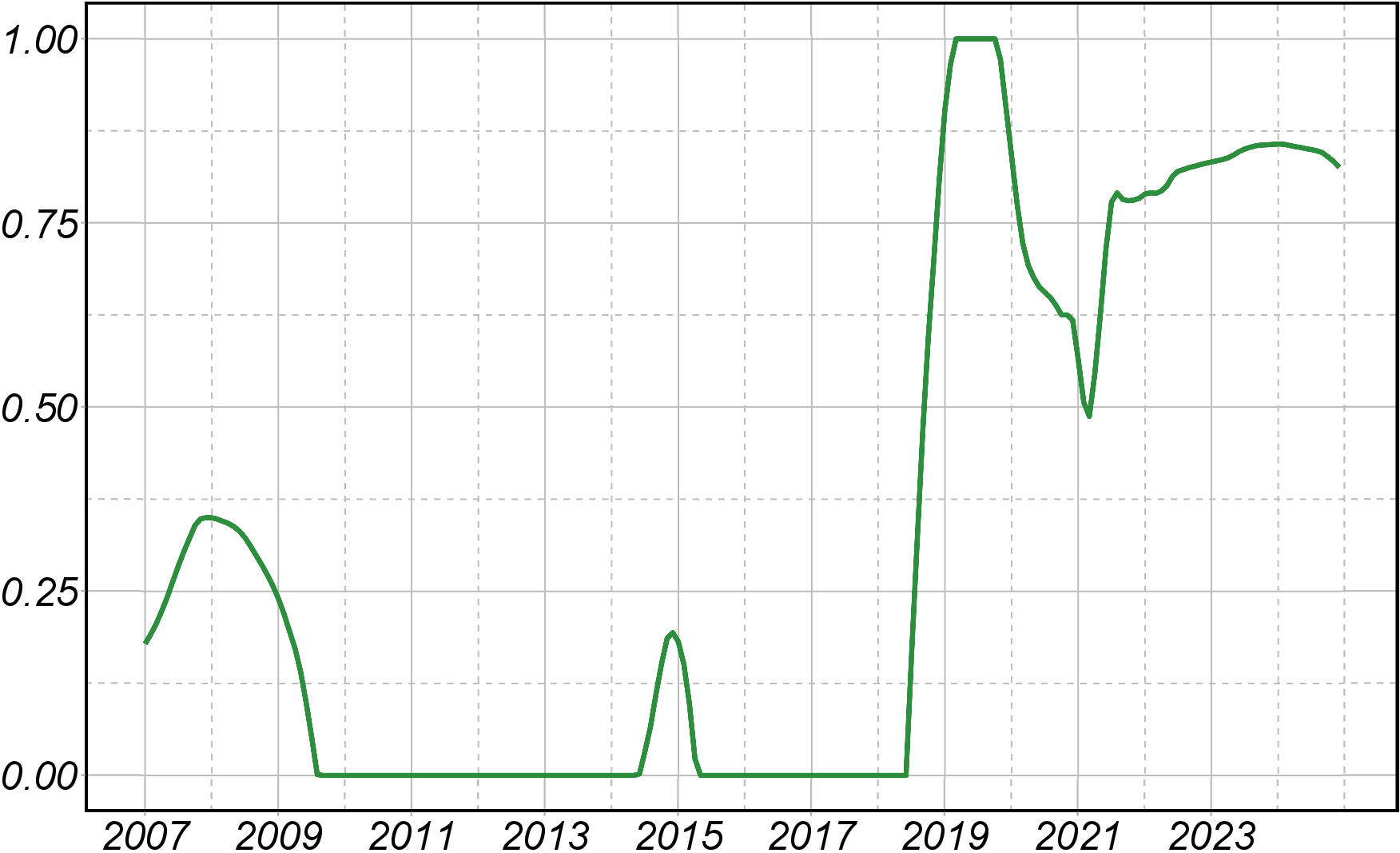}}   \\
	\vspace{2mm}
	\subfloat[Trade network.]{\includegraphics[width=0.48\columnwidth]{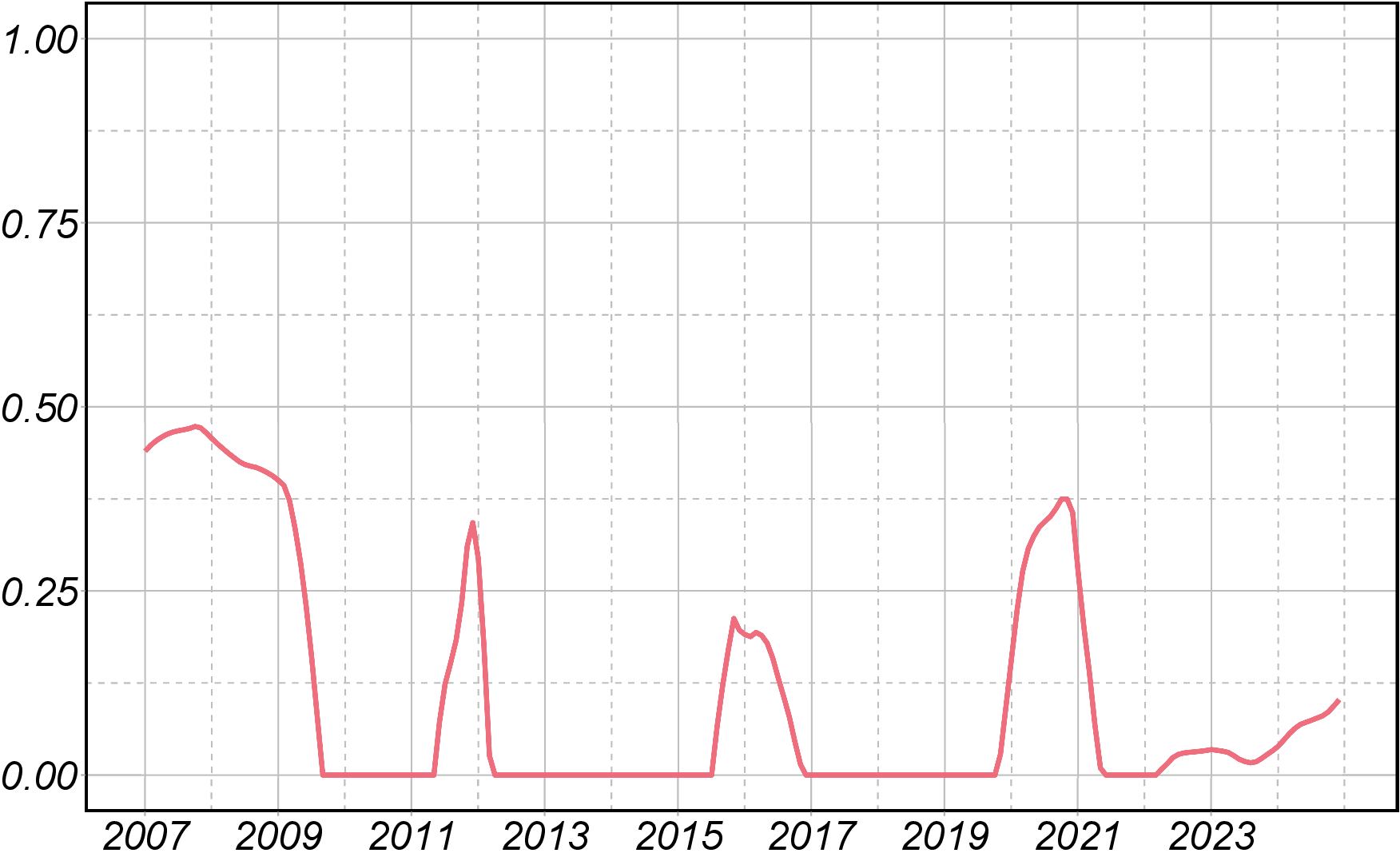}}
	\hspace{1mm}
	\subfloat[Policy network.]{\includegraphics[width=0.48\columnwidth]{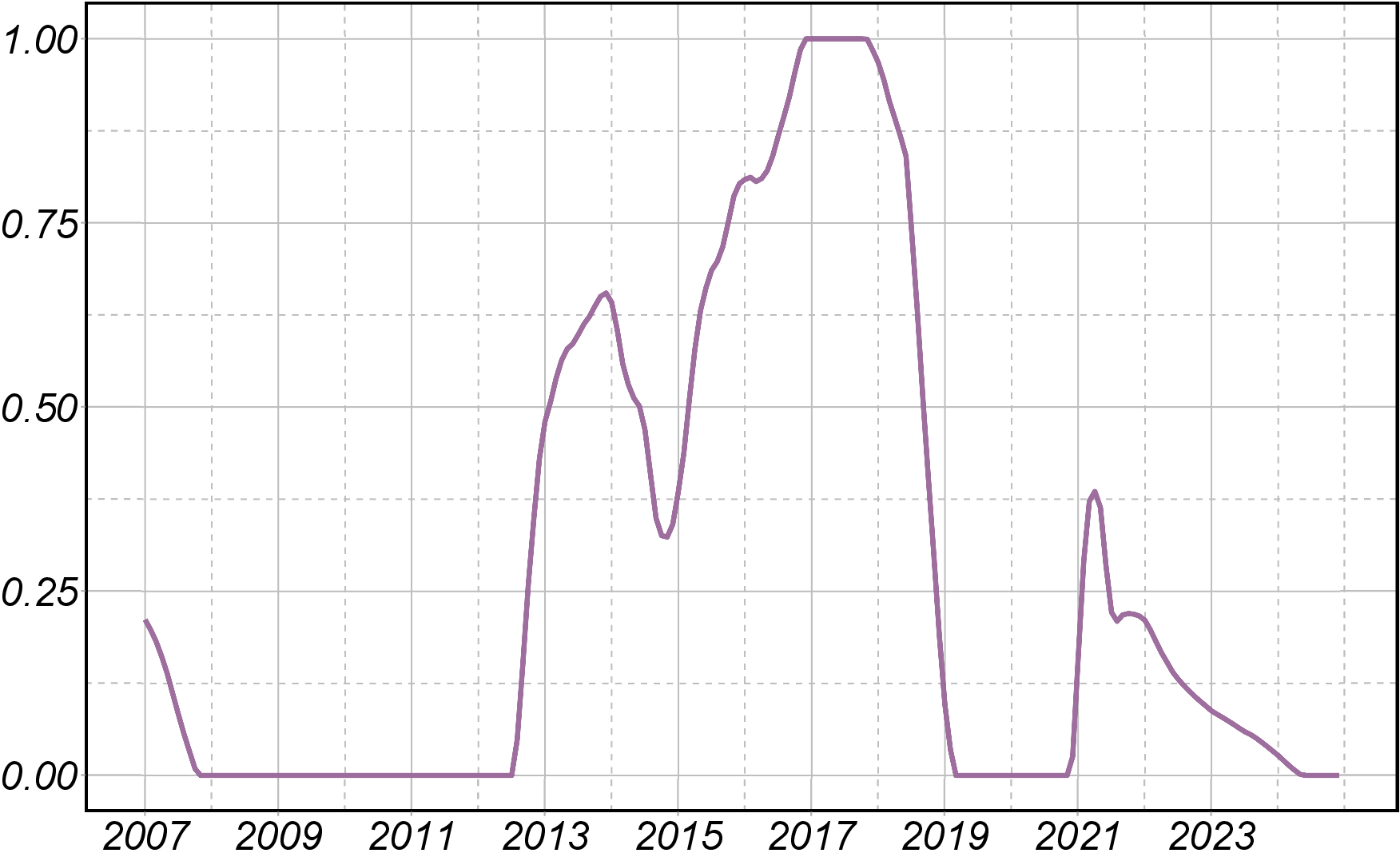}}

	\caption{\textbf{Time-varying weight estimates of 4 candidate models.}}
	\caption*{\footnotesize Note: The figure reports the estimated time-varying weights of 4 candidate models over 2007--2024.}
	\label{Fig Empirical CPI inflation_Time-varying candidate weights}
\end{figure}

\smallskip

Overall, the in-sample results provide evidence that global inflation transmission is governed by a multi-layer, regime-dependent structure in which production, financial, trade, and policy linkages alternately dominate. Therefore, conventional time series models with a single network structure or time-invariant parameters fail to capture the shifting dominance of these transmission channels, underscoring the value of the proposed PTVMA.

\smallskip

Table \ref{Tab Application RMSPE out sample} reports the out-of-sample prediction performance of the proposed model for CPI inflation. Predictive accuracy is evaluated by the RMSPE, and our model is compared with seven alternative specifications: NAR, time-varying NAR (tv-NAR), VAR(1) estimated by the OLS (VAR), VARX(1) estimated by the Elastic Net (HD-VARX), adaptive LASSO (ALASSO), random forest (RF), and stochastic gradient boosting (SGB), spanning one-, two-, and three-step-ahead forecasts across VAR, penalized regression, and state-of-the-art machine learning methods. Exogenous variables are the nominal exchange rates for all economies. In addition, the ALASSO, the RF and the SGB also involves in- and out-degrees of networks from $ \mathbf{W}^{(1)} $ to $ \mathbf{W}^{(8)} $ and degrees of networks $ \mathbf{W}^{(9)} $ and $ \mathbf{W}^{(10)} $ as controls, see Appendix \ref{app:B} for details.

\smallskip

\begin{table}[H]
\centering
\begin{spacing}{1.2}
\caption{\textbf{Out-of-sample performance of CPI inflation.}}
\label{Tab Application RMSPE out sample}
\begin{tabular}{cccccccc}
	\toprule
	RMSPE        & NAR   & tv-NAR & VAR   & HD-VARX & ALASSO & RF    & SGB   \\ \hline
	1-step-ahead & 0.978 & 0.973  & 0.650 & 0.281   & 0.997  & 0.717 & 0.615 \\
	2-step-ahead & 1.049 & 0.956  & 0.678 & 0.476   & 1.084  & 0.913 & 0.597 \\
	3-step-ahead & 1.137 & 0.939  & 0.722 & 0.647   & 1.197  & 1.075 & 0.681 \\
	\bottomrule
\end{tabular}
\vspace*{5mm}
\caption*{\footnotesize Note: This table reports the out-of-sample RMSPE of the PTVMA relative to seven competing models for one-, two-, and three-step-ahead predictions, including NAR, tv-NAR, VAR(1) model estimated by the OLS (VAR), VARX(1) model with exchange rates estimated by the Elastic Net (HD-VARX), adaptive LASSO (ALASSO), random forest (RF), and stochastic gradient boost (SGB).}

\end{spacing}
\end{table}

Table \ref{Tab Application RMSPE out sample} shows a pronounced advantage of the PTVMA in short-run CPI inflation forecasting. At the one-step horizon, our method uniformly outperforms all competitors and attains the lowest RMSPE (with all the ratios strictly smaller than one), indicating a marked improvement in short-run predictive accuracy. At the two-step horizon, the proposed PTVMA performs slightly worse than the NAR and ALASSO but continues to outperform the remaining methods. At the three-step horizon, however, the PTVMA only performs better than the tv-NAR and two VAR-based methods, suggesting that the short-horizon advantage attenuates as the forecast horizon expands. This is because the time-varying coefficient specification amplifies the negative impact of cumulative errors on multi-step-ahead forecasting. It is also worth noting that the proposed PTVMA still dominates the tv-NAR across three horizons, highlighting the advantage of forecast combination via time-varying model averaging.


\section{Conclusion}\label{sec7}

This paper introduces a dynamic model averaging criterion within a general time-varying network VAR framework which incorporates multiple adjacency matrices to capture network spillover effects. The main aim is to determine a time-varying optimal combination of multi-layer network VAR candidate models whose number may be divergent. In particular, all the candidate models with different linkage structures and varying lagged orders of dependent variables and node-specific predictors are allowed to be misspecified. Both the in-sample and out-of-sample asymptotic properties including the asymptotic optimality and estimation convergence rates are derived under some technical conditions. In particular, the proposed PTVMA achieves the selection consistency (of correctly-specified models) as the conventional model selection criterion. Furthermore, an extended conformal prediction method is proposed to construct prediction intervals for network time series. The Monte Carlo simulation demonstrates that the proposed PTVMA delivers superior out-of-sample predictive performance relative to the benchmark, even when all the candidate models are misspecified. The empirical case study further reveals the time-varying patterns of optimal weight estimation and shows that our model and method consistently outperform other competing models in the short-run CPI inflation prediction.

\bigskip
\bigskip
\bigskip



\newpage


\appendix

\begin{center}
{\Large Supplementary Material for ``Time-Varying Model Averaging of \\[1.2ex] Multi-layer Network Vector Autoregressions''}

\bigskip

{\small \textsc{Degui Li}$^{*}$, \textsc{Yuying Sun}$^{\dag}$, \textsc{Boyao Wu}$^{\ddag}$ }

{\small \medskip }

{\small\em $^{*}$University of Macau, $^{\dag}$ Chinese Academy of Sciences, $^{\ddag}$University of International Business and Economics}

\end{center}

\bigskip

\noindent Appendix \ref{app:A} contains the detailed proofs of the main theorems and Appendix \ref{app:B} provides the network structure construction used in the empirical application.


\section{Proofs of theorems}\label{app:A}
\renewcommand\theequation{A.\arabic{equation}}



\begin{proof}[\bf Proof of Theorem \ref{thm:3.1}]

Let $L_t(\w)$ and $L_t^\ast(\w)$ be defined as in (\ref{eq3.1}) and (\ref{eq3.2}), respectively. Define
\be
\Psi_{t1}(\w)&=&\sum_{s=1}^{\sT}\left[\Y_{s}-\bmu_s\right]^{\top}\left[\bmu_s-\bmu_s^*(\w)\right]K_{s,t}+\sum_{s=1}^{\sT}\left[\Y_{s}-\bmu_s\right]^{\top}\left[\bmu_s^*(\w)-\wh\bmu_s(\w)\right]K_{s,t}\n\\
&=:&\Psi_{t1}^{\dag}(\w)+\Psi_{t1}^{\ddag}(\w),\n
\ee
and
\be
\Psi_{t2}(\w)&=&\sum_{s=1}^{\sT}\left[\bmu_s^*(\w)-\widehat\bmu_s(\w)\right]^{\top}\left[\bmu_s^*(\w)-\widehat\bmu_s(\w)\right]K_{s,t}+2\sum_{s=1}^{\sT}\left[\bmu_{s}-\bmu_s^*(\w)\right]^{\top}\left[\bmu_s^*(\w)-\widehat\bmu_s(\w)\right]K_{s,t}\n\\
&=:&\Psi_{t2}^{\dag}(\w)+2\cdot\Psi_{t2}^{\ddag}(\w).\n
\ee
Note that
\be
&&{\sf PTVMA}_t(\w)-L_t(\w)\n\\
&=&\sum_{s=1}^{\sT}\left[\Y_{s}-\widehat\bmu_s(\w)\right]^{\top}\left[\Y_s-\widehat\bmu_s(\w)\right]K_{s,t}+\lambda\sum_{m=1}^{\sM} w^m\rho_m-\sum_{s=1}^{\sT}\left[\bmu_{s}-\widehat\bmu_s(\w)\right]^{\top}\left[\bmu_s-\wh\bmu_s(\w)\right]K_{s,t}\n\\
&=&\sum_{s=1}^{\sT}\left(\Y_{s}-\bmu_s\right)^{\top}\left(\Y_s-\bmu_s\right)K_{s,t}+2\sum_{s=1}^{\sT}\left(\Y_{s}-\bmu_s\right)^{\top}\left[\bmu_s-\wh\bmu_s(\w)\right]K_{s,t}+\lambda\sum_{m=1}^{\sM} w^m\rho_m\n\\
&=&\sum_{s=1}^{\sT}\left(\Y_{s}-\bmu_s\right)^{\top}\left(\Y_s-\bmu_s\right)K_{s,t}+2\sum_{s=1}^{\sT}\left(\Y_{s}-\bmu_s\right)^{\top}\left[\bmu_s-\bmu_s^*(\w)\right]K_{s,t}+\n\\
&&2\sum_{s=1}^{\sT}\left(\Y_{s}-\bmu_s\right)^{\top}\left[\bmu_s^*(\w)-\wh\bmu_s(\w)\right]K_{s,t}+\lambda\sum_{m=1}^{\sM} w^m\rho_m\n\\
&=&\sum_{s=1}^{\sT}\left(\Y_{s}-\bmu_s\right)^{\top}\left(\Y_s-\bmu_s\right)K_{s,t}+2\Psi_{t,1}(\w)+\lambda\sum_{m=1}^{\sM} w^m\rho_m,\label{eqA.1}
\ee
where the first term is unrelated to $\w$. Consider a further decomposition for $L_t(\w)$:
\be
L_t(\w)&=&\sum_{s=1}^{\sT}\left[\bmu_{s}-\wh\bmu_s(\w)\right]^{\top}\left[\bmu_s-\wh\bmu_s(\w)\right]K_{s,t}\n\\
&=&\sum_{s=1}^{\sT}\left[\bmu_{s}-\bmu_s^*(\w)\right]^{\top}\left[\bmu_s-\bmu_s^*(\w)\right]K_{s,t}+\sum_{s=1}^{\sT}\left[\bmu_s^*(\w)-\wh\bmu_s(\w)\right]^{\top}\left[\bmu_s^*(\w)-\wh\bmu_s(\w)\right]K_{s,t}+\n\\
&&2\sum_{s=1}^{\sT}\left[\bmu_{s}-\bmu_s^*(\w)\right]^{\top}\left[\bmu_s^*(\w)-\wh\bmu_s(\w)\right]K_{s,t}\n\\
&=&L_t^*(\w)+\Psi_{t2}(\w).\label{eqA.2}
\ee
With the decomposition in (\ref{eqA.1}) and (\ref{eqA.2}), to prove Theorem \ref{thm:3.1}, it suffices to show that
\be
&&\sup_{\w\in\calW}\frac{|\Psi_{t1}(\w)|}{ L_t^*(\w)}=o_P(1),\label{eqA.3}\\
&&\sup_{\w\in\calW}\frac{|\Psi_{t2}(\w)|}{ L_t^*(\w)}=o_P(1),\label{eqA.4}\\
&&\sup_{\w\in\calW}\frac{\left|\lambda\sum_{m=1}^\sM w^m\rho_m\right |}{ L_t^*(\w)}=o_P(1).\label{eqA.5}
\ee

By Assumption \ref{ass:2}(ii), we readily have that
\[
\sup_{\w\in\calW}\left|\frac{\lambda\sum_{m=1}^\sM w^m\rho_m}{ L_t^*(\w)}\right|\leq \frac{\lambda \bar{\rho}}{\xi_t}=o_P(1),
\]
completing the proof of (\ref{eqA.5}).

We next turn to the proof of (\ref{eqA.4}). By Assumption \ref{ass:1}, we have
\be
\sup_{\w\in\calW}\left|\Psi_{t,2}^{\dag}(\w)\right|&=&\sup_{\w\in\calW}\left|\sum_{s=1}^\sT\left[\bmu_s^*(\w)-\wh\bmu_s(\w)\right]^\top\left[\bmu_s^*(\w)-\wh\bmu_s(\w)\right]K_{s,t}\right|\n\\
&\leq&\sup_{\w\in\calW}\sum_{m=1}^\sM\sum_{k=1}^\sM w^mw^k\left|\sum_{s=1}^\sT\sum_{i=1}^\sN\left(\mu_{i,s}^{(m)*}-\wh\mu_{i,s}^{(m)}\right)\left(\mu_{i,s}^{(k)*}-\widehat\mu_{i,s}^{(k)}\right)K_{s,t}\right|\n\\
&\leq&\sum_{s=1}^\sT\sum_{i=1}^\sN K_{s,t}\max_{1\leq m,k\leq \sM}\left|\left(\mu_{i,s}^{(m)*}-\widehat\mu_{i,s}^{(m)}\right)\left(\mu_{i,s}^{(k)*}-\widehat\mu_{i,s}^{(k)}\right)\right|\n\\
&\leq&\sum_{s=1}^\sT\sum_{i=1}^\sN K_{s,t}\left(\max_{1\leq m\leq \sM}\left\|\Z_{i,s-1}^{(m)}\right\|\left\|\widehat\btheta_s^{(m)}-\btheta_{s,\ast}^{(m)}\right\|\right)^2\n\\
&\leq&\sum_{s=1}^\sT\sum_{i=1}^\sN K_{s,t}\left\|\Z_{i,s-1}\right\|^2\sum_{m=1}^{\sM}\left\|\widehat\btheta_s^{(m)}-\btheta_{s,\ast}^{(m)}\right\|^2\n\\
&=&O_P\left(\bar\rho^2\sT\sN\sM h\zeta^2\right),\label{eqA.6}
\ee
where $\zeta$ is the rate defined in Assumption \ref{ass:1}(i). On the other hand, note that
\be
{\sf E}\left[\sup_{\w\in{\mathscr W}}\sum_{s=1}^\sT\left\|\bmu_s-\bmu_s^*(\w)\right\|^2 K_{s,t}\right]
&\leq&{\sf E}\left[\max_{1\leq m\leq \sM}\sum_{s=1}^\sT\left\|\bmu_s-\bmu_{s,\ast}^{(m)}\right\|^2 K_{s,t}\right]\n\\
&\leq& \sum_{m=1}^\sM \sum_{s=1}^\sT \sum_{i=1}^{\sN} K_{s,t} \mE \left| \mu_{i,s} - \mu_{i,s}^{\m\ast} \right|^2 \n\\
&\leq& \sum_{m=1}^{M} \sum_{s=1}^\sT \sum_{i=1}^{\sN} K_{s,t} \mE \| \Z_{i,s-1} \|^2 \cdot \| \btheta_s - \bPi^{\m\top} \btheta_{s,\ast}^\m \|^2 \n\\
&=& O\left( \bar\rho \sM\sN\sT h\right )
\ee
by Assumption \ref{ass:1}(\romannumeral2)(\romannumeral3), which, together with (\ref{eqA.6}) and the Cauchy-Schwarz inequality, leads to
\be
\sup_{\w\in{\mathscr W}}\left|\Psi_{t,2}^{\ddag}(\w)\right|&=&\sup_{\w\in\calW}\sum_{s=1}^{\sT}\left[\bmu_{s}-\bmu_s^*(\w)\right]^{\top}\left[\bmu_s^*(\w)-\widehat\bmu_s(\w)\right]K_{s,t}\n\\
&\leq&\sup_{\w\in{\mathscr W}}\left[\sum_{s=1}^\sT\|\bmu_s-\bmu_s^*(\w)\|^2{K_{s,t}}\right]^{1/2}\left[\sum_{s=1}^\sT\|\bmu_s^*(\w)-\widehat\bmu_s(\w)\|^2{K_{s,t}}\right]^{1/2}\n\\
&=&O_P\left((\bar\rho \sM\sN \sT h)^{1/2}\right)\times O_P\left(\bar\rho \zeta(\sT\sN\sM h)^{1/2}\right)\n\\
&=&O_P\left(\bar{\rho}^{3/2} \zeta \sM\sN\sT h\right).\label{eqA.7}
\ee
By virtue of (\ref{eqA.6}), (\ref{eqA.7}) and Assumption \ref{ass:2}(ii), we may show that
\begin{eqnarray}
\sup_{\w\in\calW}\frac{|\Psi_{t2}(\w)|}{ L_t^*(\w)}&\leq&\sup_{\w\in\calW}\frac{|\Psi_{t2}^\dag(\w)|}{ L_t^*(\w)}+\sup_{\w\in\calW}\frac{|\Psi_{t2}^\ddag(\w)|}{ L_t^*(\w)}\notag\\
&=&O_P\left(\xi_t^{-1} \bar{\rho}^2 \zeta^2 \sM\sN\sT h + \xi_t^{-1}\bar\rho^{3/2} \zeta\sM\sN \sT h\right) \notag\\
&=&O_P\left(\xi_t^{-1}  \bar{\rho}^2 \zeta \sM\sN\sT h\right)=o_P(1),\notag
\end{eqnarray}
completing the proof of (\ref{eqA.4}).

We finally turn to the proof of (\ref{eqA.3}). Note that
\be
{\sf E}\left[\sup_{\w\in\calW}\left|\Psi_{t,1}^\dag(\w)\right|^2\right]&=&{\sf E}\left[\sup_{\w\in\calW}\left|\sum_{s=1}^\sT\left(\Y_s-\bmu_s\right)^\top\left[\bmu_s-\bmu_s^*(\w)\right]K_{s,t}\right|^2\right]\n\\
&\leq&{\sf E}\left[\sup_{\w\in\calW}\left(\sum_{m=1}^\sM w^m\left|\sum_{s=1}^\sT\left(\Y_s-\bmu_s\right)^\top\left(\bmu_s-\bmu_{s,\ast}^{(m)}\right)K_{s,t}\right|\right)^2\right]\n\\
&\leq&{\sf E}\left[\left(\max_{1\leq m\leq \sM}\left|\sum_{s=1}^\sT\left(\Y_s-\bmu_s\right)^\top\left(\bmu_s-\bmu_{s,\ast}^{(m)}\right)K_{s,t}\right|\right)^2\right]\n\\
&\leq&\sum_{m=1}^\sM{\sf E}\left[\left|\sum_{s=1}^\sT\beps_s^\top\left(\bmu_s-\bmu_{s,\ast}^{(m)}\right)K_{s,t}\right|^2\right]\n\\
&=&\sum_{m=1}^\sM{\sf E}\left[\left|\sum_{s=1}^\sT\sum_{i=1}^\sN\epsilon_{i,s}\left(\mu_{i,s}-\mu_{i,s}^{(m)*}\right)K_{s,t}\right|^2\right]\n\\
&=& \sum_{m=1}^\sM\sum_{s=1}^\sT\sum_{i=1}^\sN K_{s,t}^2{\sf E}\left[ \epsilon_{i,s}^2 \left(\mu_{i,s}-\mu_{i,s}^{(m)*}\right)^2 \right] \n\\
&\leq& \sum_{m=1}^\sM\sum_{s=1}^\sT\sum_{i=1}^\sN K_{s,t}^2{\sf E} \epsilon_{i,s}^2 \cdot \mE \| \Z_{i,s-1} \|^2 \cdot \| \btheta_s - \bPi^{\m\top} \btheta_{s,\ast}^\m \|^2 \n\\
&=&O(\bar{\rho} \sM \sN \sT h),\label{eqA.8}
\ee
by Assumption \ref{ass:1}(\romannumeral3). On the other hand, since
\[
{\sf E}\left[\sum_{s=1}^\sT\left\|(\Y_s-\bmu_s)K_{s,t}^{1/2}\right\|^2\right]= {\sf E}\left[\sum_{s=1}^\sT\left\|\beps_s K_{s,t}^{1/2}\right\|^2\right]
=\sum_{s=1}^\sT K_{s,t}{\sf E}\left[{\beps_s^\top\beps_s}\right]=O(\sN\sT h),
\]
by (\ref{eqA.6}) and the Cauchy-Schwarz  inequality, we have
\be
\sup_{\w\in\calW}\left|\Psi_{t,1}^{\ddag}(\w)\right|&=&\sup_{\w\in\calW}\left|\sum_{s=1}^\sT(\Y_s-\bmu_s)^\top[\bmu_s^*(\w)-\widehat\bmu_s(\w)]K_{s,t}\right|\n\\
&\leq& \left[\sum_{s=1}^\sT\|\Y_s-\bmu_s\|^2K_{s,t}\right]^{1/2}\sup_{\w\in\calW}\left[\sum_{s=1}^\sT\|\bmu_s^*(\w)-\widehat\bmu_s(\w)\|^2K_{s,t}\right]^{1/2}\n\\
&=& O_P\left((\sN\sT h)^{1/2}\right)\times O_P\left(\bar\rho \zeta(\sT\sN\sM h)^{1/2}\right)\n\\
&=& O_P\left(\bar\rho \zeta\sN\sT \sM^{1/2}h\right).\label{eqA.9}
\ee
Combining (\ref{eqA.8}), (\ref{eqA.9}) and Assumptions \ref{ass:1}(i) and \ref{ass:2}(ii), we readily have that
\begin{eqnarray}
\sup_{\w\in\calW}\frac{\left|\Psi_{t,1}(\w)\right|}{ L_t^*(\w)}&\leq& \sup_{\w\in\calW}\frac{\left|\Psi_{t,1}^\dag(\w)\right|}{ L_t^*(\w)}+\sup_{\w\in\calW}\frac{\left|\Psi_{t,1}^\ddag(\w)\right|}{ L_t^*(\w)}\n\\
&\leq&O_P\left(\xi_t^{-1}\left[(\bar{\rho} \sM \sN\sT h)^{1/2}+\bar\rho \zeta\sM^{1/2}\sN\sT h\right]\right)\n\\
&=&O_P\left(\xi_t^{-1}\bar\rho \zeta\sM^{1/2}\sN\sT h\right)=o_P(1),\n
\end{eqnarray}
completing the proof of (\ref{eqA.3}).

\end{proof}


\begin{proof}[\bf Proof of Theorem \ref{thm:3.2}]

Let
$$
{\sf PTVMA}^{*}_t(\w)=\frac{1}{\sT h}{\sf PTVMA}_t(\w)-\frac{1}{\sT h}\sum_{i=1}^\sN\sum_{s=1}^\sT\sigma_{i}^{2}(s/T)K_{s,t},
$$
where the second term on the right side is unrelated to $\w$. Hence,
\[
\widehat{\w}_t=\mathop{\arg\min}\limits_{\w\in \calW}  {\sf PTVMA}_t(\w)=\mathop{\arg\min}\limits_{\w\in \calW}{\sf PTVMA}_t^{*}(\w).
\]

We first show the in-sample point-wise weak consistency, i.e., $\widehat{\w}_t$ converges to $\w_t^*$ in probability. For any $\w\in\mathscr{W}$, we have
\be
&&\sup_{\w\in\calW}\frac{1}{\sN}\left\vert {\sf PTVMA}_t^*(\w) - (\sT h)^{-1}{\sf E}[L_{t}(\w)] \right\vert \notag\\
&=&\sup_{\w\in\calW} \left| \frac{1}{\sN \sT h}\sum_{s=1}^\sT \Big[ \sum_{i=1}^\sN \sigma_{i}^{2}(s/\sT)(u_{is}^2-1)K_{s,t}+\left\|\bmu_s-\widehat{\bmu}_s(\w)\right\|^2K_{s,t}+2(\Y_s-\bmu_s)^\top[\bmu_s-\widehat{\bmu}_s(\w)]K_{s,t}\right. \notag\\
&&\left. - {\sf E}\|\bmu_{s}-\widehat{\bmu}_{s}(\w)\|^2K_{s,t} \Big] + \frac{\lambda}{\sN\sT h} \sum_{m=1}^\sM w^m\rho_m \right| \notag\\
&\leq&\sup_{\w\in\calW}\bigg\{\bigg| \frac{1}{\sN \sT h}\sum_{s=1}^\sT \bigg[\left\|\bmu_s-\widehat{\bmu}_s(\w)\right\|^2K_{s,t} - \|\bmu_s-\bmu_s^*(\w)\|^2K_{s,t}\bigg] \bigg|+ \n\\
&&\bigg| \frac{1}{\sN \sT h}\sum_{s=1}^\sT K_{s,t}\bigg[ \|\bmu_s-\bmu_s^*(\w)\|^2 - {\sf E}\|\bmu_s-\bmu_s^*(\w)\|^2 \bigg] \bigg| + \n\\
&&\bigg| \frac{1}{\sN \sT h}\sum_{s=1}^\sT K_{s,t}\left[{\sf E}\|\bmu_s-\bmu_s^*(\w)\|^2 - {\sf E}\|\bmu_{s}-\widehat{\bmu}_{s}(\w)\|^2\right]  \bigg| + \n\\
&&\left|\frac{2}{\sN \sT h}\Psi_{t,1}(\w)\right|+ \left| \frac{1}{\sN \sT h}\lambda\sum_{m=1}^\sM w^m\rho_m\right|\bigg\}+\left| \frac{1}{\sN \sT h}\sum_{s=1}^\sT\sum_{i=1}^\sN \sigma_{i}^{2}(s/\sT)(u_{is}^2-1)K_{s,t}\right| \n\\
&=:&\Pi_1 + \Pi_2 + \Pi_3 + \Pi_4 + \Pi_5+\Pi_{6},\label{eqA.10}
\ee
where $\Psi_{t,1}(\w)$ is defined in the proof of Theorem \ref{thm:3.1}.

From the proof of Theorem \ref{thm:3.1}, $\bar{\rho}\sM\zeta\rightarrow0$ in Assumption \ref{ass:2}(ii) and $\lambda\bar{\rho}=o(\sN Th)$ in Assumption \ref{ass:3}(iii), we may show that
\begin{eqnarray}
\Pi_1&=&O_P\left(\frac{\bar\rho \zeta\sM\sN\sT h}{\sN\sT h}\right)=O_P\left(\bar\rho \zeta\sM\right)=o_P(1),\label{eqA.11}\\
\Pi_3&=&O\left(\frac{\bar\rho \zeta\sM\sN\sT h}{\sN\sT h}\right)=O\left(\bar\rho \zeta\sM\right)=o(1),\label{eqA.12}\\
\Pi_4&=&o_P\left(\frac{\bar\rho \zeta\sM\sN\sT h}{\sN\sT h}\right)=o_P\left(\bar\rho \zeta\sM\right)=o_P(1),\label{eqA.13}\\
\Pi_5&=&O\left(\frac{\lambda\bar{\rho}}{\sN\sT h}\right)=o(1).\label{eqA.14}
\end{eqnarray}
Meanwhile, it follows from Assumption \ref{ass:3}(ii)(iii) that
\begin{equation}\label{eqA.15}
\Pi_2=O_P\left(\bar\rho^2(\sN\sT h)^{-1/2}\right)=o_P(1),\quad \Pi_6=O_P\left((\sN\sT h)^{-1/2}\right)=o_P(1).
\end{equation}
Combining (\ref{eqA.10})--(\ref{eqA.15}), we readily have that
\[
\sup_{\w\in\calW}\frac{1}{\sN}\left|{\sf PTVMA}_t^*(\w) - (\sT h)^{-1}{\sf E}[L_{t}(\w)]\right| = o_P(1),
\]
which, together with the argmin continuity theorem and the fact that ${\sf PTVMA}_t^*(\w)$ has a unique minimizer on $\mathscr{W}$, completes the proof of weak consistency, i.e., $\widehat{\w}_t\cp \w_t^*$.

We next turn to the proof of (\ref{eq3.3}). Letting $\widehat\bmu_{t} = (\wh\bmu_{t}^{(1)},\ldots,\wh\bmu_{t}^{(\sM)})^\top$ and
\[
{\boldsymbol\Delta}_t=\frac{1}{\sN \sT h}\sum_{s=1}^\sT\widehat\bmu_{s}\widehat\bmu_{s}^\top K_{s,t},
\]
we write
\be
\frac{1}{\sN Th}{\sf E}[L_{t}(\w)] &=& \frac{1}{\sN\sT h}{\sf E}\left(\sum_{s=1}^\sT[\bmu_{s}-\widehat{\bmu}_{s}(\w)]^\top[\bmu_{s}-\widehat{\bmu}_{s}(\w)]K_{s,t}\right)\n \\
&=& \frac{1}{\sN \sT h}{\sf E}\left[ \sum_{s=1}^\sT \bmu_{s}^\top\bmu_{s}K_{s,t} - 2\w^\top\sum_{s=1}^\sT\widehat\bmu_{s}\bmu_{s}K_{s,t} +\w^\top\sum_{s=1}^\sT\widehat\bmu_{s}\widehat\bmu_{s}^\top K_{s,t}\w\right]\n  \\
&=&\w^\top{\sf E}\left[{\boldsymbol\Delta}_t\right]\w - 2\w^\top{\sf E}\left[\frac{1}{\sN Th}\sum_{s=1}^\sT\widehat\bmu_{s}\bmu_{s}K_{s,t} \right] + c_\ast,\n
\ee
for some positive constant $c_\ast$ which is unrelated to $\w$. Letting
\[
{\boldsymbol\Delta}_t^\ast=\frac{1}{\sN Th}\sum_{s=1}^\sT\widetilde\bmu_{s}^\ast\widetilde\bmu_{s}^{\ast\top} K_{s,t},
\]
it follows from Assumption \ref{ass:3}(i) that $\lambda_{\min}({\sf E}[{\boldsymbol\Delta}_t^\ast])$ is larger than $\underline\kappa$. Furthermore, by Assumption \ref{ass:1}(i) and $\bar{\rho}\zeta\sM \rightarrow0$ in Assumption \ref{ass:2}(ii), we may show that
\[
\lambda_{\min}({\sf E}[{\boldsymbol\Delta}_t])\geq \lambda_{\min}({\sf E}[{\boldsymbol\Delta}_t^\ast])-O_P\left(\bar{\rho}\zeta\sM+\bar{\rho}\zeta^2\sM\right)\geq \underline\kappa/2,
\]
indicating that ${\sf E} L_{t}[(\mathbf{w})]$ is strongly convex, and consequently, for any $\w\in\calW$,
\begin{equation}\label{eqA.16}
\frac{1}{\sN \sT h}\left[{\sf E} L_{t}(\w) -{\sf E} L_{t}(\w_t^*)\right] >\frac{1}{2}\underline{\kappa}\|\w-\w_t^*\|^2.
\end{equation}
Since $\widehat{\w}_t$ is the minimizer of ${\sf PTVMA}_t^*(\w)$, we readily have that
\be
0 &\geq&{\sf PTVMA}_t^*(\widehat\w_t) - \sf{\sf PTVMA}_t^*(\w_t^*)  \n\\
&=&\left({\sf PTVMA}_t^*(\widehat\w_t) - \frac{1}{\sT h}{\sf E}[L_{t}(\widehat\w_t)]\right) + \left(\frac{1}{\sT h}{\sf E}[L_{t}(\widehat\w_t)]-\frac{1}{\sT h}{\sf E}[L_{t}(\w_t^*)]\right)-\n\\
&& \left({\sf PTVMA}_t^*(\w_t^*)-\frac{1}{\sT h}{\sf E} [L_{t}(\w_t^*)]\right).\label{eqA.17}
\ee
Then, by virtue of (\ref{eqA.10})--(\ref{eqA.15}) and (\ref{eqA.17}), we obtain
\be
&&\frac{1}{\sN}\left(\frac{1}{\sT h}{\sf E} [L_{t}(\widehat\w_t)] - \frac{1}{\sT h}{\sf E} [L_{t}(\w_t^*)]\right)\n\\
&\leq & \frac{1}{\sN}\left[{\sf PTVMA}_t^*(\w_t^*) - \frac{1}{\sT h}{\sf E}[ L_{t}(\w_t^*)]\right] - \frac{1}{\sN}\left[{\sf PTVMA}_t^*(\widehat\w_t) - \frac{1}{\sT h}{\sf E}[ L_{t}(\widehat\w_t)]\right]\n \\
&\leq& 2\supw \frac{1}{\sN} \left|{\sf PTVMA}_t^*(\w) -\frac{1}{\sT h}{\sf E} [L_{t}(\w)] \right|\n \\
&=&O_P\left(\frac{\bar\rho(\lambda+ \zeta\sM\sN\sT h)}{\sN\sT h}+\bar\rho^2(\sN\sT h)^{-1/2}\right).\n
\ee
This, together with (\ref{eqA.16}) with $\widehat\w_t$ replacing $\w$, completes the proof of (\ref{eq3.3}).
\end{proof}


\begin{proof}[\bf Proof of Theorem \ref{thm:3.3}]

Similar to the proof of Theorem \ref{thm:3.2}, we have
\begin{equation}\label{eqA.18}
\frac{1}{\sN}{\sf PTVMA}_t^*(\w)=\frac{1}{\sN \sT h}L_{t}^*(\w)+O_P\left(\frac{\bar\rho(\lambda+ \zeta\sM\sN\sT h)}{\sN\sT h}+\bar\rho^2(\sN\sT h)^{-1/2}\right).
\end{equation}

For $\w=(w^1,\ldots,w^{\sM})\in{\calW}$, we define $\w^\circ=(w^1_\circ,\ldots,w^\sM_\circ)$ as a weight vector with $w_\circ^m=0$ for the correctly-specified models ($m\in{\calD}$) and $w_\circ^m=w^m/(1-\tau(\w))$ for all the other misspecified models ($m\notin{\calD}$), where $\tau(\w)=\sum_{m\in\calD}w^m$. Note that
\be\label{eqA.19}
\mu_{i,t}^{(m)*}-\mu_{i,t}=0,\ \ 1\leq i\leq \sN\quad \text{and}\quad \frac{1}{\sN}\left\|\bmu_{t}-\bmu_{t,\ast}^{(m)}\right\|^2=0\ \ \text{for $m\in{\calD}$}.
\ee
By (\ref{eqA.19}) and the definition of $\w^\circ$, we readily have that
\begin{eqnarray}
\frac{1}{\sN\sT h}L_{t}^*(\w)&=&\frac{1}{\sN\sT h}\sum_{s=1}^\sT \|\bmu_{s}-\bmu_{s}^*(\w)\|^2K_{s,t} \n\\
&=&\frac{1}{\sN\sT h}\left[\sum_{s=1}^\sT \left\|\sum_{m\notin {\calD}}{w}^m (\bmu_{s}-\bmu_{s,\ast}^{(m)}) \right\|^2K_{s,t}\right] \n\\
&=&\frac{[1-\tau(\w)]^2}{\sN\sT h}\left[\sum_{s=1}^\sT \left\|\sum_{m\notin {\calD}}(1-\tau(\w))^{-1}{w}^m (\bmu_{s}-\bmu_{s,\ast}^{(m)}) \right\|^2K_{s,t}\right]  \n\\
&=&\frac{[1-\tau(\w)]^2}{\sN\sT h}L_{t}^*(\w^\circ).\label{eqA.20}
\end{eqnarray}
Hence, replacing $\w$ with $\wh\w_t$ and using (\ref{eqA.18}) and (\ref{eqA.20}), we obtain
\begin{equation}\label{eqA.21}
\frac{1}{\sN}{\sf PTVMA}_t^*(\wh{\w}_t)=\frac{[1 - \tau(\wh\w_t)]^2}{\sN\sT h} L_{t}^*(\wh{\w}_t^\circ)+O_P\left(\frac{\bar\rho(\lambda+ \zeta\sM\sN\sT h)}{\sN Th}+\bar\rho^2(\sN\sT h)^{-1/2}\right),
\end{equation}
where $\wh{\w}_t^\circ$ is defined similarly to ${\w}^\circ$ but with $w^m$, $m\notin\calD$, and $\tau(\w)$ replaced by $\wh w_t^m$ and $\tau(\wh\w_t)=\sum_{m\in\calD}\wh w_t^m$, respectively.

Let $\check\w$ be a weight vector with $\sum_{m\in{\calD}}\check w^m=1$. By (\ref{eqA.18}) and (\ref{eqA.20}), we also have
\begin{equation}\label{eqA.22}
\frac{1}{\sN}{\sf PTVMA}_t^*(\check\w)=O_P\left(\frac{\bar\rho(\lambda+ \zeta\sM\sN\sT h)}{\sN\sT h}+\bar\rho^2(\sN\sT h)^{-1/2}\right).
\end{equation}
As $\wh{\w}_t=\arg\min_{\w\in\calW}{\sf PTVMA}_t^*(\w)$, we obtain
\[
\frac{1}{\sN}{\sf PTVMA}_t^*(\check\w)\geq\frac{1}{\sN}{\sf PTVMA}_t^*(\wh{\w}_t),
\]
which, together with (\ref{eqA.21}) and (\ref{eqA.22}), indicates that
\[
\frac{[1-\tau(\wh\w_t)]^2}{\sN\sT h} L_{t}^*(\wh{\w}_t^\circ)=O_P\left(\frac{\bar\rho(\lambda+ \zeta\sM\sN\sT h)}{\sN\sT h}+\bar\rho^2(\sN\sT h)^{-1/2}\right).
\]
This, together with the condition (\ref{eq3.4}), indicates that we must have $\tau(\widehat{\w}_t)\stackrel{P}\rightarrow 1 $, completing the proof of Theorem \ref{thm:3.3}.
\end{proof}


\begin{proof}[\bf Proof of Theorem \ref{thm:3.4}]

By Lemma 1 in \cite{GZWCZ19}, it suffices to show that
\begin{equation}\label{eqA.23}
\mathop{\sup}\limits_{\w\in \calW}\left | \frac{R_{\sT+1}(\w)  }{R_{\sT+1}^{*} (\w) }-1  \right |=o(1),
\end{equation}
and
\begin{equation}\label{eqA.24}
\mathop{\sup}\limits_{\w\in \calW}\left | \frac{{{{\sf PTVMA}_\sT^*(\w)}}  }{R_{\sT+1}^{*} (\w) }-1  \right |=o_P(1).
\end{equation}

We first verify (\ref{eqA.23}). Note that
\be\label{eqA.25}
\sup\limits_{\w\in \calW}\left\Vert\widehat{\Y}_{\sT+1}(\w)-\Y_{\sT+1}^{*}(\w)\right\Vert =O_P\left(\bar{\rho}\zeta (\sM\sN)^{1/2}\right),
\ee
by Assumption \ref{ass:1}(i). By Assumption \ref{ass:4}(ii), we may show that
\be
&&\sup_\sT{\sf E}\left\{\xi_{\sT+1}^{*-1}\max_{1\leq m,m'\leq \sM}\left\|\wh{\Y}_{\sT+1}^{(m)}-\Y_{\sT+1}^{(m)*}\right\|\left\|\wh{\Y}_{\sT+1}^{(m')}-\Y_{\sT+1}^{(m')*}\right\|\right\}^{1+\delta_1/2}\n\\
&\leq&\sup_\sT{\sf E}\left\{\xi_{\sT+1}^{*-1/2}\max_{1\leq m\leq \sM}\left\|\wh{\Y}_{\sT+1}^{(m)}-\Y_{\sT+1}^{(m)*}\right\|\right\}^{2+\delta_1}<\infty.\label{eqA.26}
\ee
This indicates that $\xi_{\sT+1}^{*-1}\max_{1\leq m,m'\leq \sM}\|\wh{\Y}_{\sT+1}^{(m)}-\Y_{\sT+1}^{(m)*}\|\|\wh{\Y}_{\sT+1}^{(m')}-\Y_{\sT+1}^{(m')*}\|$ is uniformly integrable over $\sT$.  Thus, by Assumption \ref{ass:4}(ii), the weight constraint, i.e., $\sum_{i=1}^\sM w^i=1$, and (\ref{eqA.26}), we have
\be
&&\xi _{\sT+1}^{*-1}\sup_{\w\in \calW} \left| \left[\Y_{\sT+1}-\wh{\Y}_{\sT+1}\left ( \w\right )     \right] ^{\top} \left[\Y_{\sT+1}-\wh{\Y}_{\sT+1}\left ( \w\right )     \right]- \left[ \Y_{\sT+1}-\Y_{\sT+1}^{*}\left( \w \right)  \right] ^{\top}\left[ \Y_{\sT+1}-\Y_{\sT+1}^{*}\left( \w \right)  \right] \right| \n\\
&=&\xi_{\sT+1}^{*-1}\sup_{\w\in \calW} \left| \left[ \wh{\Y}_{\sT+1}(\w)- \Y_{\sT+1}^{*}(\w)  \right]^{\top} \left[ \wh{\Y}_{\sT+1}\left (\w\right )+\Y_{\sT+1}^{*}(\w)-2\Y_{\sT+1} \right]  \right| \n\\
&=&\xi_{\sT+1}^{*-1}\sup_{\w\in \calW}\left|\sum_{m=1}^\sM\sum_{m'=1}^\sM w^mw^{m'}\left[\wh{\Y}_{\sT+1}^{(m)}-{\Y}_{T+1}^{(m)*}\right]^\top\left[\wh{\Y}_{\sT+1}^{(m')}+{\Y}_{\sT+1}^{(m')*}-2\Y_{\sT+1}\right]\right|\n\\
&\leq&\xi_{\sT+1}^{*-1}\max_{1\leq m,m'\leq \sM}\left\|\wh{\Y}_{\sT+1}^{(m)}-{\Y}_{\sT+1}^{(m)*}\right\|\left\|\wh{\Y}_{\sT+1}^{(m')}+{\Y}_{\sT+1}^{(m')*}-2\Y_{\sT+1}\right\|\sup_{\w\in\calW}\left(\sum_{m=1}^\sM w^m\sum_{m'=1}^\sM w^{m'}\right)\n\\
&=& \xi_{\sT+1}^{*-1}\max_{1\leq m,m'\leq \sM}\left\|\wh{\Y}_{\sT+1}^{(m)}-{\Y}_{\sT+1}^{(m)*}\right\|\left\|\wh{\Y}_{\sT+1}^{(m')}-\Y_{\sT+1}^{(m')*}-2\beps_{\sT+1}^{(m')*}\right\|\n\\
&\leq& \xi_{\sT+1}^{*-1}\max_{1\leq m,m'\leq \sM}\left\|\wh{\Y}_{\sT+1}^{(m)}-{\Y}_{\sT+1}^{(m)*}\right\|\left\|\wh{\Y}_{\sT+1}^{(m')}-{\Y}_{\sT+1}^{(m')*}\right\|+2 \xi_{\sT+1}^{*-1}\max_{1\leq m,m'\leq \sM}\left\|\wh{\Y}_{\sT+1}^{(m)}-\Y_{\sT+1}^{(m)*}\right\|\left\|\beps_{\sT+1}^{(m')*}\right\|,\n
\ee
implying that $\xi_{\sT+1}^{*-1}\sup_{\w\in\calW}\left|\left\|\Y_{\sT+1}-\wh{\Y}_{\sT+1}(\w)\right\|^2-\left\|\Y_{\sT+1}-{\Y}^*_{\sT+1}(\w)\right\|^2\right|$ is uniformly integrable. This, together with (\ref{eqA.25}), leads to
\be
\sup\limits_{\w\in \calW}\left | \frac{R_{\sT+1}(\w)  }{R_{\sT+1}^{\ast} (\w) }-1  \right | 
&\leq&\xi_{\sT+1}^{*-1}\mathop{\sup}\limits_{\w\in\calW}\left | {\sf E}\left [  \left \|\Y_{\sT+1}-\wh{\Y}_{\sT+1}\left (\w\right )     \right \| ^{2}-\left \| \Y_{\sT+1}-\Y_{\sT+1}^{*}(\w)     \right \|^{2}  \right ]  \right | \n\\
&\leq&{\sf E}\left \{ \xi_{\sT+1}^{\ast -1}\mathop{\sup}\limits_{\w\in \calW}\left | \left \| \Y_{\sT+1}-\wh{\Y}_{\sT+1}(\w)     \right \| ^{2}-\left \| \Y_{\sT+1}-\Y_{\sT+1}^{*}\left ( \w\right )     \right \| ^{2} \right | \right \} \n \\
&=&O\left(\xi_{\sT+1}^{*-1}\sM\sN \bar{\rho}^2\zeta^2+\xi_{\sT+1}^{*-1}\sM\sN \bar{\rho}\zeta\right)=O\left(\xi_{\sT+1}^{*-1}\sM\sN \bar{\rho}\zeta\right)=o(1),\label{eqA.27}
\ee
completing the proof of (\ref{eqA.23}).

We next turn to the proof of (\ref{eqA.24}). Observe that
\be
&&\sup\limits_{\w\in \calW}\left | \frac{{\sf PTVMA}_\sT^{*} (\w)}{R_{\sT+1}^{*} (\w) }-1  \right |\n\\
&\leq&\xi_{\sT+1}^{*-1}\mathop{\sup}\limits_{\w\in \calW}\left| {\sf PTVMA}_\sT^{*}(\w) -R_{\sT+1}^{*} (\w)    \right|\n\\
&\leq&\xi _{\sT+1}^{*-1}\mathop{\sup}\limits_{\w\in \calW}\left | \frac{1}{\sT h}\sum_{t=1}^{\sT} \left\Vert \Y_{t}-\wh{\bmu}_{t}(\w)\right\Vert^2K_{t,\sT}- {\sf E}\left[\|\Y_{\sT+1}-\Y_{\sT+1}^{*}(\w)\|^2\right] \right|+\xi_{\sT+1}^{*-1}\frac{\lambda}{\sT h}\mathop{\sup}\limits_{\w\in \calW}\left|\sum_{m=1}^\sM w^m\rho_m\right|\n\\
&\leq&\xi_{\sT+1}^{*-1}\mathop{\sup}\limits_{\w\in \calW}\left| \frac{1}{\sT h}\sum_{t=1}^{\sT}\left( \left\Vert \Y_{t}-\wh{\bmu}_{t}(\w)\right\Vert^2- \left\Vert \Y_{t}-\bmu_{t}^{*}(\w) \right\Vert^{2}\right)K_{t,\sT}  \right|+ \n\\
&&\xi_{\sT+1}^{*-1}\mathop{\sup}\limits_{\w\in \calW}\left | \frac{1}{\sT h}\sum_{t=1}^{\sT}\left( \left\Vert\Y_{t}-\bmu_{t}^{*}(\w) \right\Vert^2 - {\sf E}\left[\left\Vert\Y_{t}-\bmu_{t}^{*}(\w)\right\Vert^2\right]\right) K_{t,\sT}  \right |+\n\\
&& \xi_{\sT+1}^{*-1}\mathop{\sup}\limits_{\w\in \calW}\left | \frac{1}{\sT h}\sum_{t=1}^{\sT}{\sf E}\left[|| \Y_{t}-\bmu_{t}^{*}(\w)  ||^{2}\right]K_{t,\sT} - {\sf E}\left[\|\Y_{\sT+1}-\Y_{\sT+1}^{*}(\w)\|^2\right]\right|+o_P(1),\n
\ee
which the last inequality is due to Assumption \ref{ass:4}(iii). Hence, to prove (\ref{eqA.24}), it suffices to show
\be
&&\xi_{\sT+1}^{*-1}\mathop{\sup}\limits_{\w\in \calW}\left| \frac{1}{\sT h}\sum_{t=1}^{\sT}\left( \left\Vert \Y_{t}-\wh{\bmu}_{t}(\w)\right\Vert^2- \left\Vert \Y_{t}-\bmu_{t}^{*}(\w) \right\Vert^{2}\right)K_{t,\sT}  \right|=o_P(1),\label{eqA.28} \\
&&\xi_{\sT+1}^{*-1}\mathop{\sup}\limits_{\w\in \calW}\left | \frac{1}{\sT h}\sum_{t=1}^{\sT}\left( \left\Vert\Y_{t}-\bmu_{t}^{*}(\w) \right\Vert^2 - {\sf E}\left[\left\Vert\Y_{t}-\bmu_{t}^{*}(\w)\right\Vert^2\right]\right) K_{t,\sT}  \right |=o_P(1),\label{eqA.29}\\
&& \xi_{\sT+1}^{*-1}\mathop{\sup}\limits_{\w\in \calW}\left | \frac{1}{\sT h}\sum_{t=1}^{\sT}{\sf E}\left[|| \Y_{t}-\bmu_{t}^{*}(\w)  ||^{2}\right]K_{t,\sT} - {\sf E}\left[\|\Y_{\sT+1}-\Y_{\sT+1}^{*}(\w)\|^2\right]\right|=o(1).\label{eqA.30}
\ee

By Assumptions \ref{ass:1}(i) and \ref{ass:4}(iii), we may prove that
\be
&&\xi _{\sT+1}^{*-1}\mathop{\sup}\limits_{\w\in \calW}\left | \frac{1}{\sT h}\sum_{t=1}^{\sT}K_{t,\sT}\left [ \left \|\Y_t-\wh{\bmu}_t(\w) \right \|^{2}- \left \| \Y_t-\bmu_{t}^{*}(\w)    \right \|^{2} \right ]    \right | \n\\
&=&\xi _{\sT+1}^{\ast -1}\mathop{\sup}\limits_{\w\in \calW}\left |  \frac{1}{\sT h}\sum_{t=1}^{\sT}K_{t,\sT}\bigg[   \wh{\bmu}_t(\w)- \bmu_t^{*}(\w)     \bigg]^\top\bigg[\wh{\bmu}_t(\w)- \bmu_t^{\ast}(\w)+ 2\bmu_t^{\ast}\left ( \w\right )-2\Y_t   \bigg]    \right |\n \\
&\leq& \xi_{\sT+1}^{*-1} \sup_{\w \in \calW} \frac{1}{\sT h}\sum_{t=1}^{\sT}K_{t,\sT}\Big\| \sum_{m=1}^{\sM} w^m \left(\wh{\bmu}_t^{(m)} - \bmu_{t,\ast}^{(m)}\right)\Big\| \Big\|2\left[\Y_t-\bmu_t^*(\w)\right]-\left[\wh{\bmu}_t(\w)-\bmu_t^*(\w)\right]\Big\|\n \\
&=& O_P\left(\xi_{\sT+1}^{*-1}\sM\sN \bar{\rho}\zeta+\xi_{\sT+1}^{*-1}\sM\sN \bar{\rho}^2\zeta^2\right)=O_P\left(\xi_{\sT+1}^{*-1}\sM\sN \bar{\rho}\zeta\right)=o_P(1),\n
\ee
leading to (\ref{eqA.28}).

We next consider the proof of (\ref{eqA.29}). Recall that $ y_{i,t} - \mu_{i,t}^{*}(\w) = \eps_{i,t} + \sum_{m=1}^{M} w^m (\mu_{i,t} - \mu_{i,t}^{\m\ast}) = \eps_{i,t} + \Z_{i,t-1}^\top \sum_{m=1}^{M} w^m (\btheta_{t} - \bPi^{\m\top} \btheta_{t,\ast}^{\m}) $. By Assumptions \ref{ass:1}(\romannumeral3), \ref{ass:3}(\romannumeral2), \ref{ass:4}(iii)(iv), we may show that
\begin{align*}
&\quad \xi _{\sT+1}^{*-1}\mathop{\sup}\limits_{\w\in \calW}\left| \frac{1}{\sT h}\sum_{t=1}^{\sT}\left(\left \| \Y_t-\bmu_t^{*}(\w) \right \|^{2} - {\sf E}\left[\left \| \Y_t-\bmu_t^{*}(\w) \right \|^{2}\right]\right) K_{t,\sT} \right|   \\
&\leq \xi _{\sT+1}^{*-1}\mathop{\sup}\limits_{\w\in \calW}\left| \frac{1}{\sT h} \sum_{t=1}^{\sT} \sum_{i=1}^{\sN} \left(\left| \eps_{i,t} + \Z_{i,t-1}^\top \bdel_t (\w) \right|^{2} - {\sf E} \left| \eps_{i,t} + \Z_{i,t-1}^\top \bdel_t (\w) \right|^{2} \right) K_{t,\sT} \right|   \\
&\leq \xi _{\sT+1}^{*-1}\mathop{\sup}\limits_{\w\in \calW} \left| \frac{1}{\sT h} \sum_{t=1}^{\sT} \sum_{i=1}^{\sN} K_{t,\sT} \left( \eps_{i,t}^2 - \mE \eps_{i,t}^2 \right) \right|   %
+ \xi _{\sT+1}^{*-1}\mathop{\sup}\limits_{\w\in \calW} \left| \frac{2}{\sT h} \sum_{t=1}^{\sT} \sum_{i=1}^{\sN} K_{t,\sT} \eps_{i,t} \Z_{i,t-1}^\top \bdel_t (\w) \right|   \\
&\quad + \xi _{\sT+1}^{*-1}\mathop{\sup}\limits_{\w\in \calW} \left| \frac{1}{Th} \bdel_T^\top (\w) \left[ \sum_{t=1}^\sT\sum_{i=1}^\sN K_{t,\sT}\left(\Z_{i,t-1}\Z_{i,t-1}^\top-{\sf E}\left[\Z_{i,t-1}\Z_{i,t-1}^\top\right]\right) \right] \bdel_T (\w) \cdot (1 + O_p(h)) \right|   \\
&= O_P\left(\xi_{\sT+1}^{*-1}(\sT h)^{-1/2} {\bar{\rho}}^2\sN^{1/2}\right)=o_P(1)
\end{align*}

Finally, we observe that
\begin{align*}
&\quad \xi_{\sT+1}^{*-1}\mathop{\sup}\limits_{\w\in \calW} \left| \frac{1}{\sT h}\sum_{t=1}^{\sT}{\sf E} \left[\left\|\Y_{t}-\bmu_{t}^{*}(\w) \right\|^{2}\right] K_{t,\sT} - {\sf E}\left\| \Y_{\sT+1}-\Y_{\sT+1}^{*}(\w) \right\|^{2} \right|   \\
&= \xi_{\sT+1}^{*-1}\mathop{\sup}\limits_{\w\in \calW} \left| \frac{1}{\sT h} \sum_{t=1}^{\sT} \sum_{i=1}^{\sN} K_{t,\sT} \mE \left| \eps_{i,t} + \Z_{i,t-1}^\top \bdel_t (\w) \right|^{2} - \sum_{i=1}^{\sN} {\sf E} \left| \eps_{i,\sT+1} + \Z_{i,\sT}^\top \bdel_{\sT} (\w) \right|^{2} \right|   \\
&\leq \xi_{\sT+1}^{*-1}\mathop{\sup}\limits_{\w\in \calW} \left| \sum_{i=1}^{\sN} \left( \frac{1}{\sT h} \sum_{t=1}^{\sT} K_{t,\sT} \mE \eps_{i,t}^2 - \mE \eps_{i,\sT+1}^2 \right) \right|   \\
&\quad + \xi_{\sT+1}^{*-1}\mathop{\sup}\limits_{\w\in \calW} \left| \sum_{i=1}^{\sN} \bdel_{\sT}^\top (\w) \left( \frac{1}{\sT h} \sum_{t=1}^{\sT} K_{t,\sT} \mE \Z_{i,t-1}\Z_{i,t-1}^\top - \mE \Z_{i,\sT}\Z_{i,\sT}^\top \right) \bdel_{\sT} (\w) \cdot (1+O(h)) \right|   \\
&= O\left(\bar\rho \xi_{\sT+1}^{*-1} \sN  h\right) = o(1).
\end{align*}
due to Assumption \ref{ass:4}(i)(\romannumeral4). This completes the proof of (\ref{eqA.30}).
\end{proof}


\begin{proof}[\bf Proof of Theorem \ref{thm:3.5}]

As in the proof of Theorem \ref{thm:3.2}, we first show the weak consistency of $\wh{\w}_T$ and then derive its convergence rate. For any $\w\in\calW$, note that
\be
&&\sup_{\w\in\calW}\frac{1}{\sN}\left\| {\sf PTVMA}_\sT^*(\w) - R_{\sT+1}(\w) \right\| \notag\\
&=& \sup_{\w\in\calW} \left| \frac{1}{\sN\sT h}\sum_{t=1}^\sT \left\|\Y_t-\wh{\bmu}_t(\w)\right\|^2 K_{t,\sT} + \frac{\lambda}{\sN\sT h} \sum_{m=1}^\sM w^m\rho_m - \frac{1}{\sN} {\sf E}\|\Y_{\sT+1}-\widehat{\Y}_{\sT+1}(\w)\|^2 \right| \n\\
&\leq&\sup_{\w\in\calW}\bigg\{\bigg| \frac{1}{\sN\sT h}\sum_{t=1}^\sT \bigg[\left\|\Y_t-\wh{\bmu}_t(\w)\right\|^2K_{t,\sT} - \|\Y_t-\bmu_t^*(\w)\|^2K_{t,\sT}\bigg] \bigg| \n\\
&&  + \bigg| \frac{1}{\sN\sT h}\sum_{t=1}^\sT K_{t,\sT}\bigg[ \|\Y_t-\bmu_t^*(\w)\|^2 - {\sf E}\|\Y_t-\bmu_t^*(\w)\|^2 \bigg] \bigg|+ \n\\
&& + \bigg| \frac{1}{\sN \sT h}\sum_{t=1}^\sT K_{t,\sT}{\sf E}\|\Y_t-\bmu_t^*(\w)\|^2 - \frac{1}{\sN} {\sf E}\|\Y_{\sT+1}-\Y_{\sT+1}^*(\w)\|^2  \bigg| \n\\
&& + \bigg| \frac{1}{\sN}{\sf E}\|\Y_{\sT+1}-\Y_{\sT+1}^*(\w)\|^2 - \frac{1}{\sN} {\sf E}\|\Y_{\sT+1}-\widehat{\Y}_{\sT+1}(\w)\|^2 \bigg| +\left| \frac{\lambda}{\sN\sT h}\sum_{m=1}^\sM w^m\rho_m\right|\bigg\} \n\\
&=:& \Xi_1 + \Xi_2 + \Xi_3 + \Xi_4+\Xi_5 \label{eqA.32}.
\ee
As in the proofs of (\ref{eqA.11}), (\ref{eqA.13})--(\ref{eqA.15}) and (\ref{eqA.30}), we readily have that
\begin{equation}\label{eqA.33}
\Xi_i=o_P(1),\quad i=1,\ldots,5.
\end{equation}
It follows from (\ref{eqA.32}) and (\ref{eqA.33}) that
\[
\sup_{\w\in\calW}\frac{1}{\sN}\left| {\sf PTVMA}_\sT^*(\w) - R_{\sT+1}(\w) \right| = o_P(1).
\]
Since $R_{\sT+1}$ has a unique minimizer on $\calW$, the argmin continuity theorem implies that $\widehat{\w}_\sT\cp \w_\sT^*$, completing the proof of weak consistency property.

We next give the proof of (\ref{eq3.5}). Let
$$\wt{\boldsymbol\Delta}_\sT=\frac{1}{\sN}{\sf E}\left[\wt\Y_{\sT+1} \wt\Y_{\sT+1}^\top\right],\quad \wt\Y_{\sT+1} = (\wh\Y_{\sT+1}^{(1)},\ldots,\wh\Y_{\sT+1}^{(\sM)})^\top.$$
Note that, for any $\w\in\calW$,
\begin{eqnarray}
\frac{1}{\sN}R_{\sT+1}(\w) &=& \frac{1}{\sN}{\sf E}\left(\left[\Y_{\sT+1}-\widehat{\Y}_{\sT+1}(\w)\right]^\top\left[\Y_{\sT+1}-\widehat{\Y}_{\sT+1}(\w)\right]\right) -\frac{1}{\sN\sT h}\sum_{i=1}^\sN\sum_{s=1}^\sT\sigma_{i}^{2}(s/\sT)K_{s,\sT} \notag\\
&=& \frac{1}{\sN}{\sf E}\left[\Y_{\sT+1}^\top\Y_{\sT+1} - 2(\w^\top\wt\Y_{\sT+1})\Y_{\sT+1} + \w^\top\wt\Y_{\sT+1}\wt\Y_{\sT+1}^\top\w\right] - \frac{1}{\sN\sT h}\sum_{i=1}^\sN\sum_{s=1}^\sT\sigma_{i}^{2}(s/\sT)K_{s,\sT} \notag\\
&=& \w^\top \wt{\boldsymbol\Delta}_T \w - 2\frac{1}{\sN}\w^\top{\sf E}\big[\wt\Y_{\sT+1}\Y_{\sT+1}\big] + c_\diamond,
\end{eqnarray}
where $c_\diamond$ is a constant unrelated to $\w$. Write
\[
\wt{\boldsymbol\Delta}_\sT^\ast=\frac{1}{\sN}{\sf E}\left[\wt\Y_{\sT+1}^\ast \wt\Y_{\sT+1}^{\ast\top}\right],
\]
whose minimum eigenvalue is larger than $\underline\kappa_\ast>0$ by Assumption \ref{ass:5}(i). Using the argument in the proof of Theorem \ref{thm:3.2}, we readily have that
\[
\lambda_{\min}\left(\wt{\boldsymbol\Delta}_\sT\right)\geq \lambda_{\min}\left(\wt{\boldsymbol\Delta}_\sT^\ast\right)-o_P(1)\geq \underline\kappa_\ast/2,
\]
indicating that $R_{\sT+1}(\mathbf{w})$ is strongly convex. For any $\w\in\calW$, the strong convexity of $R_{\sT+1}(\w)$ leads to
\begin{equation}\label{eqA.35}
\frac{1}{\sN}\left[R_{\sT+1}(\w) - R_{\sT+1}(\w_\sT^*)\right] >\frac{\underline{\kappa}_\ast}{2}\left\|\w-\w_\sT^*\right\|^2.
\end{equation}
As $\widehat{\w}_\sT$ is the minimizer of ${\sf PTVMA}_\sT^*(\w)$, we have
\begin{eqnarray}
0 &\geq& {\sf PTVMA}_\sT^*(\widehat\w_\sT) - {\sf PTVMA}_\sT^*(\w_\sT^*)\notag  \\
&=& {\sf PTVMA}_\sT^*(\widehat\w_\sT) - R_{\sT+1}(\widehat\w_\sT) + R_{\sT+1}(\widehat\w_\sT) - R_{\sT+1}(\w_\sT^*) + R_{\sT+1}(\w_\sT^*) - {\sf PTVMA}_\sT^*(\w_\sT^*).\notag
\end{eqnarray}
This, together with (\ref{eqA.35}) with $\widehat\w_T$ replacing $\w$, leads to
\be
\left\|\widehat\w_\sT-\w_\sT^*\right\|^2&\leq&\frac{2}{\underline\kappa_\ast}\cdot\frac{1}{\sN}\left[R_{\sT+1}(\widehat\w_\sT) - R_{\sT+1}(\w_\sT^*)\right]\n\\
&\leq&\frac{2}{\underline\kappa_\ast}\cdot\frac{1}{\sN}\left\{\left[{\sf PTVMA}_\sT^*(\w_\sT^*) - R_{\sT+1}(\w_\sT^*)\right] - \left[{\sf PTVMA}_T^*(\widehat\w_\sT) - R_{\sT+1}(\widehat\w_\sT)\right]\right\}\n \\
&\leq& \frac{4}{\underline\kappa_\ast}\cdot\frac{1}{\sN} \sup_{\w\in\calW}\left| {\sf PTVMA}_\sT^*(\w) - R_{\sT+1}(\w) \right|\n \\
&=&O_P\left(\sM \bar{\rho}\zeta+\lambda\bar\rho(\sN \sT h)^{-1} + (\sN\sT h)^{-1/2}\sM^2 +\bar{\rho} h \right),\label{eqA.36}
\ee
where the last equality is due to the rates obtained in the proofs of (\ref{eqA.27})--(\ref{eqA.30}). We complete the proof of (\ref{eq3.5}) using (\ref{eqA.36}).\end{proof}

\begin{proof}[\bf Proof of Theorem \ref{thm:3.6}]

The proof is analogous to the proof of Theorem \ref{thm:3.3} with minor amendments. Details are omitted to save the space.
\end{proof}

\begin{proof}[\bf Proof of Proposition \ref{prop:4.1}]

Denote the empirical $p$-values using the estimated residuals:
\begin{eqnarray}
\wh{p}_{i, \sT+1}&=&\frac{1}{\sT h+1}\sum_{t=\sT- \sT h}^{\sT}I\{|\wh{\epsilon}_{i,t}|>|\wh{\epsilon}_{i,\sT+1})|\},\notag\\
 \wh{q}_{i,\sT+1}&=&\frac{1}{\sT h+1}\sum_{t=\sT- \sT h}^{\sT}I\{|\wh{u}_{i,t}|>|\wh{u}_{i,\sT+1}|\},\notag
\end{eqnarray}
 and denote the empirical $p$-values using true errors:
\begin{eqnarray}
\wt{p}_{i,\sT+1}&=&\frac{1}{\sT h+1}\sum_{t=\sT- \sT h}^{\sT}I\{|{\epsilon}_{i,t}|>|{\epsilon}_{i,\sT+1}|\},\notag\\
\wt{q}_{i,\sT+1}&=&\frac{1}{\sT h+1}\sum_{t=\sT- \sT h}^{\sT}I\{|{u}_{i,t}|>|{u}_{i,\sT+1}|\},\notag
\end{eqnarray}
where $\wh{u}_{i,t}=\sigma_{i}^{-1}(\tau_t)\wh{\epsilon}_{i,t}$, and $I\{\cdot\}$ is an indicator function.

Note that $\wh{p}_{i,\sT+1}>\alpha$ is equivalent to $|\wh{\epsilon}_{i,\sT+1}|< |\wh{\epsilon}_i^{\alpha}|$, and thus
\begin{align}
Y_{i, \sT+1}&\in \left[Y_{i, \sT+1}-\wh{\epsilon}_{i, \sT+1}-\wh{\epsilon}_i^{\alpha},\ Y_{i,\sT+1}-\wh{\epsilon}_{i, \sT+1}+\wh{\epsilon}_i^{\alpha}\right]\notag\\
&=\left[\wh{Y}_{i,\sT+1}(\wh{\w}_\sT)-\wh{\epsilon}_i^{\alpha},\ \wh{Y}_{i,\sT+1}(\wh{\w}_\sT)+\wh{\epsilon}^{\alpha}\right]=C_{i}^{\alpha}(\Z_{\sT})\notag
\end{align}
by the definition of $\wh{Y}_{i,\sT+1}(\wh{\w}_\sT)$ in (\ref{eq2.7}). Then we readily have that
\begin{equation}\label{eqA.37}
{\sf P}\left(\wh{p}_{i,\sT+1}>\alpha\right)={\sf P}\left(Y_{i,\sT+1}\in C_{i}^{\alpha}(\Z_{\sT})\right).
\end{equation}

Write the empirical distribution functions of $\{u_{i,t}\}_{t=\sT-\sT h}^\sT$ and $\{\wh{u}_{i,t}\}_{t=\sT-\sT h}^\sT$ as
\[
\wt{\sf F}(x)=\frac{1}{\sT h+1}\sum_{t=\sT- \sT h}^{\sT} I\{|u_{i,t}|\leq x\}
\]
and
\[
\wh{\sf F}(x)=\frac{1}{\sT h+1}\sum_{t=\sT- \sT h}^{\sT} I\{|\wh{u}_{i,t}|\leq x\},
\]
respectively. Then, we readily have that
\[
\wh{\sf F}(|\wt{u}_{i,\sT+1}|)=1-\wh{q}_{i,\sT+1}\ \ \text{and}\ \ \wt{\sf F}(|u_{i,\sc \sT+1}|)=1-\wt{q}_{i,\sT+1}.
\]

With Assumption \ref{ass:6}(ii) and Dvoretzky–Kiefer–Wolfowitz inequality \citep[e.g.,][]{D08}, we have
\begin{equation}\label{eqA.38}
{\sf P}\left(\sup_x\left|\wt{\sf F}(x)-{\sf F}(x)\right| >z\right) \leq 2e^{-2\sT h z^2},
\end{equation}
for any $z>0$. Letting $a_\sT=(\sT h)^{-1/3}$, it follows from (\ref{eqA.38}) that
\be
{\sf E}\left(\sup_x|\wt{\sf F}(x)-{\sf F}(x)|\right)&=&\int_{0}^{\infty}{\sf P}\left( \sup_x|\wt{\sf F}(x)-{\sf F}(x)| \geq z \right) dz \n\\
&\leq& a_\sT+ \int_{a_\sT}^{\infty} 2e^{-2\sT h z^2} dz
\n\\
&\leq& a_\sT + O\left(e^{-2\sT h a_\sT^2}\right)\n\\
&\leq& 2a_\sT. \label{eqA.39}
\ee
Define
\[
{\mathscr S}=\left\{t=\sT-\sT h,\ldots,\sT,  \big\vert |\wh{u}_{i,t}|-|u_{i,t}|\big\vert\geq \delta_{i,\sT}^{1/2}\right\},\quad \delta_{i,\sT}^2=\frac{1}{\sT h+1}\sum_{t=\sT-\sT h}^{\sT}\left(|\wh{u}_{i,t}|-|u_{i,t}|\right)^2.
\]
For any random variable $ X $, we have
\be
\left|\wh{\sf F}(X)-\wt{\sf F}(X)\right|&\leq& \frac{1}{\sT h+1}\sum_{t=\sT-\sT h}^{\sT}\left|I\left\{|\wh{u}_{i,t}|\leq X\right\}-I\left\{|{u}_{i,t}|\leq X\right\}\right|\n\\
&\leq& \frac{1}{\sT h+1} \left(\sum_{t\in{\mathscr S}}+\sum_{t\notin{\mathscr S}}\right) I\left\{\big||{u}_{i,t}|-X\big| \leq \big| |\wh{u}_{i,t}| - |{u}_{i,t}| \big| \right\} \n\\
&\leq& \frac{1}{\sT h+1} |{\mathscr S}| + \frac{1}{\sT h+1}\sum_{t=\sT-\sT h}^{\sT} I\left\{\Big||{u}_{i,t}|-X\Big|\leq \delta_{i,\sT}^{1/2}\right\} \n\\
&\leq& \frac{1}{\sT h+1} |{\mathscr S}| + \wt{\sf F} (X+\delta_{i,\sT}^{1/2}) - \wt{\sf F} (X-\delta_{i,\sT}^{1/2}), \label{eqA.40}
\ee
where the second inequality is due to $\big|I\left\{a\leq x\right\}-I\left\{b\leq x\right\}\big|\leq I\left\{|b-x|\leq |a-b|\right\}$, and $|{\mathscr S}|$ denotes the cardinality of ${\mathscr S}$.  Furthermore, by the definition of $\delta_{i,\sT}^2$, we may show that
\be
\frac{1}{\sT h+1} |{\mathscr S}|&=&\frac{1}{\sT h+1}\sum_{t=\sT-\sT h}^{\sT}I\left\{ \big||\wt{u}_{i,t}|-|u_{i,t}|\big| \geq \delta_{i,\sT}^{1/2} \right\}\n\\
&\leq&\delta_{i,\sT}^{-1/2} \cdot \frac{1}{\sT h+1}\sum_{t=\sT-\sT h}^{\sT} \big| |\wt{u}_{i,t}|-|u_{i,t}| \big| \n\\
&\leq&\delta_{i,\sT}^{-1/2} \delta_{i,\sT}=\delta_{i,\sT}^{1/2}, \label{eqA.41}
\ee
which, together with (\ref{eqA.40}), leads to
\begin{equation}\label{eqA.42}
\left|\wh{\sf F}(X)-\wt{\sf F}(X)\right|\leq \delta_{i,\sT}^{1/2}+\wt{\sf F} (X+\delta_{i,\sT}^{1/2}) - \wt{\sf F} (X-\delta_{i,\sT}^{1/2}).
\end{equation}
With (\ref{eqA.39}) and (\ref{eqA.42}), using Assumption \ref{ass:6}(ii) and the basic inequality $\big|I\left\{a\leq x\right\}-I\left\{b\leq x\right\}\big|\leq I\left\{|b-x|\leq |a-b|\right\}$ again, we may show that
\be
&&\left|{\sf P}\left(\wh{q}_{i,\sT+1}\leq \alpha\right)-\alpha\right|\n\\
&=&\left|{\sf P}\left(\wh{\sf F}(|\wh{u}_{i,\sT+1}|)\geq 1-\alpha\right)-{\sf P}\left({\sf F}(|u_{i,\sT+1}|)\geq 1-\alpha\right)\right|\n\\
&\leq&\mE\left|I\left\{1-\wh{\sf F}(|\wh{u}_{i,\sT+1})\leq \alpha\right\}-I\left\{1-{\sf F}(|u_{i,\sT+1}|)\leq \alpha\right\}\right|\n\\
&\leq&{\sf P}\left(\big|{\sf F}(|u_{i,\sT+1}|)-(1-\alpha)\big|\leq\big|\wh{\sf F}(|\wh{u}_{i,\sT+1}|)-{\sf F}(|u_{i,\sT+1}|)\big|\right)\n\\
&\leq& 2 \mE \left|\wh{\sf F}(|\wh{u}_{i,\sT+1})-{\sf F}(|u_{i,\sT+1}|)\right| \n\\
&\leq& 2 \mE\left( \Big|\wh{\sf F}(|\wh{u}_{i,\sT+1}|) - \wt{\sf F}(|\wh{u}_{i,\sT+1}|)\Big| + \Big|\wt{\sf F}(|\wh{u}_{i,\sT+1}|) - {\sf F}(|\wh{u}_{i,\sT+1})\Big| +  \Big|{\sf F}(|\wh{u}_{i,\sT+1}|) - {\sf F}(|u_{i,\sT+1}|)\Big| \right) \n\\
&\leq& 2 \mE \left(\delta_{i,\sT}^{1/2}\right) + 2 \mE\left( \Big| \wt{\sf F}(|\wh{u}_{i,\sT+1}|+\delta_{i,\sT}^{1/2}) - {\sf F}(|\wh{u}_{i,\sT+1}| + \delta_{i,\sT}^{1/2}) \Big|\right) +\n\\
&& 2 \mE\left( \Big| \wt{\sf F}(|\wh{u}_{i,\sT+1}|-\delta_{i,\sT}^{1/2}) - {\sf F}(|\wh{u}_{i,\sT+1}| - \delta_{i,\sT}^{1/2}) \Big|\right)+ 2 \mE\left( \Big|{\sf F}(|\wh{u}_{i,\sT+1}|+\delta_{i,\sT}^{1/2}) - {\sf F}(|\wh{u}_{i,\sT+1}|-\delta_{i,\sT}^{1/2})\Big |\right)+ \n\\
&&2 \mE \left(  \Big|\wt{\sf F}(|\wh{u}_{i,\sT+1}|) - {\sf F}(|\wh{u}_{i,\sT+1})\Big|\right)+ 2 \mE \left(\Big|{\sf F}(|\wh{u}_{i,\sT+1}|) - {\sf F}(|u_{i,\sT+1}|)\Big|\right)\n\\
&\leq& (2+4M_{\sf F}) \mE \left(\delta_{i,\sT}^{1/2}\right) + 6 \mE \left(\sup_x|\wt{\sf F}(x)-{\sf F}(x)|\right) +  2M_{\sf F} \mE \Big| |\wh{u}_{i,\sT+1}| - |u_{i,\sT+1}| \Big| \n\\
&\leq& (2+4M_{\sf F}) \mE \left(\delta_{i,\sT}^{1/2}\right) + 12 a_\sT+  2M_{\sf F}  \mE(\psi_{i,\sT+1}),\label{eqA.43}
\ee
where $M_{\sf F}$ is defined in Assumption \ref{ass:6}(ii).

On the other hand, note that $\sigma_i(\tau_t)-\sigma_i(\tau_{\sT+1}) = O(h)$, $t=\sT-\sT h, \ldots,\sT$, by the smoothness condition in Assumption \ref{ass:6}(i), where $\tau_t=t/\sT$. Define
\[
{\mathscr S}_1=\left\{t=\sT-\sT h,\ldots,\sT, \ \ \big\vert |\wh{u}_{i,t}|-|u_{i,t}|\big\vert\leq \left[{\sf E}\left(\delta_{i,\sT}^{2}\right)\right]^{1/4}\right\}
\]
and
\[
{\mathscr S}_2=\left\{\big\vert |\wh{u}_{i,\sT+1}|-|u_{i,\sT+1}|\big\vert\leq \left[{\sf E}\left(\psi_{i,\sT+1}\right)\right]^{1/2}\right\},\quad {\mathscr S}_3=\left\{t=\sT-\sT h,\ldots,\sT,\ \  |u_{i,t}|\leq h^{-1/5} \right\}.
\]
Then, by Assumption \ref{ass:6}(i)(ii), we may show that
\be
&&\left|\wh{p}_{i,\sT+1}-\wh{q}_{i,\sT+1}\right|\n\\ &=& \left|\frac{1}{\sT h+1}\sum_{t=\sT-\sT h}^{\sT}I\left\{|\wh{\epsilon}_{i,t}|>|\wh{\epsilon}_{i,\sT+1}|\right\}-\frac{1}{\sT h+1}\sum_{t=\sT-\sT h}^{\sT}I\left\{|\wh{u}_{i,t}|>|\wh{u}_{i,\sT+1}|\right\}\right| \n\\
&=&\left|\frac{1}{\sT h+1}\sum_{t=\sT-\sT h}^{\sT}I\left\{[\sigma_i(\tau_t)-\sigma_i(\tau_{\sT+1})]|\wh{u}_{i,t}|+\sigma_i(\tau_{\sT+1})\left(|\wh{u}_{i,t}|-|\wh{u}_{i,\sT+1}|\right)>0\right\}\right.-\n\\
&&\left.\frac{1}{\sT h+1}\sum_{t=\sT-\sT h}^{\sT}I\{\sigma_i(\tau_{\sT+1})\left(|\wh{u}_{i,t}|-|\wh{u}_{i,\sT+1}|\right)>0\}\right|\n\\
&\leq&\frac{1}{\sT h+1}\sum_{t=\sT-\sT h}^{\sT}\Bigg| I\left\{[\sigma_i(\tau_t)-\sigma_i(\tau_{\sT+1})]|\wh{u}_{i,t}|+\sigma_i(\tau_{\sT+1})\left(|\wh{u}_{i,t}|-|\wh{u}_{i,\sT+1}|\right)>0\right\}-\notag\\
&&I\{\sigma_i(\tau_{\sT+1})\left(|\wh{u}_{i,t}|-|\wh{u}_{i,\sT+1}|\right)>0\}\Bigg|\n\\
&\leq & \frac{1}{\sT h+1}\sum_{t=\sT-\sT h}^{\sT}I\{\left|\sigma_i(\tau_t)-\sigma_i(\tau_{\sT+1})\right|\left|\wh{u}_{i,t}\right|\geq \sigma_i(\tau_{\sT+1})\big||\wh{u}_{i,t}|-|\wh{u}_{i,\sT+1}|\big|\}\n\\
&\leq &  \frac{1}{\sT h+1}\sum_{t=\sT-\sT h}^{\sT}I\big\{ \underline\sigma\big| | u_{i,t}|-|u_{i,\sT+1}|\big|\leq \left|\sigma_i(\tau_t)-\sigma_i(\tau_{\sT+1})\right|\left|u_{i,t}\right|+\n\\
&&3\overline\sigma\left|\wh{u}_{i,t}-u_{i,t}\right|+\overline\sigma\big |\wh{u}_{i,\sT+1}-u_{i,\sT+1}\big|\big\}\n\\
&\leq& \frac{1}{\sT h+1}\left(\sum_{t\in{\mathscr S}_1\cap {\mathscr S}_3}+\sum_{t\notin{\mathscr S}_1^c}+\sum_{t\notin{\mathscr S}_3^c}\right)I\big\{ \underline\sigma\big| | u_{i,t}|-|u_{i,\sT+1}|\big|\leq \left|\sigma_i(\tau_t)-\sigma_i(\tau_{\sT+1})\right|\left|u_{i,t}\right|+\n\\
&&3\overline\sigma\left|\wh{u}_{i,t}-u_{i,t}\right|+\overline\sigma\left|\wh{u}_{i,\sT+1}-u_{i,\sT+1}\right|,\ \ {\mathscr S}_2\big\}+ I\big\{{\mathscr S}_2^c\big\}\n\\
&\leq& \frac{1}{\sT h+1}\sum_{t=\sT-\sT h}^{\sT}I\big\{ \underline\sigma\big| | u_{i,t}|-|u_{i,\sT+1}|\big|\leq \left|\sigma_i(\tau_t)-\sigma_i(\tau_{\sT+1})\right|h^{-1/5}+3\overline\sigma\left[{\sf E}\left(\delta_{i,\sT}^{2}\right)\right]^{1/4}+ \n\\
&&\overline\sigma\left[{\sf E}\left(\psi_{i,\sT+1}\right)\right]^{1/2}\big\}+\frac{1}{\sT h+1}\left(|{\mathscr S}_1^c|+|{\mathscr S}_3^c|\right)+I\big\{{\mathscr S}_2^c\big\},\label{eqA.44}
\ee
where ${\mathscr S}_k^c$ denotes the complement of ${\mathscr S}_k$, $k=1,2,3$. As in the proof of (\ref{eqA.41}), with Assumption \ref{ass:6}(ii) and the Markov inequality, there exists a positive constant $c_1$ such that
\be
&&{\sf E} \left(|{\mathscr S}_1^c|\right)\leq c_1(\sT h+1) \left[{\sf E}\left(\delta_{i,\sT}^{2}\right)\right]^{1/4},\notag\\
&&{\sf E}\left(|{\mathscr S}_3^c|\right)\leq c_1(\sT h+1)h^{4/5},\notag\\
&&{\sf P}\left({\mathscr S}_2^c\right)\leq c_1\left[{\sf E}\left(\psi_{i,\sT+1}\right)\right]^{1/2},\notag
\ee
which, together with (\ref{eqA.44}), lead to
\begin{equation}\label{eqA.45}
\left| {\sf P}\left(\wh{p}_{i,\sT+1}\leq\alpha\right)- {\sf P}\left(\wh{q}_{i,\sT+1}\leq\alpha\right)\right|\leq c_2\left(\left[{\sf E}\left(\delta_{i,\sT}^{2}\right)\right]^{1/4}+ \left[{\sf E}\left(\psi_{i,\sT+1}\right)\right]^{1/2} + h^{4/5}\right),
\end{equation}
where $c_2$ is a positive constant. With (\ref{eqA.43}) and (\ref{eqA.45}), as ${\sf E}\left(\psi_{i,\sT+1}\right)<1$, we have
\begin{eqnarray}
&&\left|{\sf P}\left(\wh{p}_{i,\sT+1}>\alpha\right)-(1-\alpha)\right|\notag\\
&=&\left|{\sf P}\left(\wh{p}_{i,\sT+1}\leq\alpha\right)-\alpha\right|\notag\\
&\leq& \left| {\sf P}\left(\wh{p}_{i,\sT+1}\leq\alpha\right)- {\sf P}\left(\wh{q}_{i,\sT+1}\leq\alpha\right)\right|+\left|{\sf P}\left(\wh{q}_{i,\sT+1}\leq \alpha\right)-\alpha\right|\notag\\
&\leq& c_\dag\left(\left[{\sf E}\left(\delta_{i,\sT}^{2}\right)\right]^{1/4}+\left[{\sf E}\left(\psi_{i,\sT+1}\right)\right]^{1/2}+h^{4/5}+(\sT h)^{-1/3}\right),\label{eqA.46}
\end{eqnarray}
where $c_\dag$ is a sufficiently large positive constant. By virtue of (\ref{eqA.37}) and (\ref{eqA.46}), we complete the proof of (\ref{eq4.1}) in Proposition \ref{prop:4.1}.

We next turn to the proof of (\ref{eq4.2}). Define
\[
\Delta_{i,\sT} = \frac{1}{\sT h+1}\sum_{t=\sT-\sT h}^{\sT}\big( |\wh{\epsilon}_{i,t}|-|\epsilon_{i,t}| \big)^2\quad \text{and}\quad \Psi_{i,\sT+1} =  \big| |\wh{\epsilon}_{i,\sT+1}|-|\epsilon_{i,\sT+1}| \big| .
\]
Observe that
\[
\Delta_{i,\sT} =\frac{1}{\sT h+1}\sum_{t=\sT-\sT h}^{\sT} \sigma_i^2 (\tau_t) \big( |\wh{u}_{i,t}|-|u_{i,t}| \big)^2   %
\geq \underline{\sigma}^2 \delta_{i,\sT}^2,
\]
by Assumption \ref{ass:6}(i), and similarly, $\Psi_{i,\sT+1}\geq \underline\sigma \psi_{i,\sT+1}$. In view of (\ref{eq4.1}), to prove (\ref{eq4.2}), it suffices to show that
\begin{equation}\label{eqA.47}
{\sf E}(\Delta_{i,\sT}) = o(1)\ \ \text{and}\ \ {\sf E}\left(\Psi_{i,\sT+1}\right) = o(1).
\end{equation}

Note that
\begin{eqnarray}
\Delta_{i,\sT} &\leq& \frac{1}{\sT h+1}\sum_{t=\sT-\sT h}^{\sT}  \big( \wh{\epsilon}_{i,t} - \epsilon_{i,t} \big)^2\notag   \\
&=& \frac{1}{\sT h+1}\sum_{t=\sT-\sT h}^{\sT} \left( \sum_{m=1}^{\sM} \wh{w}_{\sT}^{m} \Z_{i,t-1}^{(m)\top} \wh{\btheta}_t^{(m)} - \Z_{i,t-1}^{\top} \btheta_{t} \right)^2 \notag  \\
&\leq& \frac{2}{\sT h+1}\sum_{t=\sT-\sT h}^{\sT} \left[ \sum_{m \in \calD} \wh{w}_{\sT}^{m} \left( \Z_{i,t-1}^{(m)\top} \wh{\btheta}_t^{(m)} - \Z_{i,t-1}^{\top}\btheta_{t} \right) \right]^2 + \notag\\
&& \frac{2}{\sT h+1}\sum_{t=\sT-\sT h}^{\sT} \left[ \sum_{m \notin \calD} \wh{w}_{\sT}^{m} \left( \Z_{i,t-1}^{(m)\top} \wh{\btheta}_t^{(m)} - \Z_{i,t-1}^{\top}\btheta_{t} \right) \right]^2 \notag  \\
&=:& \Omega_{1} + \Omega_{2}.\notag
\end{eqnarray}
For $\Omega_{1}$, noting that $\bPi^{\m\top} \btheta_{t,*}^{(m)}=\btheta_{t}$ for $m\in\calD$ which is not empty, we have
\begin{eqnarray}
\Omega_1 &=& \frac{2}{\sT h+1}\sum_{t=\sT-\sT h}^{\sT} \left[ \sum_{m \in \calD} \wh{w}_{\sT}^{m} \left( \Z_{i,t-1}^{(m)\top} \wh{\btheta}_t^{(m)} - \Z_{i,t-1}^{(m)\top} \btheta_{t,*}^{(m)} + \Z_{i,t-1}^{(m)\top} \btheta_{t,*}^{(m)} - \Z_{i,t-1}^{\top}\btheta_{t} \right) \right]^2 \notag  \\
&\leq& \frac{2}{\sT h+1}\sum_{t=\sT-\sT h}^{\sT}  \left[\sum_{m \in \calD} \wh{w}_{\sT}^{m} \left\| \Z_{i,t-1}^{(m)} \right\| \cdot \left\| \wh{\btheta}_t^{(m)} - \btheta_{t,*}^{(m)} \right\| \right]^2   \notag  \\
&=&  O_P\left(\bar{\rho}\zeta^2\right)\cdot  \frac{1}{\sT h+1}\sum_{t=\sT-\sT h}^{\sT}  \left( \sum_{m \in \calD} (\wh{w}_{\sT}^{m})^2 \sum_{m\in\calD} \left\| \Z_{i,t-1}^{(m)} \right\|^2\right) \notag \\
&=&  O_P\left(\bar{\rho}\zeta^2\right) \cdot \frac{1}{\sT h+1}\sum_{t=\sT-\sT h}^{\sT}  \sum_{m\in\calD} \left\| \Z_{i,t-1}^{(m)} \right\|^2 \notag\\
&=& O_P\left(\sM\bar{\rho}\zeta^2\right)=o_P(1). \label{eqA.48}
\end{eqnarray}
For $\Omega_{2}$, using the argument in the proofs of Theorems \ref{thm:3.3} and \ref{thm:3.5} as well as Assumption \ref{ass:6}(iii), we can similarly prove that
\begin{eqnarray}
\Omega_2&\leq& O_P\left(\bar{\rho}\zeta^2\right) \cdot \frac{1}{\sT h+1}\sum_{t=\sT-\sT h}^{\sT}  \left(\sum_{m \notin \calD} \wh{w}_{\sT}^{m} \left\|\Z_{i,t-1}^{(m)}\right\|\right)^2+ O_P(\bar{\rho}) \cdot \frac{1}{\sT h+1}\sum_{t=\sT-\sT h}^{\sT} \left(\sum_{m \notin \calD} \wh{w}_{\sT}^{m} \right)^2\notag\\
&=& o_P\left(\sM\bar{\rho}\zeta^2\right) +O_P\left(\bar{\rho} \widetilde\xi_{\sT+1}^{\ast-1} \left[ \bar\rho \sM\sN \zeta + \bar\rho\lambda(\sT h)^{-1} + \sN \bar{\rho} h + \sM^2\sN^{1/2}(\sT h)^{-1/2} \right] \right) \notag\\
&=& o_P(1). \label{eqA.49}
\end{eqnarray}
By virtue of (\ref{eqA.48}) and (\ref{eqA.49}), we complete the proof of the first assertion in (\ref{eqA.47}).

Finally, we note that
\begin{eqnarray}
\Psi_{i,\sT+1} &\leq& \left| \sum_{m=1}^\sM \wh{w}_{\sT}^{m} \left( \Z_{i,\sT}^{(m)\top} \wh{\btheta}_\sT^{(m)} - \Z_{i,\sT}^{\top} \btheta_{\sT} + \Z_{i,\sT}^{\top} \btheta_{\sT} - \Z_{i,\sT}^{\top} \btheta_{\sT+1} \right) \right| \notag  \\
&\leq& \left| \sum_{m=1}^\sM \wh{w}_{\sT}^{m} \left( \Z_{i,\sT}^{(m)\top} \wh{\btheta}_\sT^{(m)} - \Z_{i,\sT}^{\top} \btheta_{\sT} \right) \right| + \sum_{m=1}^\sM \wh{w}_{\sT}^{m} \|\Z_{i,\sT}\| \cdot \|\btheta_{\sT+1} - \btheta_{\sT}\|  \notag\\
&=:& \Omega_3+\Omega_4.\notag
\end{eqnarray}
Similarly the proofs of (\ref{eqA.48}) and (\ref{eqA.49}), we may show that
\begin{equation}\label{eqA.50}
{\sf E}(\Omega_3)=o(1).
\end{equation}
By the smoothness condition in Assumption \ref{ass:6}(iii), we also have
\begin{equation}\label{eqA.51}
{\sf E}(\Omega_4)=o(1).
\end{equation}
A combination of (\ref{eqA.50}) and (\ref{eqA.51}) leads to the second assertion of (\ref{eqA.47}).
\end{proof}


\section{Construction of networks in the empirical study}\label{app:B}
\renewcommand\theequation{B.\arabic{equation}}

\noindent \textbf{Global production networks.} We use the origin of value added embodied in final demand to construct the global production network. The annual Trade in Value Added (TiVA) measure is treated as constant within each year and assigned to all months in that year. For the in-sample analysis, the adjacency matrix of the production-network layer is constructed by
\begin{align*}
\omega_{\text{In},ij}^{(1)} = I \left\{\frac{1}{\sT} \sum_{t=1}^{\sT} \text{Production}_{ij,t} > c_{\omega,1} \right\},
\end{align*}
where $ \text{Production}_{ij,t} $ denotes the estimated final demand of economy $ i $ for final goods and services sourced from economy $ j $ in month $ t $, and the threshold $ c_{\omega,1} $ is set as the 60th percentile of $ \sum_{t=1}^{\sT}\text{Production}_{ij,t} / \sT $.

\smallskip

Given the 2–3 year publication lag of the TiVA data, we use lagged input-output information to construct two production-network layers for the out-of-sample prediction. For month $ t $ in the test set, two adjacency matrices of the production network layers are constructed by
\begin{align*}
\omega_{\text{Out},ij,t}^{(1)} = I \left\{  \text{Production}_{ij,t-3} > c_{\omega,2} \right\} \text{ and }   %
\omega_{\text{Out},ij,t}^{(2)} = I \left\{ \frac{1}{L_\text{Roll}} \sum_{m=t-3-L_\text{Roll}+1}^{t-3} \text{Production}_{ij,m} > c_{\omega,3} \right\},
\end{align*}
where $ L_\text{Roll} $ denotes the length of rolling window and $ c_{\omega,2} $ and $ c_{\omega,3} $ are the thresholding parameters which are set as the 70th percentile of $ \text{Production}_{ij,t-3} $ and $ \sum_{m=t-3-L_\text{Roll}+1}^{t-3} \text{Production}_{ij,m} / L_\text{Roll} $, respectively. Here, $ \omega_{\text{Out},ij,t}^{(1)} $ and $ \omega_{\text{Out},ij,t}^{(2)} $ reflect the latest information and the average level on the production network, respectively.

\medskip

\noindent \textbf{Global equity networks.} We use the MSCI database and \cite{DY14}'s VAR decomposition framework to quantify the cross-country financial connectedness. For the in-sample analysis, the adjacency matrix of the equity layer is constructed by
\begin{align*}
\omega_{\text{In},ij}^{(2)} = I \left\{ \text{Equity}_{ij} > c_{\omega,4} \right\},
\end{align*}
where $ \text{Equity}_{ij} $ denotes the $ (i, j) $-th entry of \cite{DY14}'s spillover matrix estimated using the full sample, and $ c_{\omega,4} $ is the 60th percentile of $ \text{Equity}_{ij} $.

\smallskip

For the out-of-sample analysis, we consider three equity-network layers to capture short-, medium-, and long-term connectedness in the global market. For month $ t $ in the test set, the three adjacency matrices of the equity network layers are constructed by
\be
\omega_{\text{Out},ij,t}^{(3)} &=& I \left\{ \text{Equity}_{\text{short:}ij,t} > c_{\omega,5} \right\}, \notag\\   %
\omega_{\text{Out},ij,t}^{(4)} &=& I \left\{ \text{Equity}_{\text{medium:}ij,t} > c_{\omega,6} \right\}, \notag\\   %
\omega_{\text{Out},ij,t}^{(5)} &=& I \left\{ \text{Equity}_{\text{long:}ij,t} > c_{\omega,7} \right\}, \notag  %
\ee
where $ \text{Equity}_{\text{short:}ij,t} $, $ \text{Equity}_{\text{medium:}ij,t} $, $ \text{Equity}_{\text{long:}ij,t} $ denote the $ (i, j) $-th elements of \cite{DY14}'s spillover matrices generated by MSCI indices over windows $ [t-2, t] $, $ [t-11, t] $, and the full training sample, respectively. The thresholds $c_{\omega,5}$, $c_{\omega,6} $ and $c_{\omega,7} $ are set as the 70th percentile of $ \text{Equity}_{\text{short:}ij,t} $, the 30th percentile of $ \text{Equity}_{\text{medium:}ij,t} $ and $ \text{Equity}_{\text{long:}ij,t} $, respectively.

\medskip

\noindent \textbf{Trade and policy networks.} We use monthly bilateral export flows from the IMF's Direction of Trade Statistics (DOTS) to construct trade networks. For the in-sample analysis, the adjacency matrix of the trade-network layer is constructed by
\begin{align*}
\omega_{\text{In},ij}^{(3)} = I \left\{ \frac{1}{\sT} \sum_{t=1}^{\sT} \text{Trade}_{ij,t} > c_{\omega,8} \right\},
\end{align*}
where $ \text{Trade}_{ij,t} $ denotes the FOB export value from economy $ j $ to $ i $ in month $ t $, and $ c_{\omega,8} $ is set to be the 60th percentile of $ \sum_{t=1}^{\sT} \text{Trade}_{ij,t} / \sT $. For the out-of-sample analysis, we consider three trade-network layers to accommodate the 2--3 month publication lag and mitigate seasonal volatility in shipment values. For the month $ t $ in the test set, the three adjacency matrices of the trade network layers are constructed by
\be
\omega_{\text{Out},ij,t}^{(6)} &=& I \left\{ \text{Trade}_{ij,t-12} > c_{\omega,9} \right\}, \notag\\
\omega_{\text{Out},ij,t}^{(7)} &=& I \left\{ \frac{1}{3} \sum_{k=1}^{3} \text{Trade}_{ij,t-12k} > c_{\omega,10} \right\},\notag\\
\omega_{\text{Out},ij,t}^{(8)} &=& I \left\{ \frac{1}{L_\text{Roll}-3} \sum_{m=t-L_\text{Roll}+1}^{t-3} \text{Trade}_{ij,m} > c_{\omega,11} \right\},\notag
\ee
where $c_{\omega,9}$, $c_{\omega,10} $ and $c_{\omega,11} $ are set as the 60th percentile of $ \text{Trade}_{ij,t-12} $, $ \sum_{k=1}^{3} \text{Trade}_{ij,t-12k} / 3 $ and $ \sum_{m=t-L_\text{Roll}+1}^{t-3} $ $ \text{Trade}_{ij,m} / (L_\text{Roll}-3) $, respectively.

\smallskip

We use ideal-point distances to measure bilateral political connectedness. For year $ t $, the political distance between countries $ i $ and $ j $ is defined as the inverse of their absolute ideal-point distance, i.e.,
\begin{align*}
\text{Policy}_{ij,t} = \left| \text{IdealPoint}_{i,t} - \text{IdealPoint}_{j,t} \right|^{-1}.
\end{align*}
Then, the annual ideal-point distances are assigned to all months within the corresponding year. For the in-sample analysis, the adjacency matrix of the policy-network layer is constructed by
\begin{align*}
\omega_{\text{In},ij}^{(4)} = I \left\{ \frac{1}{\sT} \sum_{t=1}^{\sT} \text{Policy}_{ij,t} > c_{\omega,12} \right\},
\end{align*}
where $ c_{\omega,12} $ is set to be the 60th percentile of $ \sum_{t=1}^{\sT} \text{Trade}_{ij,t} / \sT $. For the out-of-sample analysis, we consider two policy-network layers to accommodate one-year publication lag in UN voting data. For the month $ t $ in the test set, the two adjacency matrices of the policy network layers are constructed by
\begin{align*}
\omega_{\text{Out},ij,t}^{(9)} = I\left\{ \text{Policy}_{ij,t-12} > c_{\omega,13} \right\}, \;   %
	\omega_{\text{Out},ij,t}^{(10)} = I\left\{ \frac{1}{L_\text{Roll} - 12} \sum_{m=t-L_\text{Roll}+1}^{t-12} \text{Policy}_{ij,m} > c_{\omega,14} \right\},
\end{align*}
where thresholding parameters $c_{\omega,13}$ and $c_{\omega,14}$ are set as the 60th percentile of $ \text{Policy}_{ij,t-12} $ and $ \sum_{m=t-L_\text{Roll}+1}^{t-12} \text{Policy}_{ij,m} / (L_\text{Roll} - 12) $, respectively.

\bigskip


\end{document}